%% file: paper.tex
\documentclass{article} 
\usepackage{nips13submit_e,times}
\usepackage{hyperref}
\usepackage{url}
\usepackage{amsmath}
\usepackage{amssymb}
\usepackage{amsthm}

\newtheorem{definition}{Definition}
\newtheorem{lemma}{Lemma}
\newtheorem{assumption}{Assumption}
\newtheorem{prop}{Proposition}
\newtheorem{cor}{Corollary}

\newcommand{\mS}{{\cal M}}
\newcommand{\setR}{{\cal R}}

\newcommand{\setC}{{\cal C}}
\newcommand{\setCp}{{\cal C^+}}

\newcommand{\C}{{\left (I_C,I_C \right ) }}

\newcommand{\sCNash}{{\cal C}^*}

\newcommand{\commStrucNash}{\left ( \sCNash, \{\aSCs(y)\}_{y \in \setR},  \{\bSCs(\cdot|y)\}_{y \in \setR} \right )}

\newcommand{\as}{\alpha^*}
\newcommand{\aC}{\alpha_C}
\newcommand{\asC}{\alpha^*_C}

\newcommand{\aSCs}{{\alpha^*_{\setC}}}

\newcommand{\bs}{\beta^*}
\newcommand{\bC}{\beta_C}
\newcommand{\bsC}{\beta^*_C}

\newcommand{\bSCs}{{\beta^*_{\setC}}}

\newcommand{\xs}{x^*}
\newcommand{\xsC}{x^*_C}
\newcommand{\xsy}{x_y^*}

\newcommand{\Dl}{\Delta}
\newcommand{\Ds}{\Delta^*}
\newcommand{\Dsy}{\Delta_y^*}

\newcommand{\QsC}{Q^*_C}

\newcommand{\UdC}{U^{(d)}_C}
\newcommand{\FdC}{F^{(d)}_C}
\newcommand{\UsC}{U^{(s)}_C}
\newcommand{\FsC}{F^{(s)}_C}

\newcommand{\Bbl}{\Big (}
\newcommand{\Bbr}{\Big )}
\newcommand{\Bsbl}{\Big [}
\newcommand{\Bsbr}{\Big ]}
\newcommand{\Bsl}{\Big \{}
\newcommand{\Bsr}{\Big \}}

\title{Modeling and Analysis of Information Communities}

\author{Peter Marbach\\ Department of Computer Science\\ University of Toronto}

\nipsfinalcopy 

\begin{document}

\maketitle

\input{intro.tex}

\input{results.tex}

\bibliography{paper}{}
\bibliographystyle{plain}

\input{analysis.tex}

\end{document}

%% file: intro.tex
\begin{abstract}
In this paper we propose a mathematical model to study communities in social networks where agents share/exchange information. We refer to such a social network as an information network. More precisely, we assume that there is a population of agents who are interested in sharing content. Agents differ in the type of content they are interested in, as well as their ability to produce content. Agents can form/join communities in order to maximize their utility for obtaining and producing content. We use the framework of Nash equilibrium to characterize the community structures that emerge under this model. In our analysis, we show that there always exists a Nash equilibrium. In addition, we characterize the community structure, as well as the content that is being produced, under a  Nash equilibrium.
\end{abstract}

\section{Introduction}\label{section:introduction}

Communities play an important role in social networks. However while there exists a large body of work that formally models and studies macroscopic properties of social networks such as the degree distribution and diameter, less work is available on mathematical models for microscopic properties of communities in social networks. Creating such a model is the topic of this paper. Ideally this model should be simple enough to allow a formal analysis, yet be expressive enough to provide insights into important microscopic properties of communities in social networks.

For our model we focus on a particular type of social networks, to which we refer to as \emph{information networks}, where agents (individuals) share/exchange information. Sharing/exchanging of information is an important aspect of the social networks, both for social networks that we form in our everyday lives, as well as for online social networks such as for example Twitter. 

To model communities in information networks, to which we refer to as \emph{information communities}, we assume that all content (information)  that is being shared is being produced by agents in the network. Furthermore, we assume that different agents have different abilities to produce (generate) content, and different interests in consuming (obtaining) content. Using this setup, we study the situation where agents form communities in order to share/exchange content. We formally model this situation through utility functions that captures the benefit that agents obtain from being in a particular information community, where we assume that agents obtain an utility both from consuming (obtaining) content as well as from producing content. Using this model we study the following questions: a) which community structures will emerge under this model, i.e. can we characterize the set of content producers/consumers that form a community, and b) which content will individual agents produce in a given community.

The model that we use to characterize information networks has three important components: how we model the space of content that is being produced and consumed in the, how we model agent's interests and ability to produce content, and how we measure the utility that agents obtain for consuming and producing content.

For our analysis we assume that each content item is being produced in an information network can be associated with a particular content type. We might think of a content type as a topic, or particular interest that agents have. In addition, we assume that the space of possible content types has a certain structure. In particular,  we assume that  there exists a measure of ``closeness'' (similarity) between content types that characterizes how strongly related two content types are. For example, ``basketball'' and ``baseball'' are both sports and one would assume that these two topics are stronger related with each other than for example ``basketball'' and ``mathematics''. The assumption that there exists such a measure of closeness is an important feature of our model.

Another important aspect of the proposed model is how we characterize the agent's interest, and agent's ability to produce content. For our analysis, we assume that agents in a community might have different interests with respect to the content they are interested in, and different abilities to produce content. The model which type of content a given agent is interested in is based on the following intuition. We assume that each agents has a main interest, i.e. there exists a content type (topic) that the agent is interested in the most. However, we assume that  that agents are not only interested in getting information for their main interest but also other type of content. Moreover, the further away a given content type is from their center of interest, the less interested they are in this content type. Or in other words, the less similar a given content type is with respect to the agent's main interest, the less interested the agent is in this content type. For example, an agent that is interested in ``baseball'' might also be somewhat interested in ``basketball'', but not necessarily in ``mathematics''.

Similarly, we assume that agent's have different abilities to produce content. Here we assume that agent's have the highest ability to produce content for their center of interest. Or in other words, if an agent produces content items for the topic they are most interested in, then this content item is very likely to be relevant for this topic (content type). In addition, the further away (dissimilar) a given content type is from the agent's center of interest, the less likely it is that a content item that the agent produces for this topic is indeed relevant.

Finally, we model the utilities that agents obtain for consuming and producing content as follows. We assume that for each content item that an agent consumes, the agent receives a reward if that content item is of interest to the agent. In addition, each content item that the agent consumes that agents pays a cost.  This cost accounts for  the effort/time required by an agent to consume a single content item (and decide whether it is of interest or not). Similarly, we assume that for each content item that an agent produces it receives a benefit that is equal to the number of agents that consume the content item, and find it to be of interest to them. In addition, each content item that an agent produces incurs a cost that is proportional to the number of agents that consume the content item. Under this model, an agent receives a high utility for content consumption if it receives a lot of content item that are of interest, and very few content items that are not of interest. Similarly, an agent receives a high utility for content production if most agents that consume the content items the agent produces find the content item of interest. These definitions imply that an information network is efficient (i.e. agents receive a high utility), if agents receive all content that is of interest to them, and no content that is not of interest. Of in other words, the information network has a structure such that the content each agent receives is content that is of interest to the agent.

In our analysis we focus on a particular content space, namely a one dimensional content space with the torus metrics. In addition, we focus on the limiting case of a dense agent population where we characterize the agent population through a density function. These assumptions simplify the analysis, and lead to results that are easier to interpret and obtain insights into the community structure in information networks.

An interesting outcome of the analysis is that the proposed model, albeit being very simple, indeed seems to be able to provide interesting insights into the microscopic structure of information communities. For example, the characterization of how content is being produced, i.e. which content each agent in a community produces, indeed matches what has been experimentally observed in real-life social networks. We discuss this connection in more details in Section~\ref{section:related}~and~\ref{section:results}.

In summary, the contributions of the paper are as follows. We propose a mathematical model to study information communities that is simple enough to allow a formal analysis. In particular, using a game-theoretic framework, we characterize the community structure that arises under a Nash equilibrium. Furthermore, the model is able to capture important microscopic properties of information communities that have been observed in real-life social networks, such as for example the structure of content production in information communities. 

Whereas this paper focuses on the analysis of the proposed model, there are several possible extension, as well as applications of the model to study additional properties of information networks such as content filtering and dissemination, as well as the connectivity structure within an information community. We discuss this research directions in more details in Section~\ref{section:conclusions}.

The rest of the paper is structured as follows. In Section~\ref{section:related} we discuss existing work that is most closely related to the model and analysis presented in this paper. In Sections~\ref{section:content_model}~-~\ref{section:nash}  we define our mathematical model for information communities, and in Section~\ref{section:results} we present our result on community structures that emerge under a Nash equilibrium.

\section{Related Work}\label{section:related}
We note that there is a large body work on the experimental and theoretical study of social networks and their properties; due to space constraints  we highlight here only the work that is most relevant to the discussion presented in this paper.

There exists extensive work, both experimental and theoretical, on the macroscopic properties of social network graphs such as the small world phenomena, shrinking diameter, and power-law degree distribution. By now there are several mathematical models that describe well these properties; examples of such models are Kroenecker graphs~\cite{Kronecker} and geometric protean graphs~\cite{geo-p}. The difference between this body of work and the model presented here is that these models on the social network graph a) do not explicitly model and analyze community structures, and b) focus on macroscopic properties of the social network graph rather than microscopic properties of communities in social networks. 

 There is also a large body of work on community detection algorithms (see for example~\cite{survey_clustering} for a survey) including minimum-cut methods, hierarchical clustering, Givran-Newman algorithm, modularity maximization, spectral clustering, and many more. In this paper we are not so much concerned with detecting communities (clusters) in a social network (complex graph), but with modeling and characterizing the community structure that emerges in an information network. An interesting approach to community detection is taken in~\cite{game} by Chen et al. who use a game-theoretic approach to detect overlapping communities in social networks. In~\cite{game} it is assumed that there already exists an underlying social graph for the social network, and the utility that agents obtain when joining a community depends on a) which other agents joined the community and b) the underlying social graph. This is different from the approach in this paper where the utility depends on content production and consumption. In addition, the goal of the work in~\cite{game} is to develop an algorithm to detect overlapping communities, whereas in this paper we aim at characterizing the microscopic structure of information communities.


 Related to the analysis in this paper is the work on content forwarding and filtering in social networks~\cite{goel,filtering,hegde}. In particular the work by Zadeh, Goel and Munagala~\cite{goel}, and the work by Hegde, Massoulie, and Viennot~\cite{hegde}. In~\cite{goel}, Zadeh, Goel and Munagala consider the problem of information diffusion in social networks under a broadcast model where content forwarded (posted) by a user is seen by all its neighbors (followers, friends) in the social graph. For this model, the paper~\cite{goel} studies whether there exists a network structure and filtering strategy that lead to both high recall and high precision. High recall means that all users receive all the content that they are interested in, and high precision means that all users only receive content they are interested in. The main result in~\cite{goel} shows that this is indeed the case under suitable graph models such as for example Kronecker graphs. In~\cite{hegde},  Hegde, Massoulie, and Viennot study the problem where users are interested in obtaining content on specific topics, and study whether there exists a graph structure and filtering strategy that allows users to obtain all the content they are interested in. Using a game-theoretic framework (flow games), the analysis in~\cite{hegde} shows that under suitable assumptions there exists a Nash equilibrium, and selfish dynamics converge to a Nash equilibrium. The main difference between the model and analysis in~\cite{goel,hegde} and the approach in this paper is that model and analysis in~\cite{goel,hegde} does not explicitly consider and model community structures,  and the utility obtained by users under the models in~\cite{goel,hegde} depends only on the content that agents receive, but not on the content agents produce.

There exists an interesting connection between the modeling assumption made by Zadeh, Goel and Munagala in~\cite{goel} and by Hegde, Massoulie, and Viennot in~\cite{hegde}, and a result obtained in this paper (Proposition~\ref{prop:nash}, Section~\ref{section:nash}). Both papers~\cite{goel,hegde} make the modeling assumption that users produce content only on a small subset of content that they are interested in receiving. Zadeh, Goel and Munagala support this assumption in~\cite{goel} through experimental results obtained on Twitter data that shows that Twitter users indeed tend to produce content on a narrower set of topics than they consume. Proposition~\ref{prop:nash} in Section~\ref{section:results} in this paper provides a formal validation/explanation for this assumption as it shows that under the proposed model it is optimal for agents (users) to produce content on a small subset of the content type that they are interested in consuming. This result illustrates that the proposed model is able to capture and explain important microscopic properties of information networks and communities.

%% file: results.tex
\newpage

\section{Content Production and Consumption}\label{section:content_model}
In this section we present the mathematical model we use for our analysis. In particular,  we describe the model that we use to characterize the content that agents produce and consume in an information network. In addition, we introduce how we characterize agents' interest in a particular type of content, as well as their abilities to produce a particular type of content.

For our analysis, we model the content that agents produce in an information community as follows. We assume that each content item that is being produced by an agent in the information network belongs to a particular content type. One might think of a content type as a topic, or particular interest that agents have. Furthermore, we assume that there exists  a structure that relates the different content types with each other. In particular, we assume that  there exists a measure of ``closeness'' between content types that characterizes how strongly related two content types are. For example ``basketball'' and ``baseball'' are both sports, and one would assume that these two topics are stronger related with each other than for example ``basketball'' and ``mathematics''.

More formally, we model the type of an content item as follows. We assume that the type of a content item is given by a point $x$ in a metric space $\mS$. The assumption that content types lie in a metric space allows for a natural way to compare the closeness between content types. In particular, the distance between two content types $x,y \in \mS$ is then given by the distance measure $d(x,y)$, $x,y \in \mS$, of the metric space $\mS$.


Having defined the space $\mS$ of content types, we next describe agent's interest and ability to produce content. Recall that we assume that agents in a community might have different interests with respect to the content they are interested in, as well as  different abilities to produce content. 

We first model the content that a given agent is interested in. Our model is based on the following intuition. We assume that agents have a main interest, i.e. for each agent there exists a content type $x \in \mS$ that they are most interested in. Moreover, agents are interested in more than one content type, i.e. agents are not only interested in getting information for their main interest $x \in \mS$ but also other type of content. However, the further away a given content type is from their center of interest, the less interested they are in this content type.

To model this situation, we associate with each agent that consumes content a center of interest $y \in \mS$. The center of interest of a given agent is the content type (topic) that an agent is most interested in. Given the center of interest  $y \in \mS$ of an agent, the interest of agent $y$ in content of type  $x \in \mS$ is then given by
$$ p(x|y) = f(d(x,y)),$$
where $d(x,y)$ is the distance between the center of interest $y$ and the content type $x$, and $f:[0,\infty) \mapsto [0,1]$ is a non-increasing function. The interpretation of the function $p(x|y)$ is as follows: when agent $y$ consumes (reads) a  content item of type $x$, then it finds it interesting with probability $p(x|y)$ given by
$$p(x|y) = f(d(x,y)).$$
As the function $f$ is non-increasing, this model captures the intuition that agent an $y \in \mS$ is more interested in content that is close to its center of interest $y$. 

We will refer in the following to an agent by its center of interest $y \in \mS$, i.e. when we refer to an agent $y \in \mS$ we refer to the agent whose center of interest is equal to $y$. 

Next consider a given agent that is producing content, and let $y \in \mS$ be the center of interest of the agent. The ability of agent $y$ to produce content of type $x \in \mS$ is then given by
$$ q(x|y) = g(d(x,y)),$$
where $g:[0,\infty) \mapsto [0,1]$ is a non-increasing function.  The interpretation of this function is as follows. If agent $y$ produces content of type $x$, then the content will  be relevant to content type $x$ with probability $q(x|y)$ given by
$$q(x|y) = g(d(x,y)).$$
As the function $g$ is non-increasing, this model captures the intuition that agent an $y \in \mS$ is better at producing content that is close to its center of interest $y$.

\subsection{Assumptions on the functions $f$ and $g$}
For our analysis we make additional assumptions on the functions $f$ and $g$ that we use to define the functions $p(\cdot | y)$ and   $q(\cdot | y)$. To state these assumptions, we introduce a few more definitions.

Throughout the paper, we use the following notation. Given a real-valued function $f: \mS \mapsto R$ on a metric space $\mS$, we define the support  $supp(f)$ by
$$supp(f) = \bar A$$
where
$$A = \{ x \in \mS | f(x) \neq 0\},$$
and $\bar A$ is the closure of $A$.

Given real-valued function $f: \mS \mapsto R$ on a metric space $\mS$, we say that $f$ is symmetric with respect to  $y \in \mS$ if for $x,x' \in \mS$ such that
$$ d(x,y) = d(x',y),$$
we have that
$$ f(x) = f(x').$$

Using these definitions, we make the following assumptions for the function $f$ and $g$.
\begin{assumption}\label{ass:fg}
  The function $f: [0,L] \mapsto [0,1]$ is strictly decreasing and twice continuously differentiable on $[0,L]$, and has a bounded first derivative on $[0,L]$.
Furthermore, the function $f$ is locally strictly concave, i.e. there exists a constant $b$, $0 < b \leq L$, such that 
$$f''(x) < 0, \qquad x \in [0,b].$$
The function $g: [0,L] \mapsto [0,1]$ is non-increasing on $[0,L]$, and strictly concave and twice  continuously differentiable on its support $supp(g)$ with
$$g(0) > 0$$
and
$$g'(0) = 0.$$
\end{assumption}
The assumptions that the functions $f$ and $g$ are twice differentiable and (locally) strictly concave, are a technical assumptions used in the proofs of our results.

\subsection{Content Space $\setR$}\label{section:setR}
In addition, we consider for our analysis a metric space that has a particular structure. More precisely, we consider a one dimensional metric space with the torus metric. The reason for using  this structure is that it simplifies the analysis and allows us to obtain simple expressions for our results, that can easily been interpreted. 

More formally, we consider in the following one-dimensional metric space  for our analysis. The metric space is given by an interval  $\setR = [-L, L) \in R$, $0 < L$, with the torus metric, i.e. the distance between two points $x,y \in \setR$ is given by 
$$d(x,y) = ||x-y|| = \min \{ |x-y|, 2L - |x-y|\},$$
where $| x |$ is the absolute value of $x \in (-\infty,\infty)$.

Note that we have that
$$||x-y|| \leq L, \qquad x,y \in \setR.$$
Furthermore, we have the following two properties,
\begin{enumerate}
\item for $x,y \in \setR$, the addition of $x$ and $y$ is given by
$$ x+ y =
\left \{ \begin{array}{ll}
x +y, & x+y \in [-L,L). \\
-2L + x + y, & x+y \geq L. \\
2L + x + y, & x+y < -L.
\end{array} \right .$$
\item  we have that
$$x < y, \qquad x,y \in \setR.$$
if there exists a point $b$, $b \in (0,L)$, such that
$$ x = y - b.$$
\end{enumerate}

Using the torus metric for the content space $\setR$ eliminates ``border effects'', in the sense are no points that have a ``special'' position as it would be for example the case if we would an interval $[-L,L]$ as the content space. This simplifies the analysis, and leads to simpler expressions for our results.

\subsection{Limiting Case of a Dense Agent Population}
Finally, we characterize the set of agents that participate in the information network, i.e. the set of agents consume and produce content in the information network. As we identify an agent by its center of interest $y \in \setR$, the set of agents that consume content, as well as the set of agents that produce content, are then given by a subset of $\setR$. There are two possible approaches to characterize these two sets of agents.

One approach is to assume that there exists a finite number of agents that produce, and consume, content. For example, the set of agents that consume content could be given by set of $K$ agents $y_k$, $k=1,...,K$ such that
$$ y_k = -L + k \delta,$$
where $\delta > 0$ is given by
$$ \delta = \frac{2L}{K}.$$
This approach would result of a set $K$ agents are ``uniformly'' spaced  over the set $\setR$ with  distance $\delta$.

Another  approach is to consider the limiting case of a dense agent population. In particular, we obtain this limiting case by letting  the above distance $\delta$ between agents to approach 0, and consider the resulting density $\mu(y)$, $y \in \setR$, of agents in $\setR$. We refer to this case as the continuous agent population model. 

In this paper we consider the second approach and characterize the distribution of agents on $\setR$ by a density $\mu(y)$, where we make the following assumption on the density $\mu(y)$. 

\begin{assumption}\label{ass:agents}
Both agents that consume content, and agent that produce content, are uniformly distributed on $\setR$ with density
$$\mu(y) = 1, \qquad y \in \setR.$$ 
\end{assumption}

\newpage

\section{Information Communities}~\label{section:communities}
Having defined the content space, and the agent's interest and ability to produce content, we next model a community in an information network. Here we assume that a community is given by a set of users that produce content, and a set of users that consumes content.  

More formally, for our analysis we assume that an information community $C = (C_d,C_s)$ is given by a set $C_d \subseteq \setR$ of agents that consume content in the community $C$, and a set $C_s \subseteq \setR$ of agents that produce content in the community $C$. The subscript in $C_d$ in $C_s$ refer to ``demand'', and ``supply'', respectively. In the following we define more precisely how content is being produced and consumed, and the resulting utility functions for content consumption and production. 

For a given a community $C= (C_d,C_s)$, we assume that each agent $y \in C_s$  can decide how much effort they put into producing content of type $x$. We model this situation as follows. We let $\beta_C(x|y)$, $x \in \setR$,  be the rate at which an agent $y \in C_s$ generates content of type $x$ in the community $C$. The function $\beta_C(\cdot|y)$ is a non-negative function, and we have that
$$\beta_C(x|y) \geq 0, \qquad x \in \setR.$$
The total rate (over all agents $y \in C_s$) at which content of type $x$ is generated in the community $C$ is then given by
$$\beta_{C}(x) = \int_{C_s}  \beta_C(x|y) dy.$$
Recall that content of type $x$ that is generated by agent $y \in C_s$ is relevant to $x$ with probability $q(x|y)$, and that total rate (over all agents $y \in C_s$) at which relevant content is generated in $C$ is equal to
$$Q_C(x) = \int_{C_s}  \beta_C(x|y) q(x|y) dy.$$
We refer to $Q_C(x)$ as the content supply function for community $C$. 

Furthermore,  we assume that each agent $y \in C_d$ can decide on how much time it allocates to read (consume) content that is being produced in community $C$. Let $\alpha_C(y)$,
$$ 0 \leq \alpha_C(y) \leq 1,$$
be the the fraction of time that agent $y \in C_d$ allocates for consuming content in community $C$, and let the function $P_C(x)$ be given by
$$P_{C}(x) = \int_{C_d}  \alpha_C(y) p(x|y) dy.$$
We refer to $P_C(x)$ as the content demand function of  community $C$.

Using these definitions, we next characterize the the utilities that agents obtain in a given community $C= (C_d,C_s)$ for consuming, and producing, content. 
We first define the utility for content consumption. For this, we assume that consuming a content item incurs a cost $c$. The cost  $c$ is a (per item) processing cost that reflects the effort/time required by an agent to  read a single content item (and decide whether it is of interest or not). If the content item is of interest then the agent receives a reward equal to 1; otherwise the agent receives a reward equal to 0.

Using this cost and reward structure, we then obtain the following utility function for agents $y \in C_d$.  Note that if the fraction of time that agent $y \in C_d$ is consuming content in community $C$ is equal to $\aC(y)$,
$$ 0 \leq \aC(y) \leq 1,$$
then the (time-average) rate $\mu_C(x|y)$ at which an agent $y \in C_d$ receives rewards in community $C$ is given by
\begin{eqnarray*}
\mu_C(x|y) &=& \alpha_C(y) \int_{C_s} \beta_C(x|z) q(x|z) p(x|y) dz \\
&=& \alpha_C(y) Q_C(x) p(x|y).
\end{eqnarray*}
The (time-average) rate at which an agent $y \in C_d$ consumes content of type $x$ in the community $C$ is given by 
$$\alpha_{C}(y)\beta_C(x),$$
and the (time-average) cost rate on agent $y$ is given by
$$\alpha_{C}(y)\beta_C(x)c.$$
Combining these results, the (time-average) utility  rate for content consumption $U^{(d)}_C(y)$ of an agent $y \in C_d$ in community $C$ is given by
\begin{equation}~\label{eq:UdC}
  U^{(d)}_C(y) = \alpha_{C}(y) \int_{\setR} \Big [ Q_C(x) p(x|y) - \beta_C(x) c \Big ] dx.
\end{equation}

Similarly, the (time-average) utility rate for content production  $U^{(s)}_C(y)$ of an agent $y \in C_s$ in community $C$ is given by
\begin{eqnarray} \nonumber
U^{(s)}_C(y) = \int_{x \in \setR} \beta_{C}(x|y) \int_{z\in C_d} \alpha_C(z) \Big [  q(x|y) p(x|z)- c \Big ] dz dx \\~\label{eq:UsC}
= \int_{\setR} \beta_{C}(x|y) \Big [q(x|y) P_C(x) - \alpha_C c \Big ] dx,
\end{eqnarray}
where
$$\aC = \int_{I_C} \aC(y) dy.$$
This utility rate has the following interpretation. Note that $[q(x|y) p(x|z)- c]$ is the expected reward that agent $z \in C_d$ receives from content of type $x$ that is being produced by agent $y \in C_s$. Therefore in an economic setting  $[q(x|y) p(x|z)- c]$ is the amount (price) that agent $z$ is willing to pay agent $y$ for obtaining from agent $y$ content of type $x$. In this sense, one interpretation of the utility rate of a content producer $y \in C_s$ is that it reflects  the revenue that $y$ would obtain for the content that $y$ produces in $C$. An alternative interpretation, and the one we adopt in this paper, is that the utility rate of content producer $y \in C_s$ reflects the reputation, or ``reputation score'', of agent $y$ in the community  $C$, i.e. it captures how beneficial the contributions of  a content producer $y$ are for the community $C$. 
 
An interesting aspect of this model is that that total utility rates over all agents for content consumption and content production is the same, i.e. one can show that
\begin{eqnarray*}
\int_{C_d} U^{(d)}_C(y) dy = \int_{C_s} U^{(s)}_C(y)dy 
= \left (\int_{\setR} P_C(x)Q_C(x) dx \right ) - \alpha_C \beta_C c,
\end{eqnarray*}
where
$$\beta_C = \int_{\setR} \beta_C(x) dx.$$

\subsection{Interpretation of the Model}

In our model we assume that there are agents that produce content as well as agents that consume content. There are two possible ways to interpret/apply this model. First we can assume that there exist two sets of agents: one set of agents that exclusively consumes content, and a second set of agents that exclusively produces content. In this model, the set of agents that produce and consume content are disjoint. Second, we can assume that there exists a single set of agents where each agent both consumes and produces content, and agents are best at producing the content type which they are most interested in.

For the model that we consider here, i.e. under the assumption that there exists an infinite number of agents as given by Assumption~\ref{ass:agents} both interpretation lead to identical results, i.e. our results in Section~\ref{section:results} are valid for both cases.


\newpage
\section{Community Structure and Nash Equilibrium}\label{section:nash}
In the previous section  we characterized the utility that agents obtain when consuming  and producing content  in a given community. In the following, we assume that there are several communities, and each agent can decide which community to join, and how  much time/effort to allocate to produce and consume content in a community they join. Moreover, we assume that agents make these decisions  in order to maximize their utility rates. For this situation we use a game-theoretic approach to characterize the community structures that emerge in an information network. In particular, we  characterize the community structure that emerges using the concept of  a Nash equilibrium. To do that we describe  in the following how we define a community structure, and a Nash equilibrium, in an information network.

For our analysis, we characterize  a community structure in an information network by a triplet $(\setC, \{\alpha_{\setC}(y)\}_{y \in \setR},  \{\beta_{\setC}(\cdot|y)\}_{y \in \setR})$  where $\setC$ is the set of communities $C$ that exist in the structure, and 
$$\alpha_{\setC}(y) = \{ \alpha_C(y) \}_{C \in \setC}$$  and
$$\beta_{\setC}(y) = \{ \beta_C(\cdot |y) \}_{C \in \setC}$$
indicate how individual agents consume and produce content in the different communities $C=(C_d,C_s) \in \setC$. In the following, we require that for a given community structure   $(\setC, \{\alpha_{\setC}(y)\}_{y \in \setR},  \{\beta_{\setC}(\cdot|y)\}_{y \in \setR})$, we have for each community
$$C=(C_d,C_s) \in \setC$$
that
$$ \alpha_C(y) > 0, \qquad y \in C_d,$$
and 
$$ ||\beta_{\setC}(y)|| > 0, \qquad y \in C_s,$$
i.e. each agent  $y \in C_d$ allocates a strictly positive fraction of time  $\alpha_C(y)$ to community $C$, and each agent $y \in C_s$ has a positive total rate $|| \beta_{\setC}(y)||$ in community $C$. 

Furthermore, we require that the total fraction of time for content consumption over all communities is bounded by a constant $E_p$, $0 < E_p \leq 1$, and we have that
$$|| \alpha_{\setC}(y) || = \sum_{C \in \setC} \alpha_C(y) \leq E_p, \qquad y \in \setR,$$
where we use the convention that for a community $C = (C_d,C_s)$ we have that
$$\aC(y) = 0, \qquad \mbox{ if } y \notin C_d.$$

Similarly, we assume that the total content production rates of each agent can not exceed a given threshold $E_q$, $0 < E_q$, and we have that
$$|| \beta_{\setC}(y) || = \sum_{C \in \setC} || \beta_C(\cdot|y) || \leq E_q, \qquad y \in \setR,$$
where
$$ || \beta_C(\cdot|y) || = \int_{\setR}  \beta_C(x|y) dx,$$
and we use the convention that for a community $C = (C_d,C_s)$ we have that
$$   || \beta_C(\cdot|y) || = 0, \qquad \mbox{ if } y \notin C_s.$$
These two requirements capture the intuition that each agent has limited resource (time) to consumer, and produce, content. 

In the following, we are interested in a community structure where each agent belongs to at least one community. We refer to such a community structure as a covering community structure. More precisely, we use the following definition.

\begin{definition}
We call a community structure  $(\setC, \{\alpha_{\setC}(y)\}_{y \in \setR},  \{\beta_{\setC}(\cdot|y)\}_{y \in \setR})$ a covering community structure if we have that
$$ \cup_{C=(C_d,C_s) \in \setC} C_d = \cup_{C=(C_d,C_s) \in \setC} C_s = \setR.$$
\end{definition}

Using this definition, a Nash equilibrium is then given as follows. A community structure  $(\setC^*, \{\alpha^*_{\setC}(y)\}_{y \in \setR},  \{\beta^*_{\setC}(\cdot|y)\}_{y \in \setR})$  is a Nash equilibrium if for all $y \in \setR$ we have that
\begin{align*}
&\alpha^*_{\setC}(y) = \\
&\underset{\alpha_{\setC}(y): || \alpha_{\setC}(y)  || \leq E_p}{\arg\max} \sum_{C \in \setC}  \alpha_{C}(y) \int_{\setR} [Q_C(x) p(x|y) - \beta_C(x) c] dx
\end{align*}
and
\begin{align*}
&\beta^*_{\setC}(\cdot|y) = \\
&\underset{\beta_{\setC}(\cdot|y): || \beta_{\setC}(\cdot|y)  || \leq E_q}{\arg\max} \sum_{C \in \setC} \int_{\setR} \beta_{C}(x|y) [ q(x|y) P_C(x) - \alpha_C c] dx.
\end{align*}
Note that under a Nash equilibrium each agent maximizes its own utility rate, i.e. a Nash equilibrium characterizes a community structure that would emerge in an information network where each agent tries to maximize its own utility.

In the following, we study whether there always exists a Nash equilibrium. Furthermore, if there exists a Nash equilibrium we want to characterize the  community structure  $(\setC^*, \{\alpha^*_{\setC}(y)\}_{y \in \setR},  \{\beta^*_{\setC}(\cdot|y)\}_{y \in \setR})$ under a Nash equilibrium. 

For our analysis we make the following assumption on the relation of the processing cost $c$ with respect to the functions $f$ and $g$ of Assumption~\ref{ass:fg}. 
\begin{assumption}~\label{ass:fgc}
Let $f$ and $g$ be as given by Assumption~\ref{ass:fg}, and let $c$ be the processing cost as given in Eq.~\eqref{eq:UdC}~and~\eqref{eq:UsC}. Then we have that 
$$ f(0) g(0) - c > 0.$$
\end{assumption}
Note that if the assumption is not true then consumption utility under any community structure would either be equal to 0, or would be negative. As a result, this assumption excludes the (trivial) Nash equilibrium where agents do not join any information community.

\newpage
\section{Results}~\label{section:results}
In this section we present the main results of our analysis of the above model.  The proofs for the results are given in the appendix.

\subsection{Existence of a Nash Equilibrium}
Our first result establishes that there always exists a Nash equilibrium.
\begin{prop}~\label{prop:existence}
Let Assumption~\ref{ass:fg}~-~\ref{ass:fgc} be true, then there always exists a Nash equilibrium $(\setC^*, \{\alpha^*_{\setC}(y)\}_{y \in \setR},  \{\beta^*_{\setC}(\cdot|y)\}_{y \in \setR})$ that is a covering community structure.
\end{prop}
Proposition~\ref{prop:existence} states that there always exists a Nash equilibrium that is a covering community structure, i.e. a community structure for which each agent is a member of at least one community. This result is an important result as it states that (under the above model) communities indeed emerge as an optimal structure to share content in an information network.

\newpage
While Proposition~\ref{prop:existence} states that there always exists a Nash equilibrium, it does not provide a characterization of the community structure under a Nash equilibrium. In order to do that, we focus in the following on a particular family of Nash equilibria, and provide for this family a  detailed characterization of the community structure under a Nash equilibrium. 

To characterize the family of Nash equilibria we use the following definition.
\begin{definition}
We call a community $C=(C_d,C_s)$ an {\em interval community} if we have that
$$ C_d = C_s = I_C \subset \setR,$$
where $I_C$ is an interval in $\setR$ with length $|I_C|$ such that
$$ 0< |I_C|.$$
The interval $I_C$ can be closed or open, or be closed at one end and open at the other end.

For a given interval community  $C=(C_d,C_s) = (I_C,I_C)$, $ I_C \subset \setR$,  we refer to $|I_C|$ as the length of the interval, to 
$$mid(I_C) = \frac{1}{|I_C|} \int_{I_C} y dy$$
as the midpoint of the interval, and to
$$ L_C = \frac{|I_C|}{2}$$ 
as the half-length of the interval. 
\end{definition}
Given in interval community $C=(I_C,I_C)$, we refer to $mid(I_C)$ as the center of interest of community $C$.

The next result states that there always exists a Nash equilibrium  $(\setC^*, \{\alpha^*_{\setC}(y)\}_{y \in \setR},  \{\beta^*_{\setC}(\cdot|y)\}_{y \in \setR})$ such that each community $C \in \setC^*$ is an interval community.
\begin{prop}\label{prop:nash}
  Let Assumption~\ref{ass:fg}~-~\ref{ass:fgc} be true, then there exists a Nash equilibrium $(\setC^*, \{\asC{\setC}(y)\}_{y \in \setR},  \{\beta^*_{\setC}(\cdot|y)\}_{y \in \setR})$ such that
\begin{enumerate}
\item[(a)] each community $C=(C_d,C_s) \in \setC^*$ is an interval community, i.e. we have that
$$ C_d = C_s = I_C,$$
where $I_C$ is an interval in $\setR$. 
\item[(b)] the set $\{I_C\}_{C \in \setC^*}$ is a set of mutually non-overlapping intervals that covers $\setR$, i.e. we have that
  $$ I_C \cap I_{C'} = \emptyset, \qquad C,C' \in \setC^*, C \neq C',$$
  and 
  $$\cup_{C \in \setC^*} I_C = \setR.$$
\item[(c)] there exists a constant $L_I > 0$ such that
$$ L_C = L_I, \qquad C \in \setC*$$
where
$$ 2 L_I < \min \{ b, L \}$$
and  $b$ is the constant of Assumption~\ref{ass:fg}.
\end{enumerate}
Furthermore, for each community $C=(I_C,I_C) \in \setC^*$ we have that
\begin{enumerate}
\item[(a)] $\alpha^*_C(y) = E_p$, $y \in I_C$.
\item[(b)] $ \beta^*_C(x|y) = E_q \delta(\xs(y) - x)$, $y \in I_C$ and $x \in \setR$, where
$$ \xs(y) = \arg \max_{x \in \setR} q(x|y)P_{I_C}(x).$$
  and $\delta$ is the Dirac delta function.
\item[(c)] $\UdC(u) > 0$, $y \in I_C$.
\item[(d)]  $\UsC(u) > 0$, $y \in I_C$.
\end{enumerate}
\end{prop}
We provide the proof for Proposition~\ref{prop:existence}~and~\ref{prop:nash} in Appendix~\ref{app:nash}. Below we highlight some properties of the community structure given in Proposition~\ref{prop:nash}.

Proposition~\ref{prop:nash}  states that there always exists a Nash equilibrium $(\setC^*, \{\alpha^*_{\setC}(y)\}_{y \in \setR},  \{\beta^*_{\setC}(\cdot|y)\}_{y \in \setR})$ consisting of interval communities that cover $\setR$. Recall that in an interval community $C = (C_d,C_s) = (I_C,I_C)$ we have that the content producers and content consumers cover the same interests, i.e. we have that
$$ C_d = C_s = I_C.$$
This result captures that intuition that the content that is being produced and a consumed in a community should be ``aligned'' in the sense that the content producers in the community  should have a high ability to produce content that is of interest to the content consumers.

Furthermore, Proposition~\ref{prop:nash} states that under a Nash equilibrium as given in Proposition~\ref{prop:nash}, we have that in each community $C = \C$ a content producer $y \in I_C$ focuses on producing a single type of content given by
$$\xs(y) =  \arg \max_{x \in \setR} q(x|y)P_C(x).$$
As we show in Appendix~\ref{app:xs}, exclusively producing content type $\xs(y)$ indeed maximizes the content production utility of an agent. Moreover, experimental results suggest that this property indeed holds in real-life information networks. We will discuss this in more details in Subsection~\ref{subsection:discussion_results}.

Finally, Proposition~\ref{prop:nash} states that the interval communities in Proposition~\ref{prop:nash} are non-overlapping, and each agent belongs to a single community. The fact that communities are non-overlapping seems not to be aligned with what one observes in real-life information networks. We will discuss this in more details in Section~\ref{section:conclusions}.

In the following we characterize in more details the properties of a Nash equilibrium as given in Proposition~\ref{prop:nash}. In particular, we characterize the optimal content production rate of each agent under such a Nash equilibrium, the properties of the demand function and supply function in a community in under such a Nash equilibrium, as well as the resulting utility functions.

\newpage

\subsection{Optimal Content Type $\xs(y)$}~\label{subsection:xs}
Under the community structure given in Proposition~\ref{prop:nash}, each content producer $y \in I_C$ in a given community $C \in \setC^*$ produces the single content type given by
$$\xs(y) =  \arg \max_{x \in \setR} q(x|y)P_{I_C}(x).$$

In this subsection we characterize in more details the function $\xs(y)$, $y \in I_C$ for a given community $C \in \setC^*$.
In particular, we have the following result.
\begin{prop}\label{prop:xs}
  Let  $C=(C_d,C_s) = (I_C,I_C)$ be an interval community under a Nash equilibrium  $(\setC^*, \{\alpha^*_{\setC}(y)\}_{y \in \setR},  \{\beta^*_{\setC}(\cdot|y)\}_{y \in \setR})$ as given by Proposition~\ref{prop:nash}.
Then the function $\xs(y)$ given by
$$\xs(y) = \arg\max_{x \in \setR} q(x|y) P_C(x), \qquad y \in I_C,$$
is continuously differentiable and strictly increasing on $I_C$.
Furthermore we have that
$$ \xs(y) \in [y,mid(I_C)] \cap  supp(q(\cdot|y)), \qquad y \in [mid(I_C) - L_C, mid(I_C],$$
  and
$$ \xs(y) \in [mid(I_C),y] \cap  supp(q(\cdot|y)), \qquad y \in [mid(I_C), mid(I_C) + L_C].$$
\end{prop}
We provide a proof for Proposition~\ref{prop:xs} in Appendix~\ref{app:xs}.

One important result of Proposition~\ref{prop:xs} is that the function $\xs(y)$ is strictly increasing on $I_C$. This result implies that two different agents $y,y' \in I_C$, $y \neq y'$, will produce different types of content, i.e. we have that
$$\xs(y) \neq \xs(y'), \qquad y \neq y'.$$
This is an interesting result as it states that it is optimal for each agent to specialize on producing the content type $x^*(y)$ that is not produced by any other agent in the community, i.e. agent $y$ is the unique expert for content  of type $x^*(y)$ in the information community that $y$ belongs to.

Another important aspect of Proposition~\ref{prop:xs} is that the function $\xs(y)$ has the property that
$$ \xs(y) \in [y,mid(I_C)] \cap  supp(q(\cdot|y)), \qquad y \in [mid(I_C) - L_C, mid(I_C),$$
  and
$$ \xs(y) \in [y,mid(I_C)] \cap  supp(q(\cdot|y)), \qquad y \in [mid(I_C) - L_C, mid(I_C).$$
This result states that each agent $y \in I_C$  produces content that is either equal to their center of interest $y$, or closer to the center of interest $mid(I_C)$ of the community $C$ than their center of interest $y$. To get a more detailed understanding of how agents adapt the type of content that they produce towards the center of interest $mid(I_C)$ of the community $C$, we next study the function $ \Ds(y)$ given by
$$ \Delta^*(y) = || y - x^*(y)||, \qquad y \in I_C.$$
The function $\Ds(y)$  characterizes the absolute value of the ``displacement'' of the optimal content $x^*(y)$ that agent $y$ produces, and content $y$ that the agent is best at producing which is equal to content type $y$. Or in other words, the function $\Ds(y)$ characterizes by  how much an agent $y$ adapts its content $\xs(y)$ towards the center of interest of  the community $C$, i.e. by how much agent $y$ produces content $\xs(y)$ that is closer to the center of interest $mid(I_C)$ of the community than its own center of interest given by $y \in I_C$. 

In addition, the function $\Delta^*(y)$ can be used to characterize the quality of the optimal content $x^*(y)$ that agent $y$ produces as we have that
$$ q(x^*(y)|y) = g( || y - x^*(y)||) = g(\Delta^*(y)).$$

We have the following result for the function $\Ds(y)$.

\begin{prop}\label{prop:Ds}
  Let  $C=(C_d,C_s) = (I_C,I_C)$ be an interval community under a Nash equilibrium  $(\setC^*, \{\alpha^*_{\setC}(y)\}_{y \in \setR},  \{\beta^*_{\setC}(\cdot|y)\}_{y \in \setR})$ as given by Proposition~\ref{prop:nash}.
  Then  the function $\Delta^*(y)$ given by
$$ \Delta^*(y) = ||y - x^*(y)||, \qquad y \in I_C,$$
  is strictly decreasing  on  $[mid(I_C) - L_C, mid(I_C)]$, and strictly increasing on $[mid(I_C), mid(I_C)]+L_C]$ with
$$\Delta^*(mid(I_C)) = 0.$$
\end{prop}
We provide a proof for Proposition~\ref{prop:Ds} in Appendix~\ref{app:Ds}.

Proposition~\ref{prop:Ds} provides several interesting insights how agents adapt the type of content they produce towards the community $C$ they belong to.

First, Proposition~\ref{prop:Ds} states that
$$ \Ds(y) > 0, \qquad y \in I_C \backslash\{mid(I_C)\},$$
and hence
$$ \xs(y) \neq y, \qquad y \in I_C \backslash\{mid(I_C)\}.$$
This implies that agent $y \in I_C \backslash \{ mid(I_C))$ does not produce the type of content content that it can produce with the highest quality $q(x|y)$, i.e. the content that is equal to their center of interest $y$. But  agent $y  I_C \backslash\{mid(I_C)\}$ does adapt the content it produces towards the center of interest of community $C$, i.e. $y$ produces content of type $\xs(y)$ that is closer to the center of the community $mid(I_C)$ than their own center of interest $y$.

Moreover,  Proposition~\ref{prop:Ds} states that  $\Ds(y)$ is strictly decreasing  on  $[mid(I_C) - L_C, mid(I_C))$, and strictly increasing on $(mid(I_C), mid(I_C)]+L_C]$. This implies that the further away an agent is from the center of interest $mid(I_C)$, the more it will ``adapt'' the content it produces towards to the center of interest $mid(I_C)$ of the community $C$.

\newpage
\subsection{Properties of Demand $P_C(x)$ and Supply Function $Q^*_C(x)$}~\label{subsection:P_CandQ_C}
We next we study the properties of the demand function $P_C(x)$ given by
$$ P_C(x) = \int_{I_C}  \alpha_C(y) p(x|y) dy, \qquad x \in \setR, y \in I_C,$$
and the optimal supply function $Q^*_C(x)$ given by
$$Q^*_C(x) = \int_{I_C}  \beta^*_C(x|y) q(x|y) dy, \qquad \qquad x \in \setR, y \in I_C,$$
for an interval community $C = (C_d,C_s) = (I_C,I_C) \in \setC^*$  under a  Nash equilibrium $(\setC^*, \{\alpha^*_{\setC}(y)\}_{y \in \setR},  \{\beta^*_{\setC}(\cdot|y)\}_{y \in \setR})$ as given by Proposition~\ref{prop:nash}.

We first characterize the demand function   $P_C(x)$.

\begin{prop}\label{prop:P_C}
Let  $C=(C_d,C_s) = (I_C,I_C)$ be an interval community under a Nash equilibrium  $(\setC^*, \{\alpha^*_{\setC}(y)\}_{y \in \setR},  \{\beta^*_{\setC}(\cdot|y)\}_{y \in \setR})$ as given by Proposition~\ref{prop:nash}.
Then the demand function $P_C(x)$ is given by
$$P_C(x) =  E_p \int_{y \in I_C}  p(x|y) dy, \qquad x \in \setR, y \in I_C.$$
Furthermore, the function $P_C(x)$ has the following properties,
\begin{enumerate}
\item[(a)] $P_C(x)$ is symmetric around $y_0 = mid(I_C)$.
\item[(b)] $P_C(x)$ is continuously differentiable on $\setR$.
\item[(c)] $P_C(x)$ is strictly increasing on the interval $( mid(I_C) - L, mid(I_C))$, and strictly decreasing on the interval $(mid(I_C), mid(I_C) + L)$.
\item[(d)] $P_C(x)$ is twice continuously differentiable and strictly concave in $x$ on the interval $[y_0 - L_C, y_0 + L_C]$.
\item[(e)] $mid(I_C) = \arg \max_{x \in \setR} P_C(x)$
\end{enumerate}
\end{prop}
We provide a proof for Proposition~\ref{prop:P_C} in Appendix~\ref{app:P_C}.

The results of Proposition~\ref{prop:P_C} are intuitive. First, the result that
$$ mid(I_C) = \arg \max_{x \in \setR} P_C(x)$$
states that the content type that is the most popular is the content type that is identical to the the center of interest $mid(I_C)$ of the interval community $C$. Moreover, the result that $P_C(x)$ is strictly increasing on $( mid(I_C) - L, mid(I_C))$, and strictly decreasing $(mid(I_C), mid(I_C) + L)$, implies that the further a  content type $x$ is from the center of interest $mid(I_C)$ of community $C$, the less popular it is. Or in other words, Proposition~\ref{prop:P_C} states that 'fringe' content, i.e. content that is far away from the center of interest of the community, is less in demand/popular compared with content that is close to the center of interest.

\newpage
Next we characterize the properties of the optimal content supply function $Q^*_C(x)$.

\begin{prop}\label{prop:Q_C}
  Let  $C=(C_d,C_s) = (I_C,I_C)$ be an interval community under a Nash equilibrium  $(\setC^*, \{\alpha^*_{\setC}(y)\}_{y \in \setR},  \{\beta^*_{\setC}(\cdot|y)\}_{y \in \setR})$ as given by Proposition~\ref{prop:nash}. Then the optimal content supply function $Q^*_C(x)$ is given by
  $$Q^*_C(x) = E_q \int_{I_C}  \delta(x - \xs(y)) q(\xs(y)|y) dy, \qquad x \in \setR, y \in I_C,$$
  where $\delta(\cdot)$ is the Dirac delta function, and
$$x^*(y)= \arg\max_{x \in \setR} q(x|y) P_C(x), \qquad y \in I_C.$$
Furthermore we have that
\begin{enumerate}
\item[a)] the support $supp(Q_C^*(x))$ of $Q^*_C(x)$, $x \in \setR$  is given by
 $$supp(Q_C^*(x)) = [mid(I_C) - L^*_C, mid(I_C) + L^*_C]$$
 with
 $$ 0 < L^*_C < L_C,$$
\item [b)] $Q^*_C(x) = \bsC(x) q(x|s(x))$, $x \in I^*_C$,
where $\bsC(x)$, $x \in I^*_C$, is as defined in Lemma~\ref{lemma:bsC}, and
\item[c)] $Q^*_C(x)$ is continuous, and symmetric with respect to $mid(I_C)$, on $I^*_C$.
\end{enumerate}
\end{prop}
We provide a proof for Proposition~\ref{prop:Q_C} in Appendix~\ref{app:Q_C}.

Proposition~\ref{prop:Q_C} states that
$$supp(Q_C^*(x)) = [mid(I_C) - L^*_C, mid(I_C) + L^*_C]$$
where
$$0 < L^*_C < L_C,$$
This result implies  that the content type that is being produced by agents $y \in C_s = I_C$ in the community $C$ is a strict subset of $I_C$, i.e. the content that is being produced does not cover the center of interests of all agents $y \in C_d = I_C$. Or in other words, the content that is being produced is concentrated around to center of interest $mid(I_C)$ of the community, and not spread out over the interval $I_C$. This implies that under a Nash equilibrium there is no overlap between the content that is being produced in different communities, i.e. each community produces a distinct set of content types. 

Unlike the properties of the function $P_C(x)$ as given by Proposition~\ref{prop:P_C}, the function $Q^*(x)$ is not necessarily non-decreasing on $( mid(I_C) - L, mid(I_C))$, and non-increasing $(mid(I_C), mid(I_C) + L)$. Obtaining  these properties would require  additional assumptions on the functions $f$ and $g$.

\newpage
\subsection{Properties of Utility Functions $\UdC(y)$ and $\UsC(y)$}~\label{subsection:utilities}
We next study the properties of the utility rate function for content consumption $\UdC(y)$,  and the utility rate function  for content production $\UsC(y)$ for an interval community $C = (C_d,C_s) = (I_C,I_C) \in \setC^*$  under a  Nash equilibrium $(\setC^*, \{\alpha^*_{\setC}(y)\}_{y \in \setR},  \{\beta^*_{\setC}(\cdot|y)\}_{y \in \setR})$ as given by Proposition~\ref{prop:nash}.

We first study the properties of the utility rates for content consumption.

\begin{prop}~\label{prop:UdC}
  Let  $C=(C_d,C_s) = (I_C,I_C)$ be an interval community under a Nash equilibrium  $(\setC^*, \{\alpha^*_{\setC}(y)\}_{y \in \setR},  \{\beta^*_{\setC}(\cdot|y)\}_{y \in \setR})$ as given by Proposition~\ref{prop:nash}, and let  $Q^*_C(x)$ be the corresponding content supply function as given in Proposition~\ref{prop:Q_C}. 
Then the maximal utility rate function for content consumption $U_C^{(d)}(y)$ in community $C$ by an agent $y \in \setR$ is given by
$$ \UdC(y) =  \max \Bsl 0,\FdC(y)\Bsr, \qquad y \in \setR,$$
where 
$$\FdC(y)=  E_p \left [ \int_{\setR} p(x|y) Q^*_C(x) dx - 2L_CE_q c   \right ], \qquad y \in \setR,$$\
and we have that
$$ mid(I_C) = \arg \max_{y \in \setR} \UdC(y).$$
Furthermore, the function $\UdC(y)$, $y \in \setR$, has the properties that
\begin{enumerate}
\item[a)]  $\UdC(y)$ is symmetric with respect to $mid(I_C)$, and we have that
$$\frac{d}{dy} \UdC(mid(I_C)) = 0.$$
\item[b)]  $\UdC(y)$ is strictly increasing on $[mid(I_C) - L_C,mid(I_C))$, and we have that
$$\frac{d}{dy} \UdC(y) > 0, \qquad y \in [mid(I_C) - L_C,mid(I_C)),$$
\item[c)]  $\UdC(y)$ is strictly decreasing on $(mid(I_C), mid(I_C) + L_C]$, and we have that
$$\frac{d}{dy} \UdC(y) < 0, \qquad y \in  (mid(I_C), mid(I_C) + L_C],$$
\item[b)]  $\UdC(y)$ is non-decreasing on $(mid(I_C) - L,mid(I_C))$,
\item[c)]  $\UdC(y)$ is non-increasing on $(mid(I_C), mid(I_C) + L)$.
\end{enumerate}
\end{prop}
We provide a proof for Proposition~\ref{prop:UdC} in Appendix~\ref{app:UdC}.

Note that the function $\UdC(y)$ in Proposition~\ref{prop:UdC} is not only defined for agents $y \in C_s = I_C$, but for all agents $y \in \setR$. For agents $y \notin I_C$, the function $\UdC(y)$ characterizes the maximal utility rate for content consumption that agent $y$ would obtain in community $C$. Understanding the properties of the function $\UdC(y)$ is therefore important in establishing that a community  $(\setC^*, \{\alpha^*_{\setC}(y)\}_{y \in \setR},  \{\beta^*_{\setC}(\cdot|y)\}_{y \in \setR})$ as given in Proposition~\ref{prop:nash} is indeed a Nash equilibrium.

The results of Proposition~\ref{prop:UdC} are intuitive. First, the result that
$$ mid(I_C) = \arg \max_{y \in \setR} \UdC(y)$$
states that the agent at the center of interest $mid(I_C)$ of the interval community $C$ obtains the highest utility for content consumption. Moreover, Proposition~\ref{prop:UdC} states that the further away an agent $y$ is from the center of interest $mid(I_C)$ of the interval community $C$, the smaller a utility for content consumption agent $y$ receives.

As the interval communities given in Proposition~\ref{prop:nash} are of equal length, Proposition~\ref{prop:UdC} implies that the community structure $(\setC^*, \{\alpha^*_{\setC}(y)\}_{y \in \setR},  \{\beta^*_{\setC}(\cdot|y)\}_{y \in \setR})$ as given by Proposition~\ref{prop:nash} is indeed a Nash equilibrium for agents that consume content. Indeed as the utility for content consumption decreases the further away an agent is from the center of interest $mid(I_C)$ of an interval community, it is best for an agent $y \in \setR$ to consume content in the interval community $C$ for which we have that
$$y \in I_C.$$
More specifically as the interval communities in Proposition~\ref{prop:nash} are of equal length, Proposition~\ref{prop:UdC} implies that the agent ``on the border'' between two communities $C$ and $C'$, $C,C' \in \setC^*$, obtains the same utility for content consumption in both communities $C$ and $C'$, and hence is indifferent on which community to join. All the other agents $y \in C$, will obtain a higher utility in community $C$ compared with community $C'$; and vice versa.

\newpage

We next study the properties of the utility rates for content production $U_C^{(s)}(y)$.

\begin{prop}~\label{prop:UsC}
Let  $C=(C_d,C_s) = (I_C,I_C)$ be an interval community under a Nash equilibrium  $(\setC^*, \{\alpha^*_{\setC}(y)\}_{y \in \setR},  \{\beta^*_{\setC}(\cdot|y)\}_{y \in \setR})$ as given by Proposition~\ref{prop:nash}, and let $P_C(x)$, $x \in \setR$, be the corresponding content demand function as given in Proposition~\ref{prop:P_C}.
Then the maximal utility rate function for content production $\UsC(y)$ for  community $C$ by an agent $y \in \setR$ is given by
$$U^{(s)}_C(y) = \max \Big \{ 0,\FsC(y) \Big \}, \qquad y \in \setR,$$
where
$$ \FsC(y) =  E_q \Big [q(\xs(y)|y) P_C(\xs(y)) - 2 E_p L_C c \Big ],$$
and we have that
$$ mid(I_C) = \arg \max_{y \in \setR} \UsC(y).$$
Furthermore, the function $\UsC(y)$, $y \in \setR$,  has the properties that
\begin{enumerate}
\item[a)]  $\UsC(y)$ is symmetric with respect to $y_0 = mid(I_C)$ on $[mid(I_C) - L_C,mid(I_C) + L_C]$, and we have that
$$\frac{d}{dy} \UsC(mid(I_C) = 0.$$
\item[b)]  $\UsC(y)$ is strictly increasing on $[mid(I_C) - L_C,mid(I_C))$, and we have that
$$\frac{d}{dy} \UsC(y) > 0, \qquad y \in [mid(I_C) - L_C,mid(I_C)).$$
\item[c)]  $\UsC(y)$ is strictly decreasing on $(mid(I_C), mid(I_C) + L_C]$, and we have that
$$\frac{d}{dy} \UsC(y) < 0, \qquad y \in  (mid(I_C), mid(I_C) + L_C].$$
\item[b)]  $\UsC(y)$ is non-decreasing on $(mid(I_C) - L,mid(I_C)-L_C]$.
\item[c)]  $\UsC(y)$ is non-increasing on $[mid(I_C)+L_C, mid(I_C) + L)$.
\end{enumerate}
\end{prop}
We provide a proof for Proposition~\ref{prop:UsC} in Appendix~\ref{app:UsC}.

Similar to Proposition~\ref{prop:UdC}, as we have that
$$ mid(I_C) = \arg \max_{y \in \setR} \UsC(y)$$
Proposition~\ref{prop:UsC} states that  the agent at the center of interest $mid(I_C)$ of the interval community $C$ obtains the highest utility for content production. Moreover, Proposition~\ref{prop:UsC} states that the further away an agent $y$ is from the center of interest $mid(I_C)$ of the interval community $C$, the smaller a utility for content production agent $y$ receives.

As the interval communities given in Proposition~\ref{prop:nash} are of equal length, Proposition~\ref{prop:UsC} implies that the community structure $(\setC^*, \{\alpha^*_{\setC}(y)\}_{y \in \setR},  \{\beta^*_{\setC}(\cdot|y)\}_{y \in \setR})$ as given by Proposition~\ref{prop:UsC} is a Nash equilibrium for agents that produce content. Indeed as the utility for content production decreases the further away an agent is from the center of interest $mid(I_C)$ of an interval community, it is best for an agent $y \in \setR$ to produce content in the interval community $C$ for which we have that
$$y \in I_C.$$
More specifically, as the interval communities in Proposition~\ref{prop:nash} are of equal length, Proposition~\ref{prop:UsC} implies that the agent ``on the border'' between two communities $C$ and $C'$, $C,C' \in \setC^*$, obtains the same utility for content consumption in both communities $C$ and $C'$, and hence is indifferent on which community to join. All the other agents $y \in C$, will obtain a higher utility in community $C$ compared with community $C'$; and vice versa.

By combining Proposition~\ref{prop:UdC}~and~\ref{prop:UsC}, we obtain that the community structure  $(\setC^*, \{\alpha^*_{\setC}(y)\}_{y \in \setR},  \{\beta^*_{\setC}(\cdot|y)\}_{y \in \setR})$ given in Proposition~\ref{prop:nash} is indeed a Nash equilibrium, provided that the length of the intervals is not too long and all agents $y \in C_s = I_C$ and $y \in C_d = I_C$, $C \in \setC^*$ indeed obtain a positive utility in community $C$. We provide a formal presentation of this argument in Appendix~\ref{app:nash} where we prove Proposition~\ref{prop:nash}.

\newpage
\subsection{Maximal Community Size}~\label{subsection:size}
Finally, we characterize the maximal size of an interval community under a Nash equilibrium as given by Proposition~\ref{prop:nash}. 
Intuitively, in order for a community structure $(\setC^*, \{\alpha^*_{\setC}(y)\}_{y \in \setR},  \{\beta^*_{\setC}(\cdot|y)\}_{y \in \setR})$ as given in Proposition to be a Nash equilibrium, the lengths of the intervals can not be too long, i.e. the length $L_I$ given in Proposition~\ref{prop:nash} can not be too long, as otherwise the utilities of some of the agents would become negative. This intuition is indeed true. In particular we have the following results which shows that the maximal size depends on the content processing cost $c$, as well as the function $f$ and $g$ of Assumption~\ref{ass:fg}.
\begin{prop}\label{prop:size}
If  $(\setC^*, \{\alpha^*_{\setC}(y)\}_{y \in \setR},  \{\beta^*_{\setC}(\cdot|y)\}_{y \in \setR})$ is a Nash equilibrium as given by Proposition~\ref{prop:nash} and $L_I$ is the length of the intervals  $\{I_C\}_{C \in \setC^*}$,  then we have that
$$L_I < \max \left \{ l \in [0,L] \Big | \int_{0}^l (g(0)f(x) - c) dx \geq 0\right 
\},$$
where the functions $f$ and $g$ are as given by Assumption~\ref{ass:fg}.
\end{prop}
We provide a proof for Proposition~\ref{prop:size} in Appendix~\ref{app:size}.

Proposition~\ref{prop:size} states that the higher the processing cost $c$, the smaller $g(0)$ and the faster the function $f(x)$ decays, the smaller the size of a community will be. Note that a fast decay in $f$ implies that agents are interested in a very narrow range of content topics; and hence the result of Proposition~\ref{prop:size} is a property that one would expect for information communities.

\newpage

\subsection{Discussion}~\label{subsection:discussion_results}
In this section we presented the results that we obtained for the the existence, and characteristics, of community structures that are a Nash equilibrium. An interesting aspect of the obtained result is that they indeed provide insights into properties of communities in real-life information networks. In addition, the obtained result provide insight into structure properties of information community, that might can be used to design new algorithms to analyze information networks and communities. 

Proposition~\ref{prop:nash} states that each content producer in a community focuses on generating exactly one type of content, i.e we have that that an agent $y \in C_s = I_C$ produces the content unique content type $\xs(y)$ given by
$$\xs(y) =  \arg\max_{x \in \setR} q(x|y) P_C(x), \qquad y \in I_C.$$
This is an interesting results as this property/behavior has indeed been observed in real-life social networks. In particular, Zadeh, Goel and Munagala provide experimental results using data obtain from  Twitter that shows  Twitter users produce content on a very narrow set of topics, and consume content on a large set of topics~\cite{goel}. The interpretation of this experimental result based on the analysis presented in this paper is that this behavior is indeed optimal, and that the topics for which a given Twitter user produces content is the optimal content for this user to produce. In particular, based on the result in Proposition~\ref{prop:nash}, the topic(s) for which Twitter users produce content in a given community should be a very small subset of all the topics that are being consumed in this community. In this sense, the Proposition~\ref{prop:nash} provides a theoretical explanation for this property/behavior that has been observed in real-life social networks.

Proposition~\ref{prop:xs}~and~Proposition~\ref{prop:Ds} state that agents adapt the content they produce towards  center of interest of a community they belong to. This is an interesting result, as  it provides insight regarding the issue of ``homophily versus adaption'' in social networks. One aspect of communities in social networks is that members of the community tend to have similar interests, and behave similarly. One question is whether is due to homophily, i.e. do individuals form a community because they have the same interest and behave similarly, or whether this is due to adaption, i.e. is due that through the interaction in a community its members develop similar interests and behave similarly. The results of this paper suggest that both mechanism are at work in information communities. In particular, homophily (similar interests) prompts individuals to join the same community. This is captured by the result of Proposition~\ref{prop:xs} which states the communities  of a Nash equilibrium as given in Proposition~\ref{prop:nash} consists of interval communities, and hence the center of interests of agents in a given community belong to the same interval and are close together. Moreover, when an individual joins a community then adaptation prompts the individual to conform (adapt) away from their center of interest towards the center of interest of the community. This is captured by the result of Proposition~\ref{prop:xs} which states that an agent in community $C$ does not produce content $y$ for which they have the best ability, i.e. content type that is the center of  interest of the agent, but agent $y$ adapts the content type $\xs(y)$ that is produces towards the center of interest $mid(I_C)$ of the community.

Proposition~\ref{prop:P_C}  states that the content type that is the most popular is the content type that is identical to the the center of interest of the community. Moreover, it states that the further away a  content type is from the center of interest of a community, the less popular it is in the community. This result is interesting as it confirms in a formal way the intuition that communities have a ``center of interest'', i.e. there is a focus (center of interest) of the community and the closer a type of content (topic) is to that center of interest, the more popular (in demand) the content will be. While this property of a community may seem ``intuitive obvious'', it is interesting to observe that this property is indeed recovered (confirmed) using the proposed model.

Proposition~\ref{prop:Q_C}  implies that under a Nash equilibrium there is no overlap between the content that is being produced in different communities, i.e. each community produces a distinct set of content types. This is an interesting result as it suggest that each information community can be identified by  unique ``core content'' that is only produced in this community. In particular, this result suggests that each  real-life community has a ``core interest'' that is unique to the community, and hence identifies this community. Moreover, it seems that this property should be useful when trying to discover (identify) communities in information networks, i.e. one should be able to exploit this property to potential design more efficient algorithms to discover communities in information networks. How, and whether, this is possible is future research.

Finally, Proposition~\ref{prop:UdC}~and~\ref{prop:UsC} state that  the agent at the center of interest of a  community  obtain the highest utility rate for content consumption and production; and the further away an agent $y$ is from the center of interest of the community, the smaller a utility it will receive. This result suggest that there is a ``natural'' way to rank members in a community, where a member is ranked higher if it receives a higher utility. In particular, if it is possible to observe the utilities that community members receive, then it is indeed possible to rank community members in this way. How, and when, it is
possible to observe (estimate)  utilities that community members receive, and how this information can be used to design algorithms to analysis information communities,  is future research.

\newpage

\section{Conclusions}\label{section:conclusions}
We proposed a mathematical model to study communities in information networks. The model is intentionally kept as simple as possible to allow for a formal analysis. We used the model do derive properties of the community structure and communities under a Nash equilibrium. In particular, we show that there always exists a Nash equilibrium. Moreover, we focused on Nash equilibria where that consist of interval communities, i.e. the center of interest of agents that belong to the community all lie within an interval on $\setR$. For such Nash equilibria, we characterized the content type that agents produce in a given community, the content  demand and supply functions, as well as the utility functions for content consumption and production. Finally, we derive a bound on the maximal size that an interval community under a Nash equilibrium can have. One interesting question is whether all Nash equilibria under the proposed model consist of interval communities of equal length as given in Proposition~\ref{prop:nash}.

An interesting direction for future research is to investigate whether the model can be used to study additional properties of information communities. In particular, one direction of future research that we are working on is to investigate whether the proposed model can be used to study and characterize how content is being filtered and distributed within an information community. 

Another interesting future research is to study whether the family of Nash equilibria consider in this paper are the only community structures that can emerge in a Nash equilibrium, or whether there are other Nash equilibria. Or in other words, to study whether a Nash equilibria always consists of a set of interval communities of equal length, or whether it is possible to have a Nash equilibrium where the communities are of different sizes.

While the model captures important properties that have been observed in real-life information networks as discussed in the previous section, the model has limitations and fails to capture some properties of information communities. For example, it fails to capture that communities in information networks tend to have a hierarchical and overlapping structure. The reason for this is that the current model and analysis is limited to a one-dimensional metric space (for characterizing the content types) with homogeneous agents whose interest are uniformly distributed. Extending the analysis to more general settings of the is interesting future research.

%% file: analysis.tex
\appendix

\newpage







\section{Feasible Communities}~\label{app:feasible_community}
In this appendix we introduce additional definitions that we use in our analysis.

To do that, we first note that for a given a community $C= (C_d,C_s)$ in a community structure  $(\setC^*, \{\alpha^*_{\setC}(y)\}_{y \in \setR},  \{\beta^*_{\setC}(\cdot|y)\}_{y \in \setR})$ that is a Nash equilibrium,
we have that
$$ U_C^{(d)}(y) = \aC(y) \int_{\setR} \Bsbl p(x|y) \QsC(x) - \bsC(x) c \Bsbr \geq 0, \qquad y \in C_d,$$
where
$$\QsC(x) = \int_{C_s} \bsC(x|y)  q(x|y) dy$$
and
$$\bsC(x) = \int_{C_s} \bsC(x|y) dy,$$
as well as 
$$ U_C^{(s)}(y) =\int_{\setR} \bsC(x|y) \Bsbl q(x|y) P_C(x) - \aC c \Bsbr   \geq 0, \qquad y \in C_s,$$
$$P_C(x) = \int_{C_d} \asC(y) p(x|y) dy$$
and
$$\asC = \int_{C_d} \asC(y) dy,$$
i.e. the utilities of all agents $y \in C_d$ and $y \in C_s$ are non-negative. 
Indeed if this is not the case, and the utility obtained by an agent in the community is negative, then the agent is better off by not participating in the community at all which results in an utility rate equal to 0.  Therefore, a necessary condition for a community $C$ with  allocation functions $\{\alpha_C(y)\}_{y \in C_d}$ and $\{\beta_C(\cdot|y)\}_{y \in C_s}$ to be part of a Nash equilibrium is that it is all agents obtain a non-negative utility in the community $C$. For our analysis we use a  slightly stronger condition, and we require that utilities of all agents are strictly positive. We call such communities feasible communities. Note that all communities under the structure 
 $(\setC^*, \{\alpha^*_{\setC}(y)\}_{y \in \setR},  \{\beta^*_{\setC}(\cdot|y)\}_{y \in \setR})$ as given by Proposition~\ref{prop:nash} have the property that the utilities that agents receive are strictly positive.

More precisely, we have the following definition.
\begin{definition}\label{def:feasible}
Consider  a community  $C=(C_d,C_s)$ and  allocations $\{\alpha_C(y)\}_{y \in C_s}$ and $\{\beta_C(\cdot|y)\}_{y \in C_s}$ such that
$$\alpha_C(y) > 0, \qquad y \in C_d,$$
and
$$ \int_{\setR} \beta_C(x|y) dx > 0, \qquad y \in C_s.$$
Then we call the community $C$ a feasible community under  $\{\alpha_C(y)\}_{y \in C_d}$ and $\{\beta_C(\cdot|y)\}_{y \in C_s}$, if we have that
$$ U_C^{(d)}(y) = \aC(y) \int_{\setR} \Bsbl p(x|y)Q_C(x) - \bC(x) c \Bsbr dx > 0, \qquad y \in C_d,$$
where
$$Q_c(x) = \int_{C_s} \beta_C(x|y)  q(x|y) dy$$
and
$$\bC(x) = \int_{C_s} \bC(x|y) dy,$$
as well as
$$  U_C^{(s)}(y) = \int_{\setR} \beta_C(x|y) \Bsbl q(x|y)P_C(x) - \alpha_C c \Bsbr dx > 0, \qquad y \in C_s,$$
where
$$P_c(x) = \int_{C_d} \alpha_C(y)  p(x|y) dy$$
and
$$\alpha_C = \int_{C_d} \alpha_C(y) dy.$$
\end{definition}

\newpage

\section{Proof of Proposition~\ref{prop:P_C}}~\label{app:P_C}
In this appendix we prove  Proposition~\ref{prop:P_C} which characterizes the properties of the demand function $P_C(x)$ of an interval community $C$ in a Nash equilibrium as given by Proposition~\ref{prop:nash}.

Recall that for  an interval community community $C=(C_d,C_s)$ with $ C_d = C_s = I_C \subset \setR$, the demand function $P_C(x)$ given by 
$$P_C(x) =  \int_{I_C}  \alpha_C(y) p(x|y) dy, \qquad x \in \setR.$$

Furthermore, recall that for a given interval $I_C \subset \setR$ we denote with $L_C$ the half-length of the interval $I_C$, i.e. we have that
$$ L_C = \frac{|I_C|}{2},$$
and  we use $mid(I_C)$ to denote the midpoint of an interval, i.e. we have that
$$ mid(I_C) = \frac{1}{|I_C|} \int_{I_C} y dy.$$
Also recall that the metric space $\setR$ that we use in our analysis is given by an interval  $[-L, L) \in R$, $L > 0$, with the torus metric as defined in Section~\ref{section:setR}.

We then have the following result.
\begin{lemma}\label{lemma:P_C}
Let $E_p$, $0< E_p \leq 1$, be a given constant, and let  $C=(C_d,C_s)$ be an interval community, i.e. we have that  $ C_d = C_s = I_C,$ where $I_C \subset \setR$ is an interval in $\setR$. If we have that that
$$\alpha_C(y) = E_p, \qquad y \in C_d,$$
and 
$$L_C < \frac{L}{2},$$
then the demand function 
$$P_C(x) =  E_p \int_{I_C}  p(x|y) dy, \qquad x \in \setR, y \in I_C,$$
has the properties that
\begin{enumerate}
\item[(a)] $P_C(x)$ is symmetric around $y_0 = mid(I_C)$ on $(-L,L)$. 
\item[(b)] $P_C(x)$ is continuously differentiable on $\setR$.
\item[(c)] $P_C(x)$ is twice continuously differentiable and strictly increasing on $[y_0 - L,y_0)$, and strictly decreasing on the interval $(y_0, y_0 + L)$.
\item[(d)] $P_C(x)$ is strictly concave on $[y_0 - L_C, y_0 + L_C]$.
\item[(e)] $y_0 = \arg \max_{x \in \setR} P_C(x)$
\end{enumerate}
\end{lemma}
Note that Proposition~\ref{prop:P_C} directly follows from Lemma~\ref{lemma:P_C} as the communities under a Nash equilibrium as given by Proposition~\ref{prop:nash} satisfy the assumptions made in Lemma~\ref{lemma:P_C}.

\begin{proof}
In order to prove Lemma~\ref{lemma:P_C}, without loss of generality we can assume that the interval $I_C$ is given by  
$$I_C = [-L_C,L_C], \qquad L_C > 0.$$
Using the fact that by assumption we have that
$$ \alpha_C(y) = E_p >0, \qquad y \in C_d = I_C,$$
it follows that
\begin{equation}\label{eqn:P_C}
P_C(x) =  E_p \int_{I_C} p(x|y) dy =  E_p \int_{-L_C}^{L_C} f(||y - x||) dy,
\end{equation}
where the function $f$ is as given in Assumption~\ref{ass:fg}. \\

We first show that $P_C(x)$ is symmetric around $y_0 = mid(I_C) = 0$. For this we note that for
$$x \in [0,L)$$
we have that
\begin{eqnarray*}
\int_{- L_C}^{L_C}  f(||y - x||) dy &=&  \int_{y = -L_C}^{L_C}  f(||-y + x||) dy\\
&=&  - \int_{y = L_C}^{-L_C}  f(||y + x||) dy \\
&=& \int_{- L_C}^{L_C}  f(||y + x||) dy.
\end{eqnarray*}
Combining the above result with Eq.~\eqref{eqn:P_C} it follows that
$$P_C(x) = P_C(-x),$$
and hence $P_C(x)$ is symmetric around $y_0 = mid(I_C) = 0$ on $(-L,L)$.

Next we show that $P_C(x)$ is strictly increasing on the interval $[- L, 0)$, and strictly decreasing on the interval $(0,L) $.
For this  we let
$$ s = y-x, \qquad y \in I_C, x \in (-L,L).$$
Using this definition, we obtain that
$$P_C(x) = E_p \int_{- L_C - x}^{L_C - x} f(||s||) ds$$
and 
\begin{equation}\label{eqn:diff_P_C}
\frac{d}{dx} P_C(x) = - E_pf(||L_C - x||) + E_pf(||L_C + x||).
\end{equation}
As the function  $||L_C +x||$, $x \in \setR$,  is continuous  in $x$, and by Assumption~\ref{ass:fg} the function $f$ is a continuous function on $[0,L]$, it follows that $P_C(x)$ is continuously differentiable in $x$ on $\setR$. 

Note that for $x = 0$ we have that
$$\frac{d}{dx} P_C(x=0) = 0.$$

Using the assumption that
$$L_C < \frac{L}{2},$$
we have that
$$ ||L_C - x|| < ||L_C + x||, \qquad x \in (0,L].$$
Combining this result with Assumption~\ref{ass:fg} which states that the function $f$ is strictly decreasing on $[0,L]$, we obtain that
$$f(||L_C - x||) > f(||L_C + x||),  \qquad x \in (0,L],$$
and
$$\frac{d}{dx} P_C(x) < 0, \qquad x \in (0,L).$$
It then follows that $P_C(x)$ is strictly decreasing on the interval $(mid(I_C), mid(I_C) + L)$.

Using the same argument, we can show that
$$\frac{d}{dx} P_C(x) > 0, \qquad x \in [-L,0),$$
and $P_C(x)$ is  strictly increasing on the interval $[- L,0)$.

Next we show that  $P_C(x)$ is twice continuously differentiable and strictly concave on $[-L_C, L_C]$. Using Eq.~\eqref{eqn:diff_P_C} and the assumption that
$$L_C < \frac{L}{2},$$
we obtain that
$$\frac{d^2}{dx^2} P_C(x) = E_pf'(||L_C - x||) + E_pf'(||L_C + x||), \qquad x \in [-L_C,L_C].$$
By Assumption~\ref{ass:fg}, we have that $f$ is continuously differentiable on $[0,L]$, it follows that  $P_C(x)$ is twice continuously differentiable on $[-L_c,L_c]$.
Furthermore, by Assumption~\ref{ass:fg} the function $f$ is strictly decreasing on $[0,L]$, and we obtain that 
\begin{equation}\label{eq:P_Cd2}
\frac{d^2}{dx^2} P_C(x) = E_pf'(||L_C - x||) + E_pf'(||L_C + x||) < 0, \qquad  x \in   [-L_C,L_C],
\end{equation}
and  $P_C(x)$ is strictly concave on $[-L_C, L_C]$.

Finally, combining results the that $P_C(x)$ is strictly increasing on $(-L,0)$ and strictly decreasing on $(0,L)$ with the fact that  $P_C(x)$ is continuous, 
it follows that
$$mid(I_C) = 0 = \arg \max_{x \in \setR} P_C(x).$$
\end{proof}

\newpage
\section{Optimal Content Production}\label{app:xs}
In this appendix we characterize the optimal content production rate under a Nash equilibrium  $(\setC^*, \{\alpha^*_{\setC}(y)\}_{y \in \setR},  \{\beta^*_{\setC}(\cdot|y)\}_{y \in \setR})$ as given by Proposition~\ref{prop:nash}. In particular, we show that it is optimal for an agent to produce a single content type $\xs(y)$ as indicated by Proposition~\ref{prop:xs}.

\subsection{Properties of the Optimal Content Production Rate $\bsC(x|y)$}
We start out by characterizing  the optimal contention production rate  $\beta^*_C(x|y)$ of an agent $y$ in an interval community. More precisely, given an interval community $C=(C_d,C_s)$ such that  $ C_d = C_s = I_C,$ where $I_C \subset \setR$ is an interval in $\setR$,  we characterize how an agent $y \in C_s = I_C$ optimally allocates its content production rate $\beta_C(\cdot|y)$ in community $C$. Recall that the optimal allocation $\bsC(\cdot|y)$ of an agent $y \in C_s$ in the community $C$ is given by
\begin{equation}\label{eq:optimal_production}
\bsC(\cdot|y) = \underset{\beta(\cdot|y): || \beta(\cdot|y)|| \leq \beta_C(y)}{\arg\max} \int_{\setR} \beta_C(x|y) \Bsbl q(x|y) P_C(x) - \alpha_C c \Bsbr dx,
\end{equation}
where $ \beta_C(y)>0$ denotes the maximal production rate of agent $y$ can allocate to community $C$, and the functions $P_C(x)$ and $\aC$ are given by
$$P_C(x) = \int_{I_C} \aC(y) p(x|y) dy$$
and
$$\aC = \int_{I_C} \aC(y) dy.$$

Note that if for given agent $y \in I_C$ we have that
$$\max_{x \in \setR} \Bsbl q(x|y) P_C(x) - \alpha_Cc \Bsbr \leq 0,$$
then the optimal production rate $\bsC(\cdot|y)$ of agent $y$ is given by
$$ \bsC(x|y) = 0, \qquad x \in \setR.$$
As a result we focus in the following on agents $y \in I_C$ for which we have that
$$\max_{x \in \setR} \Bsbl q(x|y) P_C(x) - \alpha_Cc \Bsbr > 0.$$
For this case we show that  it is optimal for agent $y$ to only produce content that maximizes the function $q(x|y) P_C(x)$. More precisely, we have the following result. 
\begin{lemma}\label{lemma:optimal_value_production}
Let $E_p$, $0< E_p \leq 1$, be a given constant, and let $C=(C_d,C_s)$ be an interval community, i.e. we have that  $ C_d = C_s = I_C,$ where $I_C \subset \setR$ is an interval in $\setR$. Furthermore assume that
$$\alpha_C(y) = E_p, \qquad y \in C_d=I_C,$$
and for $y \in C_s = I_C$ and $\beta_C(y) > 0$, let the function $\bsC(x|y)$, $x \in \setR$, be given by
$$\beta^*_{C}(\cdot|y) = \underset{\beta(\cdot|y): || \beta(\cdot|y)|| \leq \beta_C(y)}{\arg\max} \int_{\setR} \beta_C(x|y) \Bsbl q(x|y) P_C(x) - \alpha_C c \Bsbr dx,$$
and let
$$ b^* = \max_{x \in \setR} q(x|y) P_C(x).$$
Finally, let the set $A \subseteq \setR$ be given by
$$A = \Bsl x \in \setR | q(x|y) P_C(x) = b^*\Bsr.$$
If for agent $y \in I_C$ we have that 
$$ \max_{x \in \setR} \Bsbl q(x|y) P_C(x) - \alpha_C c \Bsbr > 0,$$
then  the optimal production rate $\bsC(\cdot|y)$  of agent $y$ has the property that
$$ \int_{A} \beta^*_C(x|y) dx =  \int_{\setR} \beta^*_C(x|y) dx.$$
\end{lemma}

\begin{proof}
We first observe that if for agent $y \in I_C$ we have that
$$ \max_{x \in \setR} \Bsbl q(x|y) P_C(x) - \alpha_C c \Bsbr > 0,$$
then optimizing the function
$$\int_{\setR} \beta_C(x|y) \Bsbl q(x|y) P_C(x) - \alpha_C c \Bsbr dx$$
is equivalent to optimizing the function
$$\int_{\setR} \beta_C(x|y) q(x|y) P_C(x)  dx,$$
as the term $\alpha_C c$ does not depend on $x$.
Furthermore, we observe that if  for agent $y \in I_C$ we have that
$$ \max_{\setR} \Bsbl q(x|y) P_C(x) - \alpha_C c \Bsbr > 0, \qquad y \in I_C,$$
then it is optimal for agent $y \in I_C$ to choose an  content production rate in community $C$ such that
$$ \int_{\setR} \beta^*_C(x|y) dx = \beta_C(y),$$
i.e. the total rate allocated in community $C$ is equal to the maximal allowed rate $\beta_C(y)$ for agent $y$.

Using these two results, we prove the lemma by contradiction as follows. Suppose that for a given agent $y \in I_C$ we have that
$$ \max_{\setR} \Bsbl q(x|y) P_C(x) - \alpha_C c \Bsbr  > 0, \qquad y \in I_C.$$
Furthermore, suppose that for this agent $y$ we have that  $\beta^*_C(x|y)$, $x \in \setR$, is an optimal solution to the optimization problem 
$$\beta^*_C(\cdot|y) = \underset{\beta(\cdot|y): || \beta(\cdot|y)|| \leq \beta_C(y)}{\arg\max} \int_{\setR} \beta_C(x|y) \Bsbl q(x|y) P_C(x) - \alpha_C c \Bsbr dx$$
such that 
$$ \int_{A} \beta^*_C(x|y) dx < \beta_C(y) = \int_{\setR} \beta^*_C(x|y) dx.$$
Note that this implies that for the set 
$$ B = \setR \backslash A$$
we have that
$$ \int_{B} \beta^*_C(x|y) dx > 0.$$
As by Assumption~\ref{ass:fg} and by Lemma~\ref{lemma:P_C} the function $ q(x|y) P_C(x)$ is continuous with respect to $x$ on $\setR$, there exists a closed set $B_0 \subseteq B$ such that 
$$ \max_{x \in B_0} \Bsbl q(x|y) P_C(x) \Bsbr= b' < b^*$$
and 
$$\int_{B_0} \beta^*_C(x|y) dx = \epsilon > 0.$$
To obtain the contradiction, we consider the production rate function $\beta'_C(\cdot|y)$, $x \in \setR$, given as follows.
We pick a point $x^* \in A$ and set
$$\beta'_C(x|y) = \beta^*_C(x|y) + \epsilon \delta(x - x^*), \qquad x \in A,$$
where $\delta$ is the Dirac delta function. Furthermore, we set
$$\beta'_C(x|y) = 0, \qquad x \in B_0,$$
and 
$$\beta'_C(x|y) = \beta^*_C(x|y), \qquad  x\notin A \cup B_0.$$
Note that for the function  $\beta'_C(\cdot|y)$  we have that
$$ \int_{\setR} \beta'_C(x|y) dx = \int_{\setR} \beta^*_C(x|y) dx.$$
Furthermore, comparing $\beta^*_C(x|y)$ with $\beta'_C(x|y)$ we obtain that 
\begin{eqnarray*}
&& \int_{\setR} \beta'_C(x|y) q(x|y) P_C(x) dx - \int_{\setR} \beta^*_C(x|y) q(x|y) P_C(x) dx \\
&& = \int_{x \in A}  \epsilon \delta(x-x^*) q(x|y) P_C(x) dx - \int_{B_0} \beta^*_C(x|y) q(x|y) P_C(x) dx \\
&& \geq  \epsilon b^* - \epsilon b'  = \epsilon (b^*-b') > 0,
\end{eqnarray*}
where we used the fact that by construction we have that
$$q(x|y) P_C(x) \leq b', \qquad x \in B_0,$$
and
$$\int_{B_0} \bsC(x|y) dx = \epsilon.$$
This leads to a contradiction, i.e. as we have that
$$\int_{\setR} \beta'_C(x|y) q(x|y) P_C(x) dx - \int_{\setR} \beta^*_C(x|y) q(x|y) P_C(x) dx > 0,$$
it follows that $\bsC(x|y)$ is not an optimal rate allocation for agent $y \in I_C$. The result of the lemma then follows, i.e. if 
$$\beta^*_C(\cdot|y) =  \underset{\beta(\cdot|y): || \beta(\cdot|y)|| \leq \beta_C(y)}{\arg\max} \int_{\setR} \beta_C(x|y) ( q(x|y) P_C(x) - \alpha_C c) dx,$$
then we have that
$$\int_{\setR} \beta^*_C(x|y)dx =  \int_{A} \beta^*_C(x|y)dx = \beta^*_C(y).$$
\end{proof}

The following corollary is obtained immediately from Lemma~\ref{lemma:optimal_value_production}.
\begin{cor}\label{cor:optimal_value_production}
Let $E_p$, $0< E_p \leq 1$, be a given constant, and let $C=(C_d,C_s)$ be an interval community, i.e. we have that  $ C_d = C_s = I_C,$ where $I_C \subset \setR$ is an interval in $\setR$. Furthermore assume that
$$\alpha_C(y) = E_p, \qquad y \in C_d=I_C$$
and
$$\bC(y) > 0, \qquad y \in C_s=I_C.$$
If for agent $y \in C_s = I_C$ we have that
$$ \max_{x \in \setR} \Bsbl q(x|y) P_C(x) - \alpha_C c \Bsbr > 0,$$
then the function 
$$\beta^*_C(x|y) = \beta_C(y) \delta(x-x^*(y)), \qquad x \in \setR,$$
where $\delta$ is the Dirac delta function and
$$\xs(y) = \arg \max_{x \in \setR} q(x|y)P_C(x), \qquad y \in I_C,$$
is a solution to
$$\beta^*_{C}(\cdot|y) = \underset{\beta(\cdot|y): || \beta(\cdot|y)|| \leq \beta_C(y)}{\arg\max} \int_{\setR} \beta_C(x|y) \Bsbl q(x|y) P_C(x) - \alpha_C c \Bsbr dx.$$
\end{cor}

In the next subsection we derive properties of the function
$$ q(x|y)P_C(x), \qquad x \in \setR, y \in I_C,$$
that we then use to show that the optimization problem
$$ \max_{x \in \setR} q(x|y)P_C(x), \qquad y \in I_C,$$
in Corollary~\ref{cor:optimal_value_production} has a unique solution $\xs(y)$.

\newpage
\subsection{Properties of the function $q(x|y)P_C(x)$}
In this subsection, we characterize the properties of the function
$$ q(x|y)P_C(x), \qquad x \in \setR, y \in I_C.$$
Our first result shows that for 
$$\xs(y) = \arg \max_{x \in \setR} q(x|y)P_C(x), \qquad y \in I_C,$$
we have that
$$  q(\xs(y)|y)P_C(\xs(y)) > 0.$$

\begin{lemma}\label{lemma:qP_C_1}
Let  $E_p$, $0< E_p \leq 1$, be a given constant, and let  $C=(C_d,C_s)$ be an interval  community, i.e. we have that $ C_d = C_s = I_C,$ where $I_C \subset \setR$ is an interval in $\setR$. If we have that
$$\alpha_C(y) = E_p, \qquad y \in C_d = I_C,$$
Then for
$$\xs(y) = \arg \max_{x \in \setR} q(x|y)P_C(x), \qquad y \in I_C,$$
we have that
$$  q(\xs(y)|y)P_C(\xs(y)) > 0.$$
\end{lemma}

\begin{proof}
As by assumption we have that
$$\alpha_C(y) = E_p > 0, \qquad y \in C_d = I_C,$$
it follows from Lemma~\ref{lemma:P_C} that
$$P_C(x) > 0, \qquad x \in \setR,$$
and from Assumption~\ref{ass:fg} that
$$ q(y|y) P_C(y) = g(0) P_C(y) > 0, \qquad y \in I_C.$$
Using this result, we obtain that
$$  q(\xs(y)|y)P_C(\xs(y)) \geq   q(y|y) P_C(y) > 0, \qquad y \in I_C,$$
and the result of the lemma follows.
\end{proof}

\newpage
Our next two results provides additional properties of the function $q(x|y) P_C(x)$ that we use in our analysis.

\begin{lemma}\label{lemma:qP_C_2}
Let  $E_p$, $0< E_p \leq 1$, be a given constant, and let  $C=(C_d,C_s)$ be an interval  community, i.e. we have that $ C_d = C_s = I_C,$ where $I_C \subset \setR$ is an interval in $\setR$. Furthermore assume that
$$\alpha_C(y) = E_p, \qquad y \in C_d = I_C,$$
and
$$L_C < \frac{L}{2}.$$
Then
the function $ q(x|y) P_C(x)$ has the properties that
\begin{enumerate}
\item[a)] for $y \in I_C$, 
the function  $ q(x|y) P_C(x)$ is continuously differentiable in $x$ on $supp(q(\cdot|y))$. 
\item[b)] for
$$y \in [mid(I_C) -L_C,mid(I_C)],$$
the function $ q(x|y) P_C(x)$ is strictly increasing in $x$ on $[mid(I_C) - L,y) \cap supp(q(\cdot|y))$, and strictly decreasing in $x$ on $(mid(I_C), mid(I_C) + L) \cap supp(q(\cdot|y))$.
\item[c)]for
$$y \in [mid(I_C),mid(I_C)+L_C],$$
the function  $q(x|y) P_C(x)$ is strictly increasing in $x$ on $[mid(I_C) - L,mid(I_C)) \cap supp(q(\cdot|y))$, and strictly decreasing in $x$ on $(y, mid(I_C) + L) \cap supp(q(\cdot|y))$.
\end{enumerate}
\end{lemma}

\begin{proof}
We first prove property a) of the function $q(x|y) P_C(x)$.
By Assumption~\ref{ass:fg}  we have that $q(x|y)$ is continuously differentiable in $x$ on $supp(q(\cdot|y))$. Furthermore, by Lemma~\ref{lemma:P_C}, we have that the function $P_C(x)$ is continuously differentiable on $\setR$. It then follows that for $y \in I_C$ we have that the function $ q(x|y) P_C(x)$ is continuously differentiable in $x$ on $supp(q(\cdot|y))$. 

Next we prove property b). 
By Assumption~\ref{ass:fg} the function $g(x)$ is strictly decreasing on $supp(g(\cdot))\cap (0,L]$, and by Lemma~\ref{lemma:P_C} the function $P_C(x)$ is strictly increasing in $x$ on $[mid(I_C) - L,mid(I_C))$, and strictly decreasing in $x$ on $(mid(I_C), mid(I_C) + L)$. It then follows that for
$$y \in [mid(I_C) -L_C,mid(I_C)],$$
we have that the function $ q(x|y) P_C(x)$ is strictly increasing in $x$ on $[mid(I_C) - L,y) \cap supp(q(\cdot|y))$, and strictly decreasing in $x$ on $[mid(I_C),mid(I_C) + L) \cap supp(q(\cdot|y))$.

Property c) of the function $q(x|y) P_C(x)$ can be shown using the same argument as for property b).

The result of the lemma then follows.
\end{proof}

\newpage
\begin{lemma}\label{lemma:qP_C_concave}
Let $E_p$, $0< E_p \leq 1$, be a given constant, and let   $C=(C_d,C_s)$ be an interval  community, i.e. we have that $ C_d = C_s = I_C,$ where $I_C \subset \setR$ is an interval in $\setR$. Furthermore assume that
$$\alpha_C(y) = E_p, \qquad y \in C_d = I_C,$$
and
$$L_C < \frac{L}{2}.$$
Then for
$$y \in [mid(I_C) -L_C,mid(I_C)],$$
we have that the function $ q(x|y) P_C(x)$ is strictly concave in $x$ on the interval
$$[y,mid(I_C)] \cap supp(q(\cdot|y)).$$

Similarly for $y \in I_C$ such that
$$y \in [mid(I_C),mid(I_C)+L_C],$$
we have that the function  $q(x|y) P_C(x)$ is strictly concave in $x$ on the interval
$$[mid(I_C),y] \cap supp(q(\cdot|y)).$$
\end{lemma}

\begin{proof}
We first consider the case where
$$y \in [mid(I_C) -L_C,mid(I_C)].$$
By Assumption~\ref{ass:fg} we have that the function $q(\cdot|y)$ is twice differentiable and strictly concave on the  $[y,mid(I_C)] \cap supp(q(\cdot|y))$, and by Lemma~\ref{lemma:P_C} we have that the function $P_C(x)$ is twice continuously differentiable and strictly concave on the  $[y,mid(I_C)] \cap supp(q(\cdot|y))$. Furthermore, by Assumption~\ref{ass:fg} the function  $q(\cdot|y)$ is strictly decreasing on $(y,mid(I_C)] \cap supp(q(\cdot|y))$, and by Lemma~\ref{lemma:P_C} the function $P_C(x)$ is strictly increasing on $[y,mid(I_C)) \cap supp(q(\cdot|y))$. If follows that for
$$ x \in [y,mid(I_C)] \cap supp(q(\cdot|y)),$$
we have that
$$  \frac{d^2}{dx^2} q(x|y) P_C(x) = q(x|y)'' P_C(x) +  2 q'(x|y) P'_C(x) +  q(x|y) P''_C(x) < 0.$$
This establishes the result for the case where $y \in [mid(I_C) -L_C,mid(I_C)]$.

The result for the case where
$$y \in [mid(I_C),mid(I_C)+L_C],$$
can be  obtained using the same argument. The result of the lemma then follows.
\end{proof}

\newpage
The next proposition summarizes the properties of the function  $q(x|y) P_C(x)$ that we obtained in this appendix, and that we will frequently  use  in our analysis.
\begin{prop}\label{prop:qP_C}
Let $E_p$, $0< E_p \leq 1$, be a given constant, and let  $C=(C_d,C_s)$ be an interval  community, i.e. we have that $ C_d = C_s = I_C,$ where $I_C \subset \setR$ is an interval in $\setR$. Furthermore assume that
$$\alpha_C(y) = E_p, \qquad y \in C_d = I_C,$$
and
$$L_C < \frac{L}{2}.$$
Then for
$$y \in [mid(I_C) - L_C, mid(I_C)]$$
we have that  the function  $ q(x|y) P_C(x)$ is
\begin{enumerate}
\item[a)] strictly increasing in $x$ on $[mid(I_C) - L,y) \cap supp(q(\cdot|y))$.
\item[b)] strictly concave in $x$ on $[y,mid(I_C)] \cap supp(q(\cdot|y))$.
\item[c)] strictly decreasing in $x$ on $(mid(I_C), mid(I_C) + L) \cap supp(q(\cdot|y))$.
\end{enumerate}
Similarly, for
$$y \in [mid(I_C), mid(I_C)]+L_C$$
we have that  the function  $ q(x|y) P_C(x)$ is
\begin{enumerate}
\item[a)] strictly increasing in $x$ on $[mid(I_C) - L,mid(I_C)) \cap supp(q(\cdot|y))$.
\item[b)] strictly concave in $x$ on $[mid(I_C),y] \cap supp(q(\cdot|y))$.
\item[c)] strictly decreasing in $x$ on $(y, mid(I_C)+L) \cap supp(q(\cdot|y))$.
\end{enumerate}
\end{prop}

\newpage
Finally, we derive one additional property for the function  $q(x|y) P_C(x)$.
\begin{lemma}\label{lemma:qP_C_3}
Let  $E_p$, $0< E_p \leq 1$, be a given constant, and let $C=(C_d,C_s)$ be an interval  community, i.e. we have that $ C_d = C_s = I_C,$ where $I_C \subset \setR$ is an interval in $\setR$. Furthermore assume that
$$\alpha_C(y) = E_p, \qquad y \in C_d = I_C,$$
and
$$L_C < \frac{L}{2}.$$
Then the following is true.
If for
$$y \in [mid(I_C) -L_C,mid(I_C)]$$
we have  that
$$\sup \left \{ x \in \setR | q(x|y)P_C(x) > 0 \right \} > mid(I_C),$$
then  we have 
$$ q\big (mid(I_C)|y \big ) P_C\big (mid(I_C) \big ) > q(x|y) P_C(x), \qquad x \in [y+L,mid(I_C)+L).$$

Similarly if for
$$y \in [mid(I_C),mid(I_C)+L_C]$$
we have that
$$\inf \left \{ x \in \setR | q(x|y)P_C(x) > 0 \right \} < mid(I_C),$$
then we have 
$$ q \big (mid(I_C)|y \big ) P_C\big (mid(I_C) \big ) > q(x|y) P_C(x), \qquad x \in [mid(I_C)-L,y - L].$$
\end{lemma}

\begin{proof}
We prove the result for the case where
$$y \in [mid(I_C) -L_C,mid(I_C)].$$
The result for the case where
$$y \in [mid(I_C),mid(I_C)+L_C]$$
can be obtained using the same argument. 

Let $y$ be such that
$$y \in [mid(I_C) -L_C,mid(I_C)].$$
We first show for this case that under the assumption that
$$ 2 L_C < L,$$
it follows that when
$$\sup \left \{ x \in \setR | q(x|y)P_C(x) > 0 \right \} > mid(I_C),$$
then we obtain that
\begin{equation}\label{eq:smaller_q}
q(x|y) < q\big (mid(I_C)|y\big ), \qquad x \in [y + L,mid(I_C)+L).
\end{equation}
In order to show this result, we have to show that
$$||x-y|| > || mid(I_C)-y||, \qquad x \in [y + L,mid(I_C)+L).$$
Note that the function $|y-x||$ is decreasing in $x$ on $[y + L,mid(I_C)+L)$,and therefore in order to show it result if suffices to show that 
$$ || mid(I_C)+L -y|| >  || mid(I_C)-y||.$$
As for $y \in [mid(I_C) -L_C,mid(I_C)]$ we have that
$$ mid(I_C)+L -y > L,$$
we have from the definition of  the torus metric that
$$  || mid(I_C)+L -y|| = 2L -  mid(I_C)- L + y = y - [mid(I_C) -y].$$
To see that
$$  y - [mid(I_C) -y]>  || mid(I_C)-y||,$$
note the following.
As by assumption we have
$$ 2 L_C < L,$$
it follows that
\begin{eqnarray*}
|| mid(I_C) - y||
&<& L_C \\
&<& L - L_C\\
&<& L - [mid(I_C) - y].
\end{eqnarray*}

Combining Eq.~\eqref{eq:smaller_q} with the fact that  by Proposition~\ref{prop:P_C} the function $P_C(x)$ is  strictly decreasing in $x$ on $(mid(I_C),mid(I_C)+L)$, we obtain that for $y \in [mid(I_C) -L_C,mid(I_C)]$ we have that
$$ q(x|y) P_C(x) < q\big (mid(I_C)|y \big ) P_C\big (mid(I_C) \big ), \qquad x \in [y+L,mid(I_C)+L).$$
The result of the lemma then follows.
\end{proof}

\newpage

\subsection{Properties of the Optimal Content $\xs(y)$}

Using the result for the function $q(x|y)P_C(x)$ that we obtained in the previous subsection, we characterize in this subsection the properties of the optimal content $\xs(y)$ given by
$$\xs(y) = \arg \max_{x \in \setR} q(x|y)P_C(x).$$

Our first result follows directly from Proposition~\ref{prop:qP_C} and Lemma~\ref{lemma:qP_C_3}.
\begin{lemma}\label{lemma:xs_opt_int}
Let   $E_p$, $0< E_p \leq 1$, be a given constant, and let $C=(C_d,C_s)$ be an interval  community, i.e. we have that $ C_d = C_s = I_C,$ where $I_C \subset \setR$ is an interval in $\setR$. Furthermore assume that
$$\alpha_C(y) = E_p, \qquad y \in C_d = I_C,$$
and
$$L_C < \frac{L}{2}.$$
Let
$$\xs(y) = \arg \max_{x \in \setR} q(x|y)P_C(x), \qquad y \in I_C.$$
Then for
$$y \in [mid(I_C) -L_C,mid(I_C)]$$
we have that
$$\xs(y) \in [y,mid(I_C)]\cap supp(q(\cdot|y)).$$
Similarly, for
$$y \in [mid(I_C),mid(I_C)+L_C]$$
we have that
$$\xs(y) \in [mid(I_C),y]\cap supp(q(\cdot|y)).$$
\end{lemma}

\newpage
Next we derive conditions under which  optimization problem
$$ \xs(y) = \arg \max_{x \in \setR} q(x|y)P_C(x), \qquad y \in I_C,$$
has a unique solution.

\begin{lemma}\label{lemma:single_content_production}
Let    $E_p$, $0< E_p \leq 1$, be a given constant, and let $C=(C_d,C_s)$ be an interval community, i.e. we have that  $ C_d = C_s = I_C,$ where $I_C \subset \setR$ is an interval in $\setR$. 
If
$$\alpha_C(y) = E_p, \qquad y \in C_d = I_C,$$
and
$$L_C < \frac{L}{2},$$
then there exists a unique solution $\xs(y)$ to the optimization problem
$$\max_{x \in \setR}  q(x|y) P_C(x), \qquad y \in I_C,$$
which is the unique solution to the equation
$$  q'(x|y) P_C(x) +  q(x|y) P'_C(x) = 0, \qquad x \in I_C \cap supp(g(\cdot|y)).$$
Furthermore, $\xs(y)$ has the properties that
\begin{enumerate}
\item[a)] $\xs(mid(I_C)) = mid(I_C)$.
\item[b)] for
$$y \in [mid(I_C) - L_C,mid(I_C))$$
we have
$$\xs(y) \in (y,mid(I_C)) \cap  supp(q(\cdot|y)).$$
\item[c)] for
$$y \in (mid(I_C),mid(I_C)+L_C]$$
we have
$$\xs(y) \in (mid(I_C),y) \cap  supp(q(\cdot|y)).$$
\end{enumerate}
\end{lemma}

\begin{proof}
We prove the result for the case where
$$y \in [mid(I_C) -L_C,mid(I_C)].$$
The result for the case where
$$y \in [mid(I_C),mid(I_C)+L_C]$$
can be obtained using the same argument.

By Lemma~\ref{lemma:xs_opt_int}, we have that
$$\xs(y) \in [y,mid(I_C)] \cap  supp(q(\cdot|y)), \qquad y \in [mid(I_C) - L_C,mid(I_C)].$$
This implies that
$$\xs(mid(I_C)) = mid(I_C).$$
Furthermore, by Lemma~\ref{lemma:qP_C_concave} we have for
$$y \in [mid(I_C) -L_C,mid(I_C))$$
that the function $ q(x|y) P_C(x)$ is strictly concave on the interval $ [y,mid(I_C)] \cap  supp(q(\cdot|y))$, and it follows that the optimization problem
$$\max_{x \in \setR}  q(x|y) P_C(x), \qquad y \in [mid(I_C) -L_C,mid(I_C)).$$
has an unique solution. 

It remains to show that the unique solution $\xs(y)$ is an interior solution, and we have that 
$$\xs(y) \in (y,mid(I_C)) \cap  supp(q(\cdot|y)), \qquad y \in [mid(I_C) - L_C,mid(I_C))$$
as well as that  $\xs(y)$ is the unique solution to the equation
$$  q'(x|y) P_C(x) +  q(x|y) P'_C(x) = 0, \qquad x \in I_C \cap supp(g(\cdot|y)).$$
Note that for
$$y \in [mid(I_C) - L_C,mid(I_C))$$
and
$$x \in  [y,mid(I_C)] \cap  supp(q(\cdot|y))$$
we have that
$$  \frac{d}{dx} q(x|y) P_C(x) = q'(x|y) P_C(x) +  q(x|y) P'_C(x).$$

By Assumption~\ref{ass:fg} we have that
$$g(0) >0,$$
and by the assumption that
$$\alpha_C(y) = E_p > 0, \qquad y \in C_d = I_C,$$
have that
$$P_C(x) > 0, \qquad x \in \setR.$$
Furthermore by Assumption~\ref{ass:fg} we have have
$$g'(0) = 0$$
and by Lemma~\ref{lemma:P_C} we have 
$$ P'_C(y) > 0, \qquad y \in [mid(I_C) - L_C, mid(I_C)).$$
Combining these results, it then follows that for
$$y \in [mid(I_C) - L_C,mid(I_C))$$
we have that
$$q'(y|y) P_C(y) +  q(y|y) P'_C(y) =  g(0) P'_C(y) > 0.$$
Hence we have that
$$\xs(y) > y.$$

Next we show that for
$$y \in [mid(I_C) - L_C,mid(I_C))$$
we have that
$$\xs(y) < mid(I_C).$$
To do this, we consider two separate cases.
For the first case, we assume that for
$$y \in [mid(I_C) - L_C,mid(I_C))$$
we have
$$\sup \left \{ x \in \setR | q(x|y)P_C(x) > 0 \right \} \geq mid(I_C).$$
For this case, we obtain for
$$y_0 = mid(I_C)$$
that
$$q'(y_0|y) P_C(y_0) +  q(y_0|y) P'_C(y_0) = q'(y_0|y) P_C(y_0) < 0,$$
where we used the fact that by Assumption~\ref{ass:fg} we have have
$$q'(x|y) < 0, \qquad x \in (y,y_0] \cap supp(q(\cdot|y)),$$
and by Lemma~\ref{lemma:P_C} we have 
$$ P'_C(mid(I_C)) = 0.$$
Hence we have for this case that
$$\xs(y) < mid(I_C).$$
For the second case, we assume that for
$$y \in [mid(I_C) - L_C,mid(I_C))$$
we have 
$$x_s = \sup \left \{ x \in \setR | q(x|y)P_C(x) > 0 \right \} < mid(I_C).$$
For this case, we obtain that
$$q'(x_s|y) P_C(x_s) +  q(x_s|y) P'_C(x_s) = q'(mid(I_C)|y) P_C(y) < 0,$$
where we used the fact that for
$$x_s = \sup \left \{ x \in \setR | q(x|y)P_C(x) \right \}$$
we have that
$$ q(x_s|y) = 0$$
and
$$q'(x_s|y) < 0.$$
It follows that
$$\xs(y) < x_s.$$
The above results state that
$$\xs(y) > y$$
and
$$\xs(y) < \min\{y_0, x_s\},$$
and hence we have that $\xs(y)$ is an interior solution, and we have that $\xs(y)$ is the unique solution to the equation
$$  q'(x|y) P_C(x) +  q(x|y) P'_C(x) = 0, \qquad x \in I_C \cap supp(g(\cdot|y)).$$

The result of the lemma then follows.

\end{proof}

\newpage
The following proposition summarizes Lemma~\ref{lemma:optimal_value_production} ~to~\ref{lemma:single_content_production} that we obtained in this appendix.

\begin{prop}\label{prop:optimal_production}
Let     $E_p$, $0< E_p \leq 1$, be a given constant, and let  $C=(C_d,C_s)$ be an interval community, i.e. we have that  $ C_d = C_s = I_C,$ where $I_C \subset \setR$ is an interval in $\setR$. Furthermore assume that
$$\alpha_C(y) = E_p, \qquad y \in C_d = I_C,$$
and
$$L_C < \frac{L}{2}.$$
If for agent $y \in I_C$ we have that
$$\max_{x \in \setR} \Bsbl q(x|y) P_C(x) - \alpha_Cc \Bsbr > 0,$$
then the optimal content production rate allocation 
$$\beta^*_{C}(\cdot|y) = \underset{\beta(\cdot|y): | \beta(\cdot|y)| \leq \beta_C(y)}{\arg\max} \int_{\setR} \beta_C(x|y) ( q(x|y) P_C(x) - \alpha_C c) dx$$
of agent $y$ is given by
$$\beta^*_C(x|y) = \beta_C(y) \delta(x-x^*(y)), \qquad x \in \setR,$$
where $\delta$ is the Dirac delta function and $x^*(y)$ is the unique solution to the optimization problem
$$\max_{x \in \setR}  q(x|y) P_C(x).$$
Furthermore, $x^*(y)$ has the properties that 
\begin{enumerate}
\item[(a)] $ q(x^*(y)|y) > 0, \qquad y \in I_C$.
\item[(b)] $\xs(y)$, $y \in I_C$ is given by the unique solution to the equation
$$  q'(x|y) P_C(x) +  q(x|y) P'_C(x) = 0, \qquad x \in I_C \cap supp(q(\cdot|y)).$$
\item[(c)] we have 
$$\xs(mid(I_C)) = mid(I_C).$$
\item[(d)] we have
$$ \xs(y) \in (y,mid(I_C))\cap  supp(q(\cdot|y)), \qquad y \in [mid(I_C) - L_C, mid(I_C)),$$
  and
$$ \xs(y) \in (mid(I_C),y) \cap  supp(q(\cdot|y)), \qquad y \in (mid(I_C), mid(I_C) + L_C].$$
\end{enumerate}
\end{prop}

\newpage
Our last  result in this subsection characterizes the optimal solution $\xs(y)$ for agents $y$ outside the interval $I_C$. In particular, we have the following result.

\begin{lemma}\label{lemma:\xs_2}
Let    $E_p$, $0< E_p \leq 1$, be a given constant, and let $C=(C_d,C_s)$ be an interval community, i.e. we have that  $ C_d = C_s = I_C,$ where $I_C \subset \setR$ is an interval in $\setR$.
If we have that
$$ \alpha_C(y) = E_p, \qquad y \in I_C,$$
and 
$$L_C < \frac{L}{2},$$
then the following is true for the function
$$\xs(y) = \arg\max_{x \in \setR} q(x|y) P_C(x).$$
For
$$ y \in [mid(I_C) - L, mid(I_C) - L_C)$$
we have that
$$ \xs(y) \leq \xs(mid(I_C) - L_C),$$
and for
$$ y \in (mid(I_C) + L_C, mid(I_C) + L)$$
we have that
$$ \xs(y) \geq \xs(mid(I_C) + L_C).$$
\end{lemma}

\begin{proof}
We prove the lemma for the case where
$$ y \in [mid(I_C) - L, mid(I_C) - L_C).$$
The case where $ y \in (mid(I_C) + L_C, mid(I_C) + L)$ can be shown using the same argument.

Let $y$ be such that
$$y \in [mid(I_C) - L, mid(I_C) - L_C),$$
and let
$$y_l = mid(I_C) - L_C.$$
Recall that by Lemma~\ref{lemma:single_content_production} we have that
$$ \xs(y_l) \in [mid(I_C)-L_C,mid(I_C)].$$
Furthermore, by the same argument as given in the proof of Lemma~\ref{lemma:single_content_production}, we have that
$$ \xs(y) \in [y, mid(I_C)] \cap supp(q(\cdot|y), \qquad y \in [mid(I_C) - L, mid(I_C) - L_C).$$
Therefore, in order to prove the result of the lemma it suffices to show that if
$$ \xs(y) \in  [mid(I_C)-L_C,mid(I_C)], \qquad y \in [mid(I_C) - L, mid(I_C) - L_C),$$
then we have that
$$\xs(y) \leq \xs(y_l).$$

To do this, we note that by the same argument as given in the proof of Lemma~\ref{lemma:qP_C_concave},
we have that the function
$$ q(x|y) P_C(x), \qquad y \in [mid(I_C) - L, mid(I_C) - L_C),$$
is strictly concave in $x$ on  $[mid(I_C) - L_C,mid(I_C)]$. It then follows that if for
$$ y \in [mid(I_C) - L, mid(I_C) - L_C)$$ 
we have that
$$ \xs(y) \in  [mid(I_C)-L_C,mid(I_C)],$$ 
then the optimization problem
$$\xs(y)  \arg \max_{x \in  [mid(I_C) - L_C,mid(I_C)]} q(x|y) P_C(x), \qquad y \in [mid(I_C) - L, mid(I_C) - L_C),$$
has a unique solution, given by
$$g'(\xs(y)-y)P_C(\xs(y)) + g(\xs(y)-y)P'_C(\xs(y)) = 0,$$
where $g$ is the function of Assumption~\ref{ass:fg}.

Therefore, in order to show that if
$$ \xs(y)  \in  [mid(I_C)-L_C,mid(I_C)],  \qquad y \in  [mid(I_C)-L_C,mid(I_C)],$$
then we have that
$$ \xs(y) \leq \xs(y_l) = \xs(mid(I_C) - L_C),$$
it suffices to show that
$$ g'(\xs(y_l)-y)P_C(\xs(y_l)) + g(\xs(y_l)-y)P'_C(\xs(y_l)) \leq 0.$$
Note that by Assumption~\ref{ass:fg} we have for
$$y \in  [mid(I_C)-L_C,mid(I_C)]$$
that
$$g'(\xs(y)-y) \leq 0, \qquad \xs(y) - y \in supp(g),$$
and by Lemma~\ref{lemma:P_C} we have that
$$P'_C(\xs(y)) \geq 0.$$
Furthermore, as the function $g$ is by Assumption~\ref{ass:fg} decreasing and strictly concave on $supp(g)$ and we have that
$$x^*(y_l) - y_l \leq x^*(y_l) - y,$$
it follows that
$$  g(\xs(y_l)-y) \leq  g(\xs(y_l)-y_l)$$
and
$$  g'(\xs(y_l)-y)   \leq  g'(\xs(y_l)-y_l).$$
Combining these results with the fact that
$$g(x) \geq 0, \qquad x \in supp(g),$$
and
$$P_C(x) \geq 0, \qquad x \in \setR,$$
we obtain that
\begin{eqnarray*}
&&  \hspace{-0.2in} g'(\xs(y_l)-y)P_C(\xs(y_l)) + g(\xs(y_l)-y)P'_C(\xs(y_l)) \\
&& \leq g'(\xs(y_l)-y_l)P_C(\xs(y_l)) + g(\xs(y_l)-y_l)P'_C(\xs(y_l)) 
\end{eqnarray*}
As by Proposition~\ref{prop:optimal_production} we have that
$$g'(\xs(y_l)-y_l)P_C(\xs(y_l)) + g(\xs(y_l)-y_l)P'_C(\xs(y_l))  = 0,$$
we obtain that
$$ g'(\xs(y_l)-y)P_C(\xs(y_l)) + g(\xs(y_l)-y)P'_C(\xs(y_l)) \leq 0.$$
The result of the lemma then follows. 
\end{proof}

\newpage

\section{Proof of Proposition~\ref{prop:xs}}

In the previous appendix we showed that for a given interval community $C=(C_d,C_s) = (I_C,I_C)$ such that
$$\alpha_C(y) = E_p > 0, \qquad y \in I_C,$$
and
$$ L_C < \frac{L}{2},$$
the optimization problem
$$x^*(y) = \arg\max_{x \in \setR} q(x|y) P_C(x), \qquad y \in I_C,$$
has a unique solution.

In this subsection we study the properties of the function $x^*(y)$, $y \in I_C$, i.e. we study the properties of the optimal content function $\xs(\cdot)$ as a function of $y \in I_C$, and prove the results of Proposition~\ref{prop:xs}.

We first show that the  function $x^*(y)$ is differentiable and strictly increasing on $I_C$.
\begin{lemma}\label{lemma:xs}
Let   $E_p$, $0< E_p \leq 1$, be a given constant, and let  $C=(C_d,C_s)$ be an interval community, i.e. we have that  $ C_d = C_s = I_C,$ where $I_C \subset \setR$ is an interval in $\setR$. Furthermore assume that
$$\alpha_C(y) = E_p, \qquad y \in I_C,$$
and
$$L_C < \frac{L}{2}.$$
Then  the function $x^*(y)$ given by
$$x^*(y)= \arg\max_{x \in \setR} q(x|y) P_C(x), \qquad y \in I_C,$$
 is continuously differentiable and strictly increasing on $I_C$.
\end{lemma}

\begin{proof}
Without loss of generality, we can assume that the interval $I_C$ is given by
$$ I_C = [-L_c,L_c].$$

We first consider the case where
$$y \in [0,L_C].$$
 Recall that by Proposition~\ref{prop:optimal_production} we have that $x^*(y)$, $y \in I_C$,  is the unique solution to the equation
$$  q'(x|y) P_C(x) +  q(x|y) P'_C(x) = 0, \qquad x \in I_C \cap supp((q(\cdot|y)).$$
As by Proposition~\ref{prop:optimal_production} we have that
$$ \xs(y) \in [0,y] \cap  supp(q(\cdot|y)), \qquad y \in [0,L_C],$$
we can re-write $\xs(y)$,  $y \in [0, L_C]$, as the unique solution to
$$x^*(y)= \arg\max_{x \in [0,y]} g(y-x) P_C(x),$$
where the function $g$ is as given in Assumption~\ref{ass:fg}. 
It then follows that $x^*(y)$, $y \in [0, L_C]$,  is the unique solution to the equation
$$  - g'(y-x) P_C(x) +  g(y-x) P'_C(x) = 0, \qquad y \in [0, L_C], x \in [0,y] \cap supp(g).$$

Let the function $F^+(x,y)$ be given by
$$ F^+(x,y) = -g'(y-x) P_C(x) +  g(y-x) P'_C(x), \qquad y \in [0,L_C], x \in [0,y].$$
Note that we have that 
$$F^+((\xs(y),y) = 0, \qquad y \in [0,L_C].$$
Furthermore we have that
$$\frac{d}{dx} F^+(x,y) =  g''(y-x) P_C(x)  -  2g'(y-x) P'_C(x) + g(y-x) P''_C(x), \qquad y \in [0,L_C], x \in [0,y]\cap supp(g),$$
and
$$\frac{d}{dy} F^+(x,y) =  - g''(y-x) P_C(x) + g'(y-x) P'_C(x), \qquad y \in [0,L_C], x \in [0,y]\cap supp(g).$$
As by Assumption we have $g(x)$ is strictly decreasing and concave in $x$ for $x \in [0,L_C]\cap supp(g)$, and  $P_C(x)$ is strictly decreasing and concave in $x$ for $x \in [0,L_C]$, it follows that for $y \in [0,L_C]$ we have for
$$ y \in [0,L_C] \mbox{ and }  x \in [0,y]\cap supp(g)$$
that 
$$ g''(y-x) P_C(x) < 0,$$
$$ - g'(y-x) P'_C(x) \leq 0,$$
and
$$ g(y-x) P''_C(x) < 0.$$
Combining the above results, it follows that
$$\frac{d}{dx} F^+(x,y) < 0 , \qquad y \in [0,L_C], x \in [0,y]\cap supp(g)$$
and
$$\frac{d}{dy} F^+(x,y) > 0 , \qquad y \in [0,L_C], x \in [0,y]\cap supp(g).$$
By Proposition~\ref{prop:optimal_production} we have that
$$\xs(y)) \in [0,y]\cap supp(g), \qquad y \in [0,L_C],$$
and it follows that
\begin{equation}\label{eqn:diff_xs_x}
\frac{d}{dx} F^+(x=\xs(y),y) < 0 , \qquad y \in [0,L_C],
\end{equation}
and
\begin{equation}\label{eqn:diff_xs_y}
\frac{d}{dy} F^+(x=\xs(y),y) > 0 , \qquad y \in [0,L_C].
\end{equation}
Using the implicit function theorem~\cite{rudin}, we then obtain that the function  $\xs(y)$ is continuous and differentiable on $[0,L_C]$ with 
\begin{eqnarray*}
\frac{d}{dy} \xs(y) &=& \frac{- \frac{d}{dy} F^+(x=\xs(y),y)}{\frac{d}{dx} F^+(x=\xs(y),y)}, \qquad y \in [0,L_C].
\end{eqnarray*}
Combining this result with Eq.~\eqref{eqn:diff_xs_x}~and~\eqref{eqn:diff_xs_y}, we obtain that
\begin{eqnarray*}
\frac{d}{dy} \xs(y) &=& \frac{- \frac{d}{dy} F^+(x=\xs(y),y)}{\frac{d}{dx} F^+(x=\xs(y),y)} > 0,
\end{eqnarray*}
and the function $\xs(y)$ is strictly increasing on $[0,L_C]$.

Next we consider the case where
$$y \in [-L_C,0].$$
For this case, let the function $F^-(x,y)$ be given by
$$ F^-(x,y) = g'(x-y) P_C(x) +  g(x-y) P'_C(x), \qquad y \in [-L_C,0], x \in [y,0].$$
Note that we have that 
$$F^-(\xs(y),y) = 0, \qquad y \in [-L_C,0].$$
Using the same argument as for the case where $y \in [0,L_C]$, we obtain from the above results that the function  $\xs(y)$ is continuous and differentiable on $[-L_C,0]$ with 
\begin{eqnarray*}
\frac{d}{dy} \xs(y) &=& \frac{- \frac{d}{dy} F^-(x=\xs(y),y)}{\frac{d}{dx} F^+(x=\xs(y),y)}, \qquad y \in [-L_C,0],
\end{eqnarray*}
that we have that
\begin{eqnarray*}
\frac{d}{dy} \xs(y) &=& \frac{- \frac{d}{dy} F^-(x=\xs(y),y)}{\frac{d}{dx} F^-(x=\xs(y),y)} > 0.
\end{eqnarray*}
It follows that the function $\xs(y)$ is strictly increasing on $[-L_C,0]$.

In order to show that the function $\xs(y)$ is differentiable in $I_C$, we have to show that
$$  \frac{- \frac{d}{dy} F^+(x=\xs(0),0)}{\frac{d}{dx} F^+(x=\xs(0),0)} =  \frac{- \frac{d}{dy} F^-(x=\xs(0),0)}{\frac{d}{dx} F^-(x=\xs(0),0)}.$$
Indeed, as by Assumption~\ref{ass:fg} we have that
$$ g'(0) = 0,$$
it follows from the above equations that
$$\frac{- \frac{d}{dy} F^+(x=\xs(0),0)}{\frac{d}{dx} F^+(x=\xs(0),0)} =  \frac{- \frac{d}{dy} F^-(x=\xs(0),0)}{\frac{d}{dx} F^-(x=\xs(0),0)} =
\frac{g''(0)P_C(0)}{g''(0)P_C(0) + g(0)P_C''(0)}.$$
It remains to prove that the function  $\xs(y)$ is continuously differentiable on $I_C$. By Assumption~\ref{ass:fg} we have that the function $g(x)$, $x \in [0,L]$, is twice continuously differentiable. Furthermore, by Proposition~\ref{prop:P_C}, we have that the function $P_C(x)$ is twice continuously differentiable on $\setR$. It then follows that the derivative
$$\frac{d}{dy} \xs(y), \qquad y \in I_C,$$
is continuous, and hence the function $\xs(y)$, $y \in I_C$, is continuously differentiable on $I_C$. 

The result of the lemma then follows.
\end{proof}

Proposition~\ref{prop:xs} then follows directly from Lemma~\ref{lemma:xs} and Proposition~\ref{prop:optimal_production}.

\newpage
\section{Proof of Proposition~\ref{prop:Ds}}\label{app:Ds}
In this appendix we prove  Proposition~\ref{prop:Ds} which characterizes the  the function $ \Ds(y)$ given by
$$ \Ds(y) = || y - x^*(y)||, \qquad y \in I_C.$$
Recall that the function $\Ds(y)$  characterizes the absolute value of the ``displacement'' of the optimal content $x^*(y)$ that agent $y$ produces, and content $y$ that the agent is best at producing which is equal to content type $y$.

We first establish the result of Proposition~\ref{prop:Ds} that the function $\Ds(y)$, $y \in I_C$, is symmetric with respect to $mid(I_C)$.

\begin{lemma}\label{lemma:Ds_symmetric}
Let    $E_p$, $0< E_p \leq 1$, be a given constant, and let  $C=(C_d,C_s)$ be an interval community, i.e. we have that  $ C_d = C_s = I_C,$ where $I_C \subset \setR$ is an interval $\setR$.
Furthermore assume that
$$\alpha_C(y) = E_p, \qquad y \in I_C,$$
and
$$L_C < \frac{L}{2}.$$
Then function  
$$ \Ds(y) = ||y - x^*(y)||, \qquad y \in I_C,$$
where
$$x^*(y) = \arg\max_{x \in \setR} q(x|y) P_C(x), \qquad y \in I_C,$$
is symmetric with respect to $mid(I_C)$ on $I_C$ with
$$\Ds(mid(I_C)) = 0.$$
\end{lemma}

\begin{proof}
Note that in order to show that $\Ds(y)$  symmetric on $I_C$ with respect to $mid(I_C)$ for $y \in I_C$, we have to show that
$$ \Ds(mid(I_C) + s) =  \Ds(mid(I_C) - s), \qquad s \in [0,L_C].$$
Note that from Proposition~\ref{prop:optimal_production} we have that
$$ \Ds(y) = \arg \max_{s \in [0,y]} g(s) P_C(y-s), \qquad y \in [mid(I_C),mid(I_C)+L_C],$$
and
$$ \Ds(y) = \arg \max_{s \in [0,y]} g(s) P_C(y + s), \qquad y \in [mid(I_C) - L_C,mid(I_C)].$$
Setting 
$$F^{(+)}(y,s) = g(s) P_C(y-s), \qquad y \in [mid(I_C),mid(I_C)+L_C],$$
and
$$F^{(-)}(y,s) = g(s) P_C(y+s), \qquad y \in [mid(I_C)-L_C,mid(I_C)],$$
we have that 
$$ \Ds(y) = \arg \max_{s \in [0,y]} F^{(+)}(y,s), \qquad y \in [mid(I_C),mid(I_C)+L_C],$$
and
$$ \Ds(y) = \arg \max_{s \in [0,y]} F^{(-)}(y,s), \qquad y \in [mid(I_C)-L_C,mid(I_C)].$$
Therefore in order to show that $\Ds(y)$  symmetric with respect to $mid(I_C)$ for $y \in I_C$, it suffices to show that
$$F^{(+)}(mid(I_C) - x,s) = F^{(-)}(mid(I_C) + x,s), \qquad x \in [0,L_C], s \in [0,x].$$ 
Note that the conditions that $ x \in [0,L_C]$ and $s \in [0,x]$ imply that
$$ x-s \geq 0.$$
Using the the result from Proposition~\ref{prop:P_C} that the function  $P_C(x)$ is symmetric with respect to $mid(I_C)$, we then obtain that for $ x \in [0,L_C]$ and $s \in [0,x]$ that 
\begin{eqnarray*}
F^{(+)}(mid(I_C) - x,s) &=& g(s) P_C(mid(I_C) + x -s)  \\ 
&=&  g(s) P_C(mid(I_C) - x + s) \\
&=& F^{(-)}(mid(I_C) -x ,s).
\end{eqnarray*}
It then follows that $\Ds(y)$ is symmetric with respect to $mid(I_C)$ on $I_C$.

Finally, from Proposition~\ref{prop:optimal_production} which states that
$$x^*(mid(I_C)) = I_C,$$
we obtain
$$\Ds(mid(I_C)) = 0.$$
The result of the lemma then follows.
\end{proof}

\newpage
We next establish the result of Proposition~\ref{prop:Ds} that the function $\Ds(y)$ is strictly decreasing on  $[mid(I_C) - L_C, mid(I_C)]$, and strictly increasing on  $[mid(I_C), mid(I_C)]+L_C]$. In addition, we characterize the derivative of the function $\Ds(y)$, $y \in I_C$, with respect to $y$, and show that
$$ \left | \Dsy(y) \right | < 1, \qquad y \in I_C \backslash \{mid(I_C)\},$$
where
$$ \Dsy(y) =  \frac{d}{dy}  \Ds(y).$$

\begin{lemma}\label{lemma:Ds}
Let    $E_p$, $0< E_p \leq 1$, be a given constant, and let $C=(C_d,C_s)$ be an interval community, i.e. we have that  $ C_d = C_s = I_C,$ where $I_C \subset \setR$ is an interval in $\setR$. Furthermore  assume that
$$\alpha_C(y) = E_p, \qquad y \in I_C,$$
and
$$L_C < \frac{L}{2}.$$
Then for the function $\Ds(y)$ given by
$$ \Ds(y) = || y - x^*(y)||, \qquad y \in  I_C,$$
where
$$x^*(y)= \arg\max_{x \in \setR} q(x|y) P_C(x), \qquad y \in I_C,$$
we have that
\begin{enumerate}
\item[a)] $\Ds(y)$ is continuously differentiable on $I_C \backslash \{mid(I_C)\}$. 
\item[b)]  $\Ds(y)$ strictly decreasing on  $[mid(I_C) - L_C, mid(I_C)]$, and strictly increasing on  $[mid(I_C), mid(I_C)]+L_C]$.
\item[c)] the absolute value of the derivative $\frac{d}{dy}  \Ds(y)$ of the function $\Ds(y)$ with respect to $y$ is less than 1 on $I_C \backslash \{mid(I_C)\}$, i.e. we have that
$$ \left | \Dsy(y) \right | < 1, \qquad y \in I_C \backslash \{mid(I_C)\},$$
where
$$ \Dsy(y) =  \frac{d}{dy}  \Ds(y).$$
\end{enumerate}
\end{lemma}

\begin{proof}
Without loss of generality, we can assume that the interval $I_C$ is given by
$$ I_C = [-L_c,L_c].$$
By Lemma~\ref{lemma:Ds_symmetric} we have that the function $\Ds(y)$ is  symmetric with respect to $mid(I_C)$, it suffices to prove part the lemma for the case where
$$y \in [0,L_C].$$ 
We first prove part a) and b) of the lemma.

Recall that $x^*(y)$ is given by
$$x^*(y) = \arg\max_{x \in \setR} q(x|y) P_C(x), \qquad y \in I_C,$$
and that by Proposition~\ref{prop:optimal_production} the necessary and sufficient condition for $x^*(y)$, $y \in [0,L_C]$,  to be the solution to the above optimization problem is given by
$$  q'(x|y) P_C(x) +  q(x|y) P'_C(x) = 0, \qquad x \in [0,y].$$

By Proposition~\ref{prop:optimal_production} which states that
$$ \xs(y) \in [0,y] \cap  supp(q(\cdot|y)), \qquad y \in [0, L_C],$$
we can rewrite $x^*(y)$ for
$$y \in [0,L_C]$$
as
$$x^*(y) = y - \Ds(y).$$

Using the result from  Proposition~\ref{prop:optimal_production} which states that $\xs(y)$ is the unique solution to the optimization problem
$$\max_{x \in \setR} q(x|y) P_C(x), \qquad y \in I_C,$$
we obtain that  $\Ds(y)$ is the unique solution to the optimization problem 
$$\Ds(y) = \arg\max_{s \in [0,y] \cap supp(g)} g(s) P_C(y-s), \qquad y \in [0,L_C].$$
Furthermore, from Lemma~\ref{lemma:qP_C_concave} we obtain that the necessary and sufficient condition for $\Ds(y)$ to be the solution to this  optimization problem is given by
$$ g'(\Ds(y)) P_C(y - \Ds(y)) -  g(\Ds(y)) P'_C(y - \Ds(y)) = 0, \qquad \Ds(y) \in [0,y] \cap supp(g) .$$
Let the function $F(s,y)$ be given by
$$ F(s,y) = g'(s) P_C(y -s) -  g(s) P'_C(y -s), \qquad y \in [0,L_C], s \in [0,y] \cap supp(g).$$
Note that we have that 
$$F((\Ds(y),y) = 0, \qquad y \in [0,L_C].$$
Furthermore we have that
$$\frac{d}{ds} F(s,y) =  g''(s) P_C(y - s)  -  2g'(s) P'_C(y - s) + g(s) P''_C(y - s), \qquad y \in [0,L_C], s \in [0,y] \cap supp(g),$$
and
$$\frac{d}{dy} F(s,y) =  g'(s) P'_C(y - s) - g(s) P''_C(y - s), \qquad y \in [0,L_C], s \in [0,y] \cap supp(g).$$
As by Assumption~\ref{ass:fg}  we have $g(s)$ non-increasing and strictly concave in $s$ for $s \in [0,y] \cap supp(g)$, and
by Proposition~\ref{prop:P_C} we have that the function $P_C(x)$ is non-increasing and strictly concave for $x \in [0,L_C]$, it follows that for
$$ y \in [0,L_C] \mbox{ and } s \in [0,y] \cap supp(g)$$
we have that 
$$ g''(s) P_C(y - s) < 0,$$
$$ g'(s) P'_C(y - s) \geq 0,$$
and
$$ g(s) P''_C(y - s) < 0.$$
Combining the above results, it follows that
$$\frac{d}{ds} F(s,y) < 0 , \qquad y \in [0,L_C], s \in [0,y] \cap supp(g) $$
and
$$\frac{d}{dy} F(s,y) > 0 , \qquad y \in [0,L_C], s \in [0,y] \cap supp(g).$$
By Proposition~\ref{prop:optimal_production} we have that
$$\xs(y)) \in [0,y], \qquad y \in [0,L_C],$$
and it follows that
\begin{equation}\label{eqn:diff_s_F}
\frac{d}{ds} F(s=\Ds(y),y) < 0 , \qquad y \in [0,L_C],
\end{equation}
and
\begin{equation}\label{eqn:diff_y_F}
\frac{d}{dy} F(s=\Ds(y),y) > 0 , \qquad y \in [0,L_C].
\end{equation}
Using the implicit function theorem~\cite{rudin}, we obtain that the function  $\Ds(y)$ is continuous and differentiable on $[0,L_C]$ with 
\begin{eqnarray*}
\frac{d}{dy} \Ds(y) &=& \frac{- \frac{d}{dy} F(s=\Ds(y),y)}{\frac{d}{ds} F(s=\Ds(y),y)}, \qquad y \in [0,L_C].
\end{eqnarray*}
As by Assumption~\ref{ass:fg} we have that the functions $f$ is continuously differentiable, and the function $g$ is continuously differentiable on $supp(g)$, it follows that the function  $\Ds(y)$ is continuously differentiable on $[0,L_C]$. Furthermore, using  Eq.~\eqref{eqn:diff_s_F}~and~\eqref{eqn:diff_y_F} we obtain that
\begin{eqnarray*}
\frac{d}{dy} \Ds(y) &=& \frac{- \frac{d}{dy} F(s=\Ds(y),y)}{\frac{d}{ds} F(s=\Ds(y),y)} > 0, \qquad y \in [0,L_C],
\end{eqnarray*}
and the function $\Ds(y)$ is strictly increasing on $[0,L_C]$. As by Lemma~\ref{lemma:Ds_symmetric} the function $\Ds(y)$ is symmetric with respect to $mid(I_C)$ on $I_C$, it follows that function $\Ds(y)$ is  continuously differentiable and strictly decreasing on $[-L_C,0]$.

Note that the above argument only shows that the function $\Ds(y)$ is continuously differentiable on the interval $[0,L_C]$ and $[-L_C,0]$, but does not show that the derivative is continuous at $y=0$. Hence the above proof does not show that the derivative at $y=0$ exists, and we can only conclude that
$$\left | \frac{d}{dy} \Ds(y) \right |  < 1, \qquad y \in  I_C \backslash \{mid(I_C)\}.$$
This establishes part a) and b) of the lemma.

It remains to prove part c) of the lemma.
Recall that
$$\frac{d}{ds} F(s,y) =  g''(s) P_C(y - s)  -  g'(s) P'_C(y - s) - \frac{d}{dy} F(s,y), \qquad y \in [0,L_C], s \in [0,y] \cap supp(g),$$
and we can re-write the derivative of $  \Ds(y)$, $y \in [0,L_C]$, as
$$\frac{d}{dy} \Ds(y) = \frac{ - \frac{d}{dy} F(s=\Ds(y),y)}
{ g''( \Ds(y)) P_C(x^*(y))  -  g'(\Ds(y)) P'_C(x^*(y)) -  \frac{d}{dy} F(s=\Ds(y),y)}, \qquad y \in [0,L_C].$$
By Proposition~\ref{prop:optimal_production} we have that
$$\xs(y)) \in [0,y], \qquad y \in [0,L_C].$$
Furthermore, from the above results we have that
$$g''( \Ds(y)) P_C(x^*(y)) < 0, \qquad y \in [0,L_C],$$
and
$$  -  g'(\Ds(y)) P'_C(x^*(y)) \leq 0, \qquad y \in [0,L_C],$$
as well as
$$- \frac{d}{dy} F(s=\Ds(y),y) < 0 , \qquad y \in [0,L_C].$$
Combining these results we obtain for $y \in [0,L_C]$ that
$$ g''( \Ds(y)) P_C(x^*(y))  -  g'(\Ds(y)) P'_C(x^*(y)) -  \frac{d}{dy} F(s=\Ds(y),y) <  -  \frac{d}{dy} F(s=\Ds(y),y) < 0,$$
and hence we have that
$$ \left | \frac{d}{dy} \Ds(y) \right |  < 1,  \qquad y \in [0,L_C].$$
\end{proof}

\newpage

Using the result of Lemma~\ref{lemma:Ds}, we derive an additional result for the function $\xs(y)$, $y \in I_C$, which characterizes the image of the function $x^*(y)$.
\begin{lemma}\label{lemma:production_value_interval}
Let   $E_p$, $0< E_p \leq 1$, be a given constant, and let  $C=(C_d,C_s)$ be an interval community, i.e. we have that  $ C_d = C_s = I_C,$ where $I_C \subset \setR$ is an interval on $\setR$. Furthermore assume that
$$\alpha_C(y) = E_p, \qquad y \in C_d,$$
and
$$L_C < \frac{L}{2}.$$
Then the image of the function $x^*(y)$,
$$x^*(y) = \arg\max_{x \in \setR} q(x|y) P_C(x), \qquad y \in I_C,$$
is given by an interval $I^*_C \subset \setR$ such that
$$ mid(I^*_C) = mid(I_C)$$
and
$$L^*_C = \frac{|I^*_C|}{2} < L_C.$$
\end{lemma}

\begin{proof}
As the interval $I_C$ is a connected set and by Lemma~\ref{lemma:xs} the function $x^*(y)$ is continuous on $I_C$, it follows that the image of  $x^*(y)$, $y \in I_C$,  is a  connected set and hence given by  an interval $I^*_C$. That means that the closure of the image is given by a closed interval. 

By Proposition~\ref{prop:optimal_production}, we have that
$$x^*(y) = y + \Ds(y), \qquad y \in [mid(L_C)-L_C,mid(L_C)],$$
and
$$x^*(y) = y - \Ds(y), \qquad y \in [mid(L_C),mid(L_C) + L_C],$$
where $\Ds$ is given by
$$ \Ds(y) = || y - \xs(y)||, \qquad y \in I_C.$$
Furthermore by Lemma~\ref{lemma:xs}, we have that the function $\xs(y)$ is strictly increasing on $I_C$.  Therefore in order to prove the lemma, we have to show that
$$\Ds(mid(I_C) - L_C) =  \Ds(mid(I_C) + L_C) > 0,$$

Note that by Lemma~\ref{lemma:Ds} we have that function  $\Ds(y)$ is symmetric with respect to $mid(I_C)$, and it follows that
$$ mid(I^*_C) = mid(I_C)$$
and
$$\Ds(mid(I_C) - L_C) =  \Ds(mid(I_C) + L_C).$$
Therefore, in order to prove the lemma it remains to show that
$$ \Ds(mid(I_C) + L_C) > 0.$$
To obtain this result, note that by Lemma~\ref{lemma:Ds} the function  $\Ds(y)$ is strictly increasing on $[mid(I_C), mid(I_C) + L_C]$,  and by  Lemma~\ref{lemma:Ds_symmetric} we have that
$$\Ds(mid(I_C)) = 0.$$
It then follows that
$$  \Ds(mid(I_C) + L_C) > 0,$$
which completes the proof of the lemma.
\end{proof}

\newpage
\section{Proof of Proposition~\ref{prop:Q_C}}\label{app:Q_C}
In this appendix we prove Proposition~\ref{prop:Q_C}. To do this, we characterize in the next subsection the density $\bsC(x)$ of the rate at which content of type $x$ is being produced in an interval community.


\subsection{Properties of $\bsC(x)$}
Having characterized in  Proposition~\ref{prop:optimal_production} the optimal content production rate $\bsC(\cdot|y)$ of an agent in a given interval community $C=(I_C,I_C)$, we characterize in this subsection  the resulting density $\bsC(x)$ of the rate at which of content of type $x$ that is being produced in the community $C$, i.e. the density $\bsC(x)$ given by
$$ \bsC(x) = \int_{I_C} \bsC(x|y) dy, \qquad x \in \setR.$$

We use the following definitions. Recall that by Lemma~\ref{lemma:production_value_interval} we have for a given interval community $C=(I_C,I_C)$, that the closure $I^*_C$ of the image of the function $\xs(y)$, $y \in I_C$ is given by an interval $I^*_C \subset \setR$ such that
$$ mid(I^*_C) = mid(I_C)$$
and
$$L^*_C = \frac{|I^*_C|}{2} < L_C.$$
Using $I^*_C$, let the function 
$$ y = s(x), \qquad x \in I^*_C,$$
be such that
$$ x^*(y) = x,$$
i.e. $s(x)$ identifies the agent $y \in I_C$ for which it is optimal to produce content type $x$.

Finally, let
$$\xsy(y) = \frac{d}{dy} \xs(y), \qquad y \in I_C,$$
be the derivative of the function $\xs(y)$ with respect to $y$. Note that by Lemma~\ref{lemma:xs} this derivative is well-defined.

We then have the following two results for $\bsC(x)$, $x \in I^*_C$. 
\begin{lemma}\label{lemma:bsC}
Let    $E_p$, $0< E_p \leq 1$, and $E_q$, $0< E_q$, be given constants, and let  $C=(C_d,C_s)$ be an interval community, i.e. we have that  $ C_d = C_s = I_C,$ where $I_C \subset \setR$ is an interval in $\setR$. Furthermore assume that
$$\alpha_C(y) = E_p, \qquad y \in I_C$$
and
$$L_C < \frac{L}{2},$$
as well as that
$$\max_{x \in \setR} \Bsbl q(x|y) P_C(x) - \alpha_Cc \Bsbr > 0, \qquad y \in I_C,$$
where $\alpha_C = 2E_pL_C$.
Assuming that
$$ \bC(y) = E_q, \qquad y \in I_C,$$
where $\bC(y)$ is the bound on the total content production rate of agent $y$ in community $C$, 
let $\bsC(\cdot|y)$ be the optimal content production rate of agent $y \in I_C$ as given in Proposition~\ref{prop:optimal_production}, i.e. we have that
$$\bsC(x|y) = E_q \delta(x-x^*(y)), \qquad y \in I_C, x \in \setR,$$
where
$$\xs(y)= \arg\max_{x \in \setR} q(x|y) P_C(x), \qquad y \in I_C.$$
Then the density
$$\bsC(x) = \int_{I_C} \bC(x|y) dy, \qquad x \in \setR,$$
of the rate at which of content of type $x$ that is being generated in community $C$ under the optimal contention production allocations $\bsC(\cdot|y)$, $y \in I_C$ is given by
$$\bsC(x) = \frac{1}{\xsy(s(x))} E_q, \qquad x \in I^*_C,$$
where $I^*_C$ is the image of the function $\xs(y)$, $y \in I_C$.
\end{lemma}

\begin{proof}
By Proposition~\ref{prop:xs} with have that the function $x^*(y)$ is strictly increasing on $I_C$, and hence it is a bijection on $I_C$. It then follows that the inverse function $s(x)$ of $x^*(y)$ is well defined on $I_C^*$.
Furthermore by Lemma~\ref{lemma:production_value_interval} we have that the closure $I^*_C$ of the  image  of $\xs(y)$, $y \in I_C$, is given
$$I^*_C = [mid(I_C) - L^*_C, mid(I_C) - L^*_C ]$$
where
$$ 0 < L_C^* < L_C.$$
Without loss of generality, we can assume that $mid(I_C) = 0$ and we have that
$$I_C = [- L_C,  L_C ]$$
and
$$I^*_C = [- L^*_C, L^*_C ].$$
As we have that
$$\alpha_C(y) = E_p > 0, \qquad y \in C_d$$
and
$$\int_{\setR} \bsC(x|y) dx = E_q > 0,   \qquad y \in C_d,$$
as well as
$$ \max_{x \in \setR} \Bsbl q(x|y) P_C(x) - \alpha_Cc \Bsbr > 0, \qquad y \in I_C,$$
it follows from Proposition~\ref{prop:optimal_production} that the optimal production rate allocation for an agent $y \in I_C$ is given by
$$\beta^*_C(x|y) = E_q \delta(x-x^*(y)), \qquad x \in \setR,$$
where $d$ is the Dirac delta function, $x^*(y)$ is the unique solution to the optimization problem
$$\max_{x \in \setR}  q(x|y) P_C(x).$$
Combing these results with the fact that $\xs(y)$ is strictly increasing on $I_C$,  we obtain that
$$\int_{-L^*_C}^x \bsC(\tau) d\tau = \int_{-L_C}^{s(x)} E_q dy$$
and
$$ \bsC(x)  = \frac{d}{dx} \int_{-L_C}^{s(x)} E_q dy = E_q \frac{d}{dx} s(x).$$
Using well-known results for the derivative of a function~\cite{rudin}, the density $\bsC(x)$ of the rate at which content of type $x$ that is being generated in community $C$ under the optimal contention production allocations $\bsC(\cdot|y)$, $y \in I_C$, is then given by
$$\bsC(x) = E_q \frac{1}{\xsy(y=s(x))}, \qquad x \in I^*_C,$$
where $\xsy(y=s(x))$ is the derivative of the function $x^*(y)$ with respect to $y$ evaluated for  $y = s(x)$. Note that by Proposition~\ref{prop:xs} the derivative $\xsy(y)$, $y \in I_C$, is well-defined.\\

This completes the proof of the lemma.
\end{proof}

\newpage

\begin{lemma}\label{lemma:bsC2}
Let   $E_p$, $0< E_p \leq 1$, and $E_q$, $0< E_q$, be given constants, and let  $C=(C_d,C_s)$ be an interval community, i.e. we have that  $ C_d = C_s = I_C,$ where $I_C \subset \setR$ is an interval in $\setR$. Furthermore assume that
$$\alpha_C(y) = E_p, \qquad y \in I_C$$
and
$$L_C < \frac{L}{2},$$
as well as that
$$\max_{x \in \setR} \Bsbl q(x|y) P_C(x) - \alpha_Cc \Bsbr > 0, \qquad y \in I_C,$$
where $\alpha_C = 2E_pL_C$.
Assuming that
$$ \bC(y) = E_q, \qquad y \in I_C,$$
where $\bC(y)$ is the bound on the total content production rate of agent $y$ in community $C$, 
let $\bsC(\cdot|y)$ be the optimal content production rate of agent $y \in I_C$ as given in Proposition~\ref{prop:optimal_production}, i.e. we have that
$$\bsC(x|y) = E_q \delta(x-x^*(y)), \qquad y \in I_C, x \in \setR,$$
where
$$x^*(y)= \arg\max_{x \in \setR} q(x|y) P_C(x), \qquad y \in I_C,$$
and let $I^*_C$ be the image of the function $\xs(y)$, $y \in I_C$.

Then the density
$$\bsC(x) = \int_{I_C} \bC(x|y) dy, \qquad x \in \setR,$$
of the rate at which of content of type $x$ that is being generated in community $C$ under the optimal contention production allocations $\bsC(\cdot|y)$, $y \in I_C$,  is continuous, and symmetric with respect to $mid(I_C)$, on $I^*_C$.

Furthermore, we have that
$$ \bsC(x) > E_q, \qquad x \in I^*_C.$$
\end{lemma}

\begin{proof}
Without loss of generality, we can assume that the intervals $I_C$ and $I^*_C$ are given by
$$ I_C = [ - L_C, L_C]$$
and
$$ I^*_C = [ - L^*_C, L^*_C].$$

We first show that the function $\bsC(x)$, $x \in I^*_C$, is symmetric with respect to $mid(I_C)$. 
Let the function $\Ds(y)$ be given by
$$\Ds(y) = | | y - \xs(y)||, \qquad y \in I_C.$$
By Lemma~\ref{lemma:Ds_symmetric} we have that the function $\Ds(y)$ is symmetric with respect to $mid(I_C)=0$ on $I_C$, as well as  strictly decreasing on $[ -L_C,0)$ and strictly increasing on $(0,L_C]$ with
$$ \Ds(0) = 0.$$
It then follows that for the derivative
$$  \Dsy(y) = \frac{d}{dy} \Ds(y), \qquad y \in  I_C\backslash\{ mid(I_C) \}$$
we have that
$$  \Dsy(-y) = -  \Dsy(y), \qquad y \in  I_C\backslash\{ mid(I_C) \}.$$
Using the definition of $\Ds(y)$, we obtain that
$$\xsy(y) = \left \{
\begin{array}{ll}
 1 -  \Dsy(y), & y \in [0, L_C], \\
 1 +  \Dsy(y), & y \in [- L_C, 0].
\end{array} \right .$$
Therefore in order to show that $\bsC(x)$ is symmetric with respect to $mid(I_C)$ on $I^*_C$, we have to show that
$$s(-x) = - s(x), \qquad x \in I^*_C.$$
To do this, recall that
$$ x^*(y) = \left \{
\begin{array}{ll}
y - \Ds(y), & y \in [0,L_C],\\
y + \Ds(y), & y \in [-L_C,0].
\end{array}
\right .$$
Using the result that $\Ds(y)$ is symmetric with respect to $mid(I_C)$ on $I_C$, we have that
$$ - x^*(y) = - y  + \Ds(y) = - y  + \Ds( - y) = x^*(-y), \qquad y \in [0,L_C].$$
It then follows that
$$  - x^*(y) = x^*(-y), \qquad y \in I_C$$
and
$$ s(-x^*(y)) = s(x^*(-y)) = -y, \qquad y  \in I_C.$$
As
$$ s(x^*(y)) = y,$$
we obtain  that
$$ s(-x) = - s(x), \qquad x \in I_C^*.$$
This proves  that $\bsC(x)$ is symmetric with respect to $mid(I_C)$ on $I^*_C$.

We next show show that the function $\bsC(x)$ is continuous on $I^*_C$. Recall that we have that
$$\bsC(x) = \frac{1}{\xsy(s(x))}E_q, \qquad x \in I^*_C,$$
where $\xsy(y)$, $y \in I_C$, is the derivative of the function
$$\xs(y) =  \arg\max_{x \in \setR} q(x|y) P_C(x), \qquad y \in I_C.$$
By Proposition~\ref{lemma:xs} we have that the function $\xs(y)$ is continuously differentiable on $I_C$ and that we have that
$$\xsy(y) > 0, \qquad y \in I_C.$$
This result implies that the function $\bsC(x)$ is continuous on $I^*_C$.

It remains to show that
$$ \bsC(x) > E_q, , \qquad x \in I^*_C.$$
From Lemma~\ref{lemma:Ds} we have that the function $\Ds(y)$ is  strictly decreasing on $[ -L_C,0)$ and strictly increasing on $(0,L_C]$, and
$$ | \Dsy(y)| < 1, \qquad y \in I_C\backslash\{mid(I_C)\}.$$
It then follows that
$$ \frac{1}{1 -  \Dsy(s(x))} > 1, \qquad x \in (mid(I_C), mid(I_C) + L^*_C]$$
and
$$\frac{1}{1 +  \Dsy(s(x))} > 1, \qquad x \in [mid(I_C) - L^*_C, mid(I_C)),$$
and we obtain that
$$ \bsC(x) > E_q,  \qquad x \in I^*_C\backslash\{mid(I_C)\}.$$

It remains to show that for
$$y_0 = mid(I_C)$$
we have that
$$\bsC(y_0) = \frac{1}{\xsy(s(y_0)}E_q > E_q,$$
or
$$\xsy(s(y_0)) < 1.$$
From the proof of Lemma~\ref{lemma:xs}, we have that
$$ \xsy(s(y_0)) = \frac{g''(0)P_C(0)}{g''(0)P_C(0) + g(0)P_C''(0)}.$$
From Assumption~\ref{ass:fg} we have that
$$g(0) > 0 \mbox{ and } g''(0) > 0,$$
and from Proposition~\ref{prop:P_C} we have that
$$P_C(0) > 0 \mbox{ and } P_C''(0) > 0.$$
Combining these results, it follows that
$$\xsy(s(y_0)) < 1.$$
This completes the proof of the lemma.
\end{proof}

\newpage
\subsection{Properties of $Q^*_C(x)$}
Having studied in the previous subsection the optimal content production rate allocation for a given interval community  $C=(I_C,I_C)$, we prove in this subsection the results of Proposition~\ref{prop:Q_C} which characterizes the  properties  of the content supply function $Q^*_C(x)$ given
$$ Q^*_C(x) = \int_{y \in I_C}  \beta^*_C(x|y) q(x|y) dy, \qquad x \in \setR,$$
where $\bsC(\cdot|y)$ is the optimal content production rate of agent $y$.

Recall that by Lemma~\ref{lemma:production_value_interval} we have for a given interval community $C=(I_C,I_C)$, that the image $I^*_C$ of the function $\xs(y)$, $y \in I_C$ is given by an interval
$$I^*_C = [mid(I_C) - L^*_C, mid(I_C) + L^*_C ] \subset I_C.$$
Furthermore, recall  that the function $s(x)$, $x \in I^*_C$,  is such that for $y = s(x)$ we have that
$$x^*(y) = x.$$

We have the following result.
\begin{lemma}\label{lemma:Q_C}
Let   $E_p$, $0< E_p \leq 1$, and $E_q$, $0< E_q$, be given constants, and let  $C=(C_d,C_s)$ be an interval community, i.e. we have that  $ C_d = C_s = I_C,$ where $I_C \subset \setR$ is an interval in $\setR$. Furthermore assume that
$$\alpha_C(y) = E_p, \qquad y \in I_C$$
and
$$L_C < \frac{L}{2},$$
as well as that
$$\max_{x \in \setR} \Bsbl q(x|y) P_C(x) - \alpha_Cc \Bsbr > 0, \qquad y \in I_C,$$
where
$$\alpha_C = 2E_pL_C.$$
Assuming that
$$ \bC(y) = E_q, \qquad y \in I_C,$$
where $\bC(y)$ is the bound on the total content production rate of agent $y$ in community $C$, 
let $\beta^*(\cdot|y)$ be the optimal content production rate of agent $y \in I_C$ as given in Proposition~\ref{prop:optimal_production}, i.e. we have that
$$\beta^*_C(x|y) = E_q \delta(x-x^*(y)), \qquad y \in I_C, x \in \setR,$$
where
$$x^*(y)= \arg\max_{x \in \setR} q(x|y) P_C(x), \qquad y \in I_C,$$
and let $I^*_C$ be the image of the function $\xs(y)$, $y \in I_C$.

The optimal content supply function $Q^*_C(x)$ is given by
  $$Q^*_C(x) = E_q \int_{y \in I_C}  \delta(x - \xs(y)) q(\xs(y)|y) dy, \qquad x \in \setR, y \in I_C,$$
and we have that
\begin{enumerate}
\item[a)] the support $supp(Q_C^*(x))$ of $Q^*_C(x)$, $x \in \setR$  is given by
 $$supp(Q_C^*(x)) =  I^*_C = [mid(I_C) - L^*_C, mid(I_C) + L^*_C]$$
 where
 $$ 0 < L^*_C < L_C.$$
\item [b)] the function $\QsC(x)$ is given by
$$Q^*_C(x) = \left \{
\begin{array}{lr}
\bsC(x) q(x|s(x)), &x  \in I^*_C,\\
0, & x \notin I^*_C,
\end{array} \right .$$
where
$$\bsC(x)  = \int_{I_C} \bsC(x|y) dy, \qquad x \in I^*_C.$$
\item[c)] the function $Q^*_C(x)$ is continuous, and symmetric with respect to $mid(I_C)$, on $I^*_C$.
\end{enumerate}
\end{lemma}
Note that Proposition~\ref{prop:Q_C} follows directly from Lemma~\ref{lemma:Q_C}  as the communities under a Nash equilibrium as given by Proposition~\ref{prop:nash} satisfy the assumptions made in Lemma~\ref{lemma:P_C}. In particular as for an interval community $C = (I_C,I_C) \in \setC^*$ under a Nash equilibrium as given by Proposition~\ref{prop:nash} we have that 
$$\UsC(y) > 0, \qquad y \in C_s = I_C,$$
and
$$\alpha_C(y) = E_p > 0, \qquad y \in C_d = I_C,$$
it follows that
$$\max_{x \in \setR} \Bsbl q(x|y) P_C(x) - \alpha_Cc \Bsbr > 0, \qquad y \in I_C,$$
where $\alpha_C = 2E_pL_C$.

Furthermore, by Proposition~\ref{prop:optimal_production} we have for an interval community $C = (I_C,I_C) \in \setC^*$ under a Nash equilibrium as given by Proposition~\ref{prop:nash} that
$$\int_{\setR} \bsC(x|y) dx = E_q > 0,   \qquad y \in C_s = I_C.$$

\begin{proof}
Note that by Proposition~\ref{prop:optimal_production}, we have that
$$\bsC(x|y) =  E_q\delta(x - \xs(y)), \qquad x \in \setR, y \in I_C,$$
where $\delta(\cdot)$ is the Dirac delta function, and
$$x^*(y)= \arg\max_{x \in \setR} q(x|y) P_C(x), \qquad y \in I_C.$$
Using the result that
$$  E_q\delta(x - \xs(y)) q(x|y) =   E_q\delta(x - \xs(y)) q(\xs(y)|y) , \qquad x \in \setR, y \in I_C,$$
we obtain the result that
 $$Q^*_C(x) = \int_{y \in I_C} \bsC(x|y) q(x|y) dy = 
 E_q \int_{y \in I_C}  \delta(x - \xs(y)) q(\xs(y)|y) dy, \qquad x \in \setR.$$

Let $I^*_C$ be the image of the function $\xs(y)$, $y \in I_C$. By Lemma~\ref{lemma:production_value_interval} we have that $I^*_C$ is given by an interval in $\setR$, such that
$$ mid(I^*_C) = mid(I_C)$$
and the half-length $L^*_C$ of the interval $I^*_C$ is strictly smaller than $L_C$, i.e. we have that 
$$L^*_C = \frac{|I^*_C|}{2} < L_C.$$
Combining this result with fact that 
$$\beta^*_C(x|y) = E_q \delta(x-x^*(y)), \qquad y \in I_C, x \in \setR,$$
and Proposition~\ref{prop:optimal_production} which states that
$$q(x^*(y)|y) > 0, \qquad y \in I_C,$$
we obtain that
 $$supp(Q_C^*(x)) = I^*_C = [mid(I_C) - L^*_C, mid(I_C) + L^*_C]$$
 with
 $$L^*_C < L_C.$$

Next we show that
$$Q^*_C(x) = \bsC(x) q(x|s(x)), \qquad x \in I^*_C.$$
Using the same argument as given in the proof of Lemma~\ref{lemma:bsC}, we have that
$$\int_{-L^*_C}^x Q^*_C(\tau) d\tau = \int_{-L_C}^{s(x)} E_q q(\xs(y)|y) dy$$
and
$$ Q^*_C(x)  = \frac{d}{dx} \int_{-L_C}^{s(x)} E_q q(\xs(y)|y) dy  = E_q q(x|s(x))\frac{d}{dx} s(x).$$
As we have that
$$\bsC(x) =  E_q q(x|s(x))\frac{d}{dx} s(x), \qquad x \in I^*_C,$$
it follows that
$$Q^*_C(x) = \bsC(x) q(x|s(x)), \qquad x \in I^*_C,$$
and
$$Q^*_C(x) = 0, \qquad x \notin I^*_C,$$

Next we show that the function  $Q^*_C(x)$ is continuous on $I^*_C$. As we have that
$$Q^*_C(x) = \bsC(x) q(x|s(x)), \qquad x \in I^*_C,$$
it follows that $Q^*_C(x)$ is continuous on $I^*_C$ if we can show that the functions $\bsC(x)$ and $q(x|s(x)$ are continuous on $I^*_C$. By Lemma~\ref{lemma:bsC2} we have that the function  $\bsC(x)$ is continuous on $I^*_C$. By Assumption~\ref{ass:fg}, we have that the function $g(x)$, $x \in [0,L]$, is continuous in $x$. From Lemma~\ref{lemma:xs} which states that the function $\xs(y)$, $y \in I_C$, is continuous in $y$, we have that the function $s(x)$, $x \in I^*_C$, is continuous in $x$. Finally, by Assumption~\ref{ass:fg} we have that the function $g(x)$, $x \in [0,L]$, is contours in $x$, and it follows that the function
$$q(x|s(x)) = g(|x - s(x)|), \qquad x \in I^*_C,$$
is continuous in $x$. This establishes that the function $Q^*_C(x)$ in continuous on $I^*_C$.

It remains to show that  the function $Q^*_C(x)$ is symmetric with respect to $mid(I_C)$ on $I^*_C$. By Lemma~\ref{lemma:bsC2}, we have that $\bsC(x)$ is symmetric with respect to $mid(I_C)$ on $I^*_C)$, therefore in order to show that $Q^*_C(x)$ is symmetric with respect to $mid(I_C)$ on $I^*_C$, it suffices to show that $q(x|s(x))$ is symmetric with respect to $mid(I_C)$ on $I^*_C$. To do this, without loss of generality we can assume that the interval $I_C$ is given by
$$ I_C = [-L_C,L_C]$$
and 
 the interval $I^*_C$ is given by
$$ I^*_C = [-L^*_C,L^*_C],$$
where
$$L_C^* < L_C.$$
Recall that
$$ q(x|s(x)) = g ( \Ds(s(x))).$$
Therefore in order so show that $q(x|s(x))$ is symmetric with respect to $mid(I_C) = 0$ on $I^*_C$, it suffices to show that
$$ g ( \Ds(s(-x))) =  g ( \Ds(s(x))), \qquad x \in I^*_C.$$
By Lemma~\ref{lemma:Ds_symmetric} the function $\Ds(y)$ is symmetric with respect to $mid(I_C) = 0$ on $I_C$, and from the proof of Lemma~\ref{lemma:bsC2} we have that
$$s(-x) = - s(x), \qquad x \in I^*_C.$$
It then follows that
$$ g ( \Ds(s(-x))) =  g ( \Ds(-s(x))) =   g ( \Ds(s(x))),  \qquad x \in I^*_C,$$
and we obtain that  $Q^*_C(x)$ is symmetric with respect to $mid(I_C) = 0$ on $I^*_C$.\\

This completes the proof of the lemma.
\end{proof}

\newpage
\section{Proof of Proposition~\ref{prop:UdC}}\label{app:UdC}
In this appendix we prove Proposition~\ref{prop:UdC} which characterize the utility rates $U_C^{(d)}(y)$ for content consumption under a Nash Equilibrium $\commStrucNash$ as given by Proposition~\ref{prop:nash}. 

To do this  we consider the case where we are given  a set $C_s = I_C$ of content producers where $I_C \subset \setR$ is an interval in  $\setR$. Note that this situation applies to the Nash equilibrium given by Proposition~\ref{prop:nash} where all communities are given by interval communities, i.e. the set of agents that produce (and consume) content in a community $C$ is given by an interval  $I_C$ in $\setR$. For this situation, we assume that we are given the content supply function $Q_C(x)$, $x \in \setR$, for the content produced by agents $y \in C_s = I_C$, and then characterize then the maximal utility of content consumption that a agent $y \in \setR$ can obtain under the content supply function $Q_C(x)$. 

Our first result characterize the maximal content consumption utility function $U_C^{(d)}(y)$, $y \in \setR$, under the additional assumption that the content production rates $\bC(x|y)$ of the agents in the set $C_s = I_C$ have the property that
$$\int_{\setR} \bC(x|y) dx = E_q > 0.$$
For this case, we have the following result. 
\begin{lemma}\label{lemma:UdC_FdC}
Let $E_q$, $0<E_q$, be a given constant, and let  $C_s = I_C$ be a set of content producers where $I_C \subset \setR$ is an interval in  $\setR$.
Assume that
$$\int_{\setR} \bC(x|y) dx = E_q > 0, \qquad y \in I_C,$$
where $\bC(x|y)$, $y \in I_C$ and $x \in \setR$, is the rate allocation of an agent $y \in I_C$.
Then the maximal utility rate for content consumption that an agent $y \in \setR$ obtains from the producers $y \in I_C$ is given by
$$U_C^{(d)}(y) = \max \Bsl 0,\FdC(y) \Bsr, \qquad y \in \setR,$$
where 
$$\FdC(y) =  E_p  \int_{\setR} \Bsbl p(x|y) Q_C(x) dx - 2L_CE_q c   \Bsbr dx, \qquad y \in \setR,$$
and
$E_p$, $0< E_p \leq 1$, is the maximal fraction of time that agent $y$ can allocate to consuming content from the producers $y \in C_s = I_C$, and 
$$Q_C(x) = \int_{I_C} \bC(x|y) dy, \qquad  x \in \setR.$$
\end{lemma}

\begin{proof}
The result of this lemma follows directly from the definition of the utility rate for content consumption of an agent $y \in \setR$ given by
$$\UdC(y) = \alpha_{C}(y) \int_{\setR} \Big [ Q_C(x) p(x|y) - \beta_C(x) c \Big ] dx,$$
and the assumptions that
$$ 0 \leq \alpha_{C}(y) \leq E_p$$
and
$$\int_{\setR} \bC(x|y) dx = E_q > 0, \qquad y \in I_C.$$
In particular, under the assumption that
$$\int_{\setR} \bC(x|y) dx = E_q > 0, \qquad y \in I_C,$$
we have that
$$  \int_{\setR} \beta_C(x) = \int_{y \in I_C} \int_{\setR} \bC(x|y) dxdy = 2 L_C E_q.$$
\end{proof}

Using the result of this lemma, we first characterize in the next the properties of the function $\FdC(y)$ which we then use then to characterize the properties of the utility rates for content consumption $\UdC(y)$, and to prove  Proposition~\ref{prop:UdC}.

\newpage
\subsection{Properties of the function  $\FdC(y)$}
In this subsection, we characterize the properties of the function $\FdC(y)$, $y \in \setR$, for the special case where we have that
$$Q_C(x) = Q^*_C(x), \qquad x \in \setR,$$
and $Q^*_C(x)$ is as given in Proposition~\ref{prop:Q_C}, i.e. we assume that the content supply function $Q^*_C(x)$, $x \in \setR$, has the same properties as the optimal supply function for an interval community in a Nash equilibrium as given by Proposition~\ref{prop:nash}.

For this case, we have the following result.
\begin{lemma}\label{lemma:FdC}
Let $E_p$, $0<E_p\leq 1$, and  $E_q$, $0<E_q$, be given constants, and let $I_C$ be an interval in $\setR$.
For
$$ L_C = \frac{|I_C|}{2},$$
assume that
$$2 L_C < \min \{ b, L \}$$
where
$b$ is the constant of Assumption~\ref{ass:fg}.
Furthermore, assume for the function $Q^*_C(x)$, $x \in \setR$, that
\begin{enumerate}
\item $\QsC(x)$ is bounded on $\setR$. 
\item the support $supp(Q_C^*)$ of $Q^*_C(x)$  is given by
 $$supp(Q_C^*) = [mid(I_C) - L^*_C, mid(I_C) + L^*_C]$$
 where
 $$ 0 < L^*_C < L_C,$$
\item $Q^*_C(x)$ is non-negative, continuous, and symmetric with respect to $mid(I_C)$, on $(mid(I_C) - L^*_C, mid(I_C) + L^*_C)$. 
\end{enumerate}
Then  the function $F^{(d)}_C(y)$ given by
$$\FdC(y) =  E_p \int_{\setR} \Bsbl p(x|y) Q^*_C(x) dx - 2L_CE_q c   \Bsbr dx, \qquad y \in \setR,$$
has the properties that
\begin{enumerate}
\item[a)]  $\FdC(y)$ is symmetric with respect to $y_0 = mid(I_C)$ on $\setR$, and we have for $y=mid(I_C)$ that
$$\frac{d}{dy} \FdC(y) = 0.$$
\item[b)]  $\FdC(y)$ is strictly increasing on $(mid(I_C) - L,mid(I_C))$, and we have that
$$\frac{d}{dy} \FdC(y) > 0, \qquad y \in (mid(I_C) - L,mid(I_C)).$$
\item[c)]  $\FdC(y)$ is strictly decreasing on $(mid(I_C), mid(I_C) + L)$, and we have that
$$\frac{d}{dy} \FdC(y) < 0, \qquad y \in  (mid(I_C), mid(I_C) + L).$$
\end{enumerate}
\end{lemma}

\begin{proof}
Let
$$I^*_C =  [mid(I_C) - L^*_C, mid(I_C) + L^*_C].$$
Note that the assumptions on the function $Q^*_C(x)$, $x \in \setR$, made in the lemma imply that
$$ 0 < \int_{I^*_C} Q^*_C(x) dx. $$

Without loss of  generality, we assume for our proof that the function $Q^*_C(x)$ is non-negative, continuous, and symmetric with respect to $mid(I_C)$, on $I^*_C = [mid(I_C) - L^*_C, mid(I_C) + L^*_C]$

We first show that  $\FdC(y)$ is symmetric with respect to $y_0 = mid(I_C)$ on $\setR$. Note that we have that
$$ p(mid(I_C) - x|mid(I_C) -y) =  p(mid(I_C) + x|mid(I_C)  +y) = f(||x-y||), \qquad x,y \in \setR.$$
Combining this result with the assumption that $Q^*_C(x)$ is symmetric with respect to $mid(I_C)$ on $(mid(I_C) - L^*_C, mid(I_C) + L^*_C)$,  and letting
$$y_0 = mid(I_C),$$
we obtain 
\begin{eqnarray*}
\FdC(y_0 + y) &=&  E_p \int_{\setR} \Bsbl p(x|y_0+y) Q^*_C(x) dx - 2L_CE_q c \Bsbr dx \\
&=&  E_p \int_{-L^*_C}^{L^*_C} \Bsbl p(y_0 + x|y_0+y) Q^*_C(y_0+x) dx - 2L_CE_q c  \Bsbr dx   \\
&=&  E_p \int_{-L^*_C}^{L^*_C} \Bsbl p(y_0 - x|y_0-y) Q^*_C(y_0-x) dx - 2L_CE_q c  \Bsbr dx  \\
 &=&  E_p \int_{\setR} \Bsbl p(x|y_0-y) Q^*_C(x) dx - 2L_CE_q c  \Bsbr dx    \\
 &=& \FdC(y_0 - y), \qquad y \in [0,L].
\end{eqnarray*}
It then follows that  $\FdC(y)$ is symmetric with respect to $y_0 = mid(I_C)$ on $\setR$.

Next we prove that the function $\FdC(y)$ is strictly increasing on $(mid(I_C) - L,mid(I_C))$, and strictly decreasing on $(mid(I_C), mid(I_C) + L)$.
As $\FdC(y)$ is  symmetric with respect to $y_0 = mid(I_C)$ on $\setR$, it suffices to show that  $\FdC(y)$ is strictly increasing on $(mid(I_C)-L, mid(I_C))$, and we have that
$$\frac{d}{dy} \FdC(y) > 0, \qquad y \in  (mid(I_C)-L, mid(I_C)).$$

Therefore, we assume for the remainder of the proof that
$$ y < mid(I_C).$$

We have that 
$$\frac{d}{dy} \FdC(y) = \int_{I^*_C} Q^*_C(x) \frac{d}{dy} p(x|y) dx,$$
where $I^*_C$ is the support of $Q^*_C(x)$, $x \in \setR$. Recall that by assumption we have that
$$I^*_C = [mid(I_C) - L^*_C, mid(I_C) + L^*_C]$$
 where
 $$ 0 < L^*_C < L_C < \frac{L}{2}.$$
In the following, we consider three separate cases.

\newpage
The first case that we consider is where we have that
$$ y \in [mid(I_C) + L^*_C - L, mid(I_C) -  L^*_C).$$
As we have that
 $$ 0 < L^*_C < L_C < \frac{L}{2},$$
 we obtain 
 $$ mid(I_C) + L^*_C - L < mid(I_C) -  L^*_C.$$
For this case we have that
$$ \frac{d}{dy} p(x|y) > 0$$
as by  Assumption~\ref{ass:fg}  the function $f$ is strictly decreasing on $[0,L]$. Combining this result we the fact that
$$ \int_{I^*_C} Q^*_C(x) dx  > 0, $$
it follows that
$$ \frac{d}{dy}   \FdC(y) =   \int_{I^*_C} Q^*_C(x) \frac{d}{dy} p(x|y) dx < 0, \qquad y \in [mid(I_C) + L^*_C - L, mid(I_C) -  L^*_C).$$

\newpage
Next we consider the case where 
$$ y \in [mid(I_C) - L^*_C,mid(I_C)) = [y_0 - L^*_C,y_0).$$
For this case we obtain that
\begin{eqnarray*}
\frac{d}{dy} \FdC(y)
&=& \int_{I^*_C} Q^*_C(x) \frac{d}{dy} p(x|y) dx \\
&=& \int_{y_0-L^*_C}^{y} Q^*_C(x) f'(y-x) dx - \cdots \\
&& - \int_0^{y_0-y} Q^*_C(y_0+s) \Big (f'(y_0 -y +s)) + f'(y_0 - y - s) \Big ) ds - \cdots \\
&& -  \int_{-y}^{y_0+L^*_C} Q^*_C(x) f'(x-y) dx \\
&=& - \int_{0}^{y_0-y} Q^*_C(y_0+s) \Big (f'(y-y_0 -s)) + f'(y-y_0 + s) \Big ) ds + \cdots \\
&&  + \int_0^{y-y_0+L^*_C} Q^*_C(y-s) \Big ( f'(s) - f'(2(y_0-y)+s) \Big ) ds.
\end{eqnarray*}
By Assumption~\ref{ass:fg} the function $f$ is strictly decreasing on $[0,L]$, and it follows that
$$f'(y-y_0 +s)) + f'(y-y_0 - s) < 0, \qquad s \in [0,y - y_0],$$
and
$$ - \int_0^{y - y_0} Q^*_C(y_0+s) \Big (f'(y-y_0 +s)) + f'(y-y_0 - s) \Big ) ds > 0.$$
Next we show that
$$ \int_0^{y-y_0+L^*_C} Q^*_C(s-y) \Big ( f'(s) - f'(2(y_0-y)+s) \Big ) ds > 0.$$
To do this, we note that
$$ 2(y_0-y) + s  \leq 2L^*_C < 2 L_C, \qquad  s \in [0,y-y_0 + L^*_C]. $$
By  Assumption~\ref{ass:fg}  the function $f$ is strictly concave on $[0,b]$. As by construction we have that
$$ 2 L_C < b,$$
it follows that
$$f'(2(y_0-y)+s) < f'(s), \qquad s \in [0,y-y_0 + L^*_C]$$
and
$$ \int_0^{y-y_0+L^*_C} Q^*_C(y-s) \Big ( f'(s) - f'(2(y_0-y)+s) \Big ) ds > 0.$$

\newpage
Finally, we consider the case where
$$ y \in (mid(I_C) - L,mid(I_C) + L^*_C - L) = [y_0 - L,y_0 + L^*_C - L).$$
Let
$$\Dl_L = y - (y_0 - L).$$
Recall that
$$ mid(I_C) + L^*_C - L < mid(I_C) -  L^*_C,$$
as by assumption we have that
$$ 0 < L^*_C < L_C < \frac{L}{2}.$$
We then obtain that
\begin{eqnarray*}
\frac{d}{dy} \FdC(y) &=& \int_{I^*_C} Q^*_C(x) \frac{d}{dy} p(x|y) dx \\
&=& - \int_{y_0-L^*_C}^{y_0 - \Dl_L} Q^*_C(x) f'(x-y) dx - \cdots \\
&& - \int_{0}^{\Dl_L} Q^*_C(y_0 + s) \Bbl f'(y_0-y-s) + f'(y_0-y+s)  \Bbr  ds + \cdots\\
&& + \int_{y_0+\Dl_L}^{y_0 + L^*_C} Q^*_C(x) f'(x-y) dx \\
&=&  - \int_{0}^{L^*_C - \Dl_L} Q^*_C(y_0-L^*_C + s) \Bbl f'(L - L^*_C + \Dl_L +s) - f'(L - L^*_C + \Dl_L + s)   \Bbr ds- \cdots \\
&& - \int_{0}^{\Dl_L} Q^*_C(y_0 + s) \Bbl f'(y_0-y-s) + f'(y_0-y+s)  \Bbr  ds \\
&=& - \int_{0}^{\Dl_L} Q^*_C(y_0 + s) \Bbl f'(y_0-y-s) + f'(y_0-y+s)  \Bbr  ds.
\end{eqnarray*}
By Assumption~\ref{ass:fg} the function $f(\cdot)$ is strictly decreasing on $[0,L]$, and it follows that
$$ - \int_{0}^{\Dl_L} Q^*_C(y_0 + s) \Bbl f'(y_0-y-s) + f'(y_0-y+s)  \Bbr  ds > 0, \qquad y \in  (y_0 - L,y_0 + L^*_C - L).$$
It then follows that  we have that
$$\frac{d}{dy} \FdC(y) > 0, \qquad y \in (mid(I_C) - L,mid(I_C)).$$

\newpage
As by assumption the function $\QsC(x)$ is symmetric with respect to $mid(I_C)$,
and we have that
$$p(y_0 - s|y) = p(y_0 + s|y) = f(s), \qquad s \in [0,L],$$
we obtain that derivative $\frac{d}{dy} \FdC(y)$ at $y=y_0$ is equal to 
$$\int_0^{L^*_C} Q^*_C(y_0-s) \big [ -f'(s) + f'(s) \big ]ds  = 0.$$
The result of the lemma then follows.
\end{proof}

\newpage
\subsection{Properties of the Function $\UdC(y)$}
We next characterize the properties of the function
$$U_C^{(d)}(y) = \max \Bsl 0,\FdC(y)\Bsr, \qquad y \in \setR,$$
where
$$\FdC(y) =  E_p \int_{\setR} \Bsbl p(x|y) Q^*_C(x) dx - 2L_CE_q c \Bsbr dx, \qquad y \in \setR,$$
and $Q^*_C(x)$ is as given in Proposition~\ref{prop:Q_C}.

Compared with Lemma~\ref{lemma:FdC}, we add for this analysis the assumption that
$$  \FdC(y) = E_p \int_{\setR} \Bsbl p(x|y) Q^*_C(x) dx - 2L_CE_q c \Bsbr dx >0, \qquad y \in I_C.$$
The assumption 
$$ \UdC(y) > 0, \qquad y \in I_C,$$
reflects the result of Proposition~\ref{prop:nash} which states that under a Nash equilibrium all agents have strictly positive utility rates. We then have the following result.
\begin{lemma}\label{lemma:UdC}
Let $E_p$, $0<E_p\leq 1$, and  $E_q$, $0<E_q$, be given constants, and let $I_C$ be an interval in $\setR$.
For
$$ L_C = \frac{|I_C|}{2},$$
assume that
$$2 L_C < \min \{ b, L \}$$
where
$b$ is the constant of Assumption~\ref{ass:fg}.
Furthermore assume for the function $Q^*_C(x)$, $x \in \setR$, that
\begin{enumerate}
\item $\QsC(x)$, $x \in \setR$, is bounded on $\setR$. 
\item the support $supp(Q_C^*)$ of $Q^*_C(x)$  is given by
 $$supp(Q_C^*) = [mid(I_C) - L^*_C, mid(I_C) + L^*_C]$$
 where
 $$ 0 < L^*_C < L_C,$$
\item $Q^*_C(x)$ is non-negative, continuous, and symmetric with respect to $mid(I_C)$, on $(mid(I_C) - L^*_C, mid(I_C) + L^*_C)$. 
\end{enumerate}
Let 
$$U_C^{(d)}(y) = \max \Bsl 0,\FdC(y)\Bsr, \qquad y \in \setR,$$
where 
$$\FdC(y) =  E_p \left [ \int_{\setR} p(x|y) Q^*_C(x) dx - 2L_CE_q c   \right ], \qquad y \in \setR.$$
If we have that
$$\UdC(y) > 0, \qquad y \in I_C,$$
then the function $\UdC(y)$, $y \in \setR$, has the properties that
\begin{enumerate}
\item[a)]  $\UdC(y)$ is symmetric with respect to $y_0 = mid(I_C)$ on $\setR$, and we  have for $y = mid(I_C)$   that
$$\frac{d}{dy} \UdC(mid(I_C)) = 0.$$
\item[b)]  $\UdC(y)$ is strictly increasing on $[mid(I_C) - L_C,mid(I_C))$, and we have that
$$\frac{d}{dy} \UdC(y) > 0, \qquad y \in [mid(I_C) - L_C,mid(I_C)).$$
\item[c)]  $\UdC(y)$ is strictly decreasing on $(mid(I_C), mid(I_C) + L_C]$, and we have that
$$\frac{d}{dy} \UdC(y) < 0, \qquad y \in  (mid(I_C), mid(I_C) + L_C].$$
\item[d)]  $\UdC(y)$ is non-decreasing on $[mid(I_C) - L,mid(I_C))$.
\item[e)]  $\UdC(y)$ is non-increasing on $(mid(I_C), mid(I_C) + L)$.
\end{enumerate}
\end{lemma}

\begin{proof}
The result of the lemma follows directly from the result of Lemma~\ref{lemma:FdC} and the the fact that
$$U_C^{(d)}(y) = \max \Bsl 0,\FdC(y)\Bsr, \qquad y \in \setR.$$
\end{proof}

\newpage
\subsection{Proof of Proposition~\ref{prop:UdC}}
The result of Proposition~\ref{prop:UdC} follows from Lemma~\ref{lemma:UdC} as the setting of Proposition~\ref{prop:UdC} satisfies all the assumptions made in the statement of Lemma~\ref{lemma:UdC}.

In particular, Proposition~\ref{prop:Q_C} states the the content supply function $Q^*_C(x)$ satisfies the assumptions made on $Q^*_C(x)$ in Lemma~\ref{lemma:UdC}.
Moreover, Proposition~\ref{prop:nash} states that under a Nash equilibrium all agents receive a strictly positive utility, and we have that for a community $C = (I_C,I_C) \in \setC^*$ in a Nash equilibrium  $(\setC^*, \{\alpha^*_{\setC}(y)\}_{y \in \setR},  \{\beta^*_{\setC}(\cdot|y)\}_{y \in \setR})$ as given by Proposition~\ref{prop:nash} that
$$\UdC(y) > 0, \qquad y \in I_C.$$
Finally, by Proposition~\ref{prop:xs} we have for a community $C=(I_C,I_C) \in \setC^*$ under a Nash equilibrium  $(\setC^*, \{\alpha^*_{\setC}(y)\}_{y \in \setR},  \{\beta^*_{\setC}(\cdot|y)\}_{y \in \setR})$ as given by Proposition~\ref{prop:nash} we have that
$$ \int_{\setR} \beta^*_C(x|y)dx = E_q, \qquad y \in C_s = I_C,$$
and 
$$\asC(y) = E_p, \qquad y \in C_d = I_C.$$

\newpage
\section{Proof of Proposition~\ref{prop:UsC}}\label{app:UsC}
In this appendix we prove Proposition~\ref{prop:UsC} which characterize the utility rates for content production under a Nash Equilibrium $\commStrucNash$ as given by Proposition~\ref{prop:nash}.

To do that, we use the same approach as in the previous subsection where we characterized the utility rates for content consumption. That is,  we consider the case where we are given  a set $C_d = I_C$ of content consumers where $I_C \subset \setR$ is an interval in  $\setR$. In addition, we assume that we are given the content demand function $P_C(x)$, $x \in \setR$, for the content consumed by the agents $y \in C_d = I_C$, and then characterize then the maximal utility of content consumption that a agent $y \in \setR$ can obtain under the content demand function $P_C(x)$.

Our first result characterize the utility function under the additional assumption that the for the fraction of time $\aC(y)$ that agents in the set $C_s = I_C$ allocate for content consumption has the property that
$$\aC(y) = E_p > 0, \qquad y \in I_C,$$
For this case, we have the following result. 
\begin{lemma}\label{lemma:UsC_FsC}
Let   $E_p$, $0<E_q \leq 1$, be a given constant, and let $C_d = I_C$ be a set of content consumers where $I_C \subset \setR$ is an interval in  $\setR$.
Assume that
$$\aC(y)  = E_p > 0, \qquad y \in I_C,$$
where $\aC(y)$, $y \in C_d = I_C$, is the fraction of time an agent $y$ consumes content. 
Then the maximal utility rate for content producer that an agent $y \in \setR$ obtains from the consumers $y \in I_C$ is given by
$$\UsC(y) = \max \Bsl 0,\FsC(y) \Bsr, \qquad y \in \setR,$$
where
$$\FsC(y) = E_q \Bsbl q(x^*(y)|y) P_C(\xs(y)) - 2 L_C E_p c \Bsbr, \qquad y \in \setR,$$
and  $E_q$, $0 < E_q$, is the maximal total content production rate that agents $y \in \setR$ can allocated to the set of agents $C_d = I_C$, and
$$ \xs(y)= \arg \max_{x \in \setR} q(x|y) P_C(x).$$
\end{lemma}

\begin{proof}
The result of this lemma follows directly from the definition of the utility rate for content production  of an agent $y \in \setR$ given by
$$\UdC(y) = \int_{\setR} \bsC(x|y) \Big [ q(x|y) P_C(x) p(x|y) - \aC c \Big ] dx,$$
where $\bsC(\cdot|y)$ is the optimal content production rate given by
$$\bsC(\cdot|y) = \underset{\beta(\cdot|y): || \beta(\cdot|y)|| \leq \beta_C(y)}{\arg\max} \int_{\setR} \beta_C(x|y) \Bsbl q(x|y) P_C(x) - \alpha_C c \Bsbr dx,$$
and the assumptions that
$$ 0 \leq \bC(y) \leq E_q$$
and
$$\aC(y) = E_p > 0, \qquad y \in I_C.$$
In particular, under the assumption that
$$\aC(y) = E_p > 0, \qquad y \in I_C,$$
we obtain that
$$\aC = \int_{I_C} \aC(y) dy = 2 L_C E_p.$$
Furthermore, by Corollary~\ref{cor:optimal_value_production} we have that
$$ \beta^*_C(\cdot|y) = E_q \delta(\xs(y) - x), \qquad y \in I_C,$$
where
$$ \xs(y) = \arg \max_{x \in \setR} q(x|y)P_C(x),$$
is a solution to the optimization problem
$$\underset{\beta(\cdot|y): || \beta(\cdot|y)|| \leq E_q}{\max} \int_{\setR} \beta_C(x|y) \Bsbl q(x|y) P_C(x) - \alpha_C c \Bsbr dx.$$
\end{proof}

\newpage
\subsection{Properties of the Function $\FsC(y)$}
Next we characterize the properties of the function $\FsC(y)$, $y \in \setR$, for the special case where we have that the function $P_C(x)$ is as given in Proposition~\ref{prop:P_C}, i.e. we assume that the content demand  function $P_C(x)$, $x \in \setR$, has the same properties as the content demand function for an interval community in a Nash equilibrium as given by Proposition~\ref{prop:nash}.

For this case, we have the following two results.
\begin{lemma}\label{lemma:FsC1}
Let $E_p$, $0<E_p\leq 1$, and  $E_q$, $0<E_q$, be given constants, and let $I_C$ be an interval in $\setR$.
For
$$ L_C = \frac{|I_C|}{2},$$
assume that
$$L_C < \frac{L}{2}$$
Furthermore, assume for the function $P_C(x)$, $x \in \setR$, that
\begin{enumerate}
\item[(a)] $P_C(x)$ is positive and  twice continuously differentiable on $\setR$.
\item[(b)] $P_C(x)$ is symmetric with respect to $mid(I_C)$. 
\item[(c)] $P_C(x)$ is strictly increasing on the interval $( mid(I_C) - L, mid(I_C))$, and strictly decreasing on the interval $(mid(I_C), mid(I_C) + L)$.
\item[(d)] $P_C(x)$ is strictly concave in $x$ on the interval $[mid(I_C) - L_C, mid(I_C)+ L_C]$.
\end{enumerate}
Then  the function $\FsC(y)$ given by
$$\FsC(y) = E_q [q(x^*(y)|y) P_C(\xs(y)) -  2L_CE_p c ], \qquad y \in \setR,$$
where
$$ \xs(y)= \arg \max_{x \in \setR} q(x|y) P_C(x),$$
has the properties that
\begin{enumerate}
\item[a)]  $\FsC(y)$ is symmetric with respect to $y_0 = mid(I_C)$ on $[mid(I_C) - L_C,mid(I_C) + L_C]$, and we have that
$$\frac{d}{dy} \FsC(mid(I_C)) = 0.$$
\item[b)]  $\FsC(y)$ is strictly increasing on $[mid(I_C) - L_C,mid(I_C))$, and we have that
$$\frac{d}{dy} \FsC(y) > 0, \qquad y \in [mid(I_C) - L_C,mid(I_C)).$$
\item[c)]  $\FsC(y)$ is strictly decreasing on $(mid(I_C), mid(I_C) + L_C]$, and we have that
$$\frac{d}{dy} \FsC(y) < 0, \qquad y \in  (mid(I_C), mid(I_C) + L_C].$$
\end{enumerate}
\end{lemma}

\begin{proof}
Let 
$$y_0 = mid(I_C).$$
We first show that $\FsC(y)$ is symmetric with respect to $y_0 = y_0$ on $[y_0 - L_C,y_0 + L_C]$. For this we rewrite $\FsC(y)$ as
$$\FsC(y) =   E_q \Bsbl g(\Ds(y)) P_C(y - \Ds(y)) -  2L_CE_p c \Bsbr, \qquad y \in [y_0, y_0 + L_C],$$
and
$$\FsC(y) =   E_q \Bsbl g(\Ds(y)) P_C(y + \Ds(y)) -  2L_CE_p c \Bsbr, \qquad y \in [y_0 - L_C, y_0],$$
where
$$\Ds(y) = || y - \xs(y)||, \qquad y \in [y_0 - L_C,y_0 + L_C].$$
By the same argument as given to prove Lemma~\ref{lemma:Ds}, we have that the function $\Ds(y)$ is symmetric with respect to $y_0$ on $[y_0 - L_C,y_0 + L_C]$. Furthermore, by assumption we have that the function $P_C(x)$ is symmetric with respect to $y_0$ on $\setR$.
It then follows that for $s \in [0,L_C]$ we have 
\begin{eqnarray*}
\FsC(y_0+s) &=&   E_q [g(\Ds(y_0+s)) P_C(y_0+s - \Ds(y_0+s)) -  2L_CE_p c ] \\
&=&  E_q [g(\Ds(y_0-s)) P_C(y_0+s - \Ds(y_0-s)) -  2L_CE_p c ] \\
&=&  E_q [g(\Ds(y_0-s)) P_C(y_0 - s + \Ds(y_0-s)) -  2L_CE_p c ] \\
&=& \FsC(y_0-s).
\end{eqnarray*}
It then follows that the function $\FsC(y)$ symmetric with respect to $y = y_0$ on  $[y_0 - L_C,y_0 + L_C]$.

Next we show that  $\FsC(y)$ is strictly increasing on $[y_0 - L_C,y_0)$,  and strictly decreasing on $(y_0, y_0 + L_C]$. 
As the function $\FsC(y)$ is symmetric with respect to  $y_0$ on  $[y_0 - L_C,y_0 + L_C]$,  it suffices to prove the result for
$$ y \in (y_0,  y_0 + L_C],$$
i.e. to show that
$$\frac{d}{dy} \FsC(y) < 0, \qquad y \in (y_0, y_0+ L_C].$$
We then obtain that
$$ \frac{d}{dy} \FsC(y) = \frac{d}{d y}  q(\xs(y)|y)   P_C(\xs(y)) = \frac{d}{dy}  g(\Ds(y)) P_C(y - \Ds(y)), \qquad y \in  (y_0,y_0+L_C].$$
As by assumption the function $P_C(x)$ is strictly concave on $[y_0,y_0+L_C]$, as well as  is strictly increasing on the interval $( y_0 - L, y_0)$, and strictly decreasing on the interval $(y_0, y_0 + L)$, it follows by the same argument as given to prove Lemma~\ref{lemma:xs} that the optimization problem
$$ \xs(y)= \arg \max_{x \in \setR} q(x|y) P_C(x), \qquad  y \in  (y_0,y_0+L_C],$$
has a unique solution. Furthermore, by the same argument as given to prove Lemma~\ref{lemma:Ds}, we obtain that the function $\Ds(y)$ is strictly increasing on $  [y_0,y_0+L_C]$.

Let
$$\Dsy(y) = \frac{d}{dy} \Ds(y), \qquad y \in  (y_0,y_0+L_C],$$
and  let $P'_C(x)$ be the derivative of $P_C(x)$ with respect to $x$. By the same argument as given in the proof of Lemma~\ref{lemma:Ds}, we have that the derivative $\Dsy(y)$ is well-defined on $(y_0,y_0+L_C]$. Furthermore, by assumption the derivative $P'_C(x)$ is well defined on $\setR$. 

Using these definitions, we obtain that
\begin{eqnarray*}
 \frac{d}{d y}  g(\Ds(y)) P_C(y - \Ds(y))
 &=&  g'(\Ds(y)) P_C(y - \Ds(y))\Dsy(y) - \cdots \\
 && -  g(\Ds(y)) P'_C(y - \Ds(y))\Dsy(y) + \cdots \\
 && + g(\Ds(y)) P'_C(y - \Ds(y)) \\
 &=&  g'(\Ds(y)) P_C(y - \Ds(y))\Dsy(y) +  \cdots\\
 && + g(\Ds(y))  P'_C(y - \Ds(y)) \Big [ 1 -  \Dsy(y) \Big ].
 \end{eqnarray*}
By Assumption~\ref{ass:fg} we have that
$$ g'(\Ds(y)) < 0,  \qquad y \in  (y_0,y_0+L_C] \cap supp(g).$$
By the same argument as given to prove Lemma~\ref{lemma:Ds}, we obtain that the function $\Ds(y)$ is strictly increasing on $  (y_0,y_0+L_C]$ and we have that
$$\Dsy(y) > 0, \qquad y \in  (y_0,y_0+L_C].$$
and 
$$ |\Dsy(y)| < 1, \qquad  \qquad y \in  (y_0,y_0+L_C].$$
Furthermore, by assumption we have that
$$ P'_C(y) < 0, \qquad  y \in  (y_0,y_0+L_C].$$
Combining these results, it follows that
$$  1 -  \Dsy(y) > 0, \qquad y \in  (y_0,y_0+L_C],$$
and
\begin{eqnarray*}
 \frac{d}{d y}  g(\Ds(y)) P_C(y - \Ds(y))
 &=&  g'(\Ds(y)) P_C(y - \Ds(y))\Dsy(y) + \cdots \\
 && +  g(\Ds(y))  P'_C(y - \Ds(y)) \Big [ 1 -  \Dsy(y) \Big ] \\
 &<& 0, \qquad y \in  (y_0,y_0+L_C].
 \end{eqnarray*}
Finally, by Assumption~\ref{ass:fg} we have that
$$g'(0) = 0.$$
Furthermore, as by assumption we have the $P_C(x)$ is symmetric with respect to $mid(I_C)$ and $P_C(x)$ is strictly concave on $[mid(I_C) - L_C, mid(I_C)+L_C]$, it follows that
$$P'_C(mid(I_C)) = 0,$$
and by the same argument as given to prove Lemma~\ref{lemma:single_content_production} we have that
$$\xs(y_0) = y_0$$
and
$$\Ds(y_0) = 0.$$
Combining these results, it follows from the above expression of the derivative of $\FsC(y)$ that for $y = y_0$ we have that
$$\frac{d}{dy} \FsC(y) = 0.$$
The result of the lemma then follows.
\end{proof}

\newpage
\begin{lemma}\label{lemma:FsC2}
Let $E_p$, $0<E_p\leq 1$, and  $E_q$, $0<E_q$, be given constants, and let $I_C$ be an interval in $\setR$.
For
$$ L_C = \frac{|I_C|}{2},$$
assume that
$$L_C < \frac{L}{2}$$
Furthermore, assume for the function $P_C(x)$, $x \in \setR$, that
\begin{enumerate}
\item[(a)] $P_C(x)$ is positive and  twice continuously differentiable on $\setR$.
\item[(b)] $P_C(x)$ is symmetric with respect to $mid(I_C)$. 
\item[(c)] $P_C(x)$ is strictly increasing on the interval $( mid(I_C) - L, mid(I_C))$, and strictly decreasing on the interval $(mid(I_C), mid(I_C) + L)$.
\end{enumerate}
Then  the function $\FsC(y)$ given by
$$\FsC(y) = E_q \Bsbl q(x^*(y)|y) P_C(\xs(y)) -  2L_CE_p c \Bsbr, \qquad y \in \setR,$$
where
$$ \xs(y)= \arg \max_{x \in \setR} q(x|y) P_C(x),$$
has the properties that
\begin{enumerate}
\item[a)]  $\FsC(y)$ is non-decreasing on $[mid(I_C) - L,mid(I_C) - L_C]$.
\item[b)]  $\FsC(y)$ is non-increasing on $[mid(I_C)+L_C, mid(I_C) + L)$.
\end{enumerate}
\end{lemma}

\begin{proof}
We prove the result for $ y \in  [mid(I_C)+L_C, mid(I_C) + L)$, i.e. to prove that  $\FsC(y)$ is non-increasing on $[mid(I_C)+L_C, mid(I_C) + L)$.
The result that $\FsC(y)$ is non-decreasing on $[mid(I_C) - L,mid(I_C) - L_C]$ can be obtained by the same argument.

Consider two agents $y,y'$ such that
$$y,y' \geq y_0+L_C$$
and
$$ y' > y,$$
and show that for 
$$ \xs(y) = \arg \max_{x \in \setR} q(x|y) P_C(x)$$
and
$$  \xs(y') = \arg \max_{x \in \setR} q(x|y') P_C(x)  ,$$
we have that
$$q(x^*(y)|y) P_C(x^*(y)) >  q(\xs(y')|y')   P_C(\xs(y')).$$
We prove this result by contradiction as follows. 
Suppose that the result is not true and we have that
$$q(x^*(y)|y) P_C(x^*(y)) \leq  q(\xs(y')|y')   P_C(\xs(y'));$$
in this case we show that there exists a value $\hat x_y$ such that
$$ q(\hat x_y|y) P_C(\hat x_y) >  q(\xs(y')|y')   P_C(\xs(y'))$$
which contradicts our assertion that
$$ \xs(y) = \arg \max_{x \in \setR} q(x|y) P_C(x).$$
In particular, let $\hat x_y$ be given by
$$ \hat x_y = y -  (y' - \xs(y')).$$
We then have that
$$  q(\xs(y')|y')   P_C(\xs(y')) < q(\hat x_y|y) P_C(\hat x_y).$$
To see that this is true note that
$$  q(\xs(y')|y') =  q(\hat x_y|y) = g( | y' - \xs(y') |)$$
and
$$  P_C(\xs(y')) <  P_C(\hat x_y),$$
as by assumption the function $P_C(x)$ is strictly decreasing for $x \geq mid(I_C)$.

The result of the lemma then follows.
\end{proof}

\newpage
We next characterize the properties of the function
$$\UsC(y) = \max \Bsl 0,\FsC(y)\Bsr, \qquad y \in \setR,$$
where
$$\FsC(y) = E_q [q(x^*(y)|y) P_C(x) - E_p c ], \qquad y \in \setR,$$
where
$$ \xs(y)= \arg \max_{x \in \setR} q(x|y) P_C(x).$$
Similar to the  analysis in the previous subsection, we add for this analysis the assumption that
$$  \UsC(y) >0, \qquad y \in I_C.$$

We have the following result.
\begin{lemma}\label{lemma:UsC}
Let $E_p$, $0<E_p\leq 1$, and  $E_q$, $0<E_q$, be given constants, and let $I_C$ be an interval in $\setR$.
For
$$ L_C = \frac{|I_C|}{2},$$
assume that
$$L_C < \frac{L}{2}$$
Furthermore, assume for the function $P_C(x)$, $x \in \setR$, that
\begin{enumerate}
\item[(a)] $P_C(x)$ is positive and  twice continuously differentiable on $\setR$.
\item[(b)] $P_C(x)$ is symmetric with respect to $mid(I_C)$. 
\item[(c)] $P_C(x)$ is strictly increasing on the interval $( mid(I_C) - L, mid(I_C))$, and strictly decreasing on the interval $(mid(I_C), mid(I_C) + L)$.
\item[(d)] $P_C(x)$ is strictly concave in $x$ on the interval $[mid(I_C) - L_C,mid(I_C) + L_C]$.
\end{enumerate}
Let
$$\UsC(y) = \max \Bsl 0,\FsC(y)\Bsr, \qquad y \in \setR,$$
where
$$\FsC(y) = E_q [q(x^*(y)|y) P_C(\xs(y)) - E_p c ], \qquad y \in \setR,$$
and
$$ \xs(y)= \arg \max_{x \in \setR} q(x|y) P_C(x).$$
If we have that
$$\UsC(y) > 0, \qquad y \in I_C,$$
then the function $\UsC(y)$, $y \in \setR$, has the properties that
\begin{enumerate}
\item[a)]  $\UsC(y)$ is symmetric with respect to $y_0 = mid(I_C)$ on $[mid(I_C) - L_C,mid(I_C) + L_C]$, and we have that
$$\frac{d}{dy} \UsC(mid(I_C)) = 0.$$
\item[b)]  $\UsC(y)$ is strictly increasing on $[mid(I_C) - L_C,mid(I_C))$, and we have that
$$\frac{d}{dy} \UsC(y) > 0, \qquad y \in [mid(I_C) - L_C,mid(I_C)).$$
\item[c)]  $\UsC(y)$ is strictly decreasing on $(mid(I_C), mid(I_C) + L_C]$, and we have that
$$\frac{d}{dy} \UsC(y) < 0, \qquad y \in  (mid(I_C), mid(I_C) + L_C].$$
\item[b)]  $\UsC(y)$ is non-decreasing on $[mid(I_C) - L,mid(I_C) - L_C]$.
\item[c)]  $\UsC(y)$ is non-increasing on $[mid(I_C), mid(I_C) + L)$.
\end{enumerate}
\end{lemma}

\begin{proof}
The result of the lemma follows directly from the result of Lemma~\ref{lemma:FsC1}~and~\ref{lemma:FsC2},
combined with the assumption that 
$$\UsC(y) > 0, \qquad y \in I_C.$$
\end{proof}

\newpage
\subsection{Proof of Proposition~\ref{prop:UsC}}
The result of Proposition~\ref{prop:UsC} follows from Lemma~\ref{lemma:UsC} as the setting of Proposition~\ref{prop:UsC} satisfies all the assumptions made in the statement of Lemma~\ref{lemma:UsC}.

In particular, Proposition~\ref{prop:P_C} states the the content demand function $P_C(x)$ satisfies the assumptions made on $P_C(x)$ in Lemma~\ref{lemma:UsC}.
Moreover, Proposition~\ref{prop:nash} states that under a Nash equilibrium all agents receive a strictly positive utility, and we have that for a community $C = (I_C,I_C) \in \setC^*$ in a Nash equilibrium  $(\setC^*, \{\alpha^*_{\setC}(y)\}_{y \in \setR},  \{\beta^*_{\setC}(\cdot|y)\}_{y \in \setR})$ as given by Proposition~\ref{prop:nash} that
$$\UsC(y) > 0, \qquad y \in I_C.$$
Finally, by Proposition~\ref{prop:xs} we have for a community $C=(I_C,I_C) \in \setC^*$ under a Nash equilibrium  $(\setC^*, \{\alpha^*_{\setC}(y)\}_{y \in \setR},  \{\beta^*_{\setC}(\cdot|y)\}_{y \in \setR})$ as given by Proposition~\ref{prop:nash}
 we have that
$$ \int_{\setR} \beta^*_C(x|y)dx = E_q, \qquad y \in C_s = I_C,$$
and 
$$\asC(y) = E_p, \qquad y \in C_d = I_C.$$

\newpage

\section{Proof of Proposition~\ref{prop:nash}}\label{app:nash}
In this appendix we prove Proposition~\ref{prop:nash}. To do this, we first derive three additional results. First we characterize the optimal community selection for a given agent, and show that it is optimal for agents to consume, or produce, content in a single community. Next, we provide sufficient conditions for an interval community to be a feasible community, i.e. a community where all agents obtain strictly positive utility rates.

\newpage
\subsection{Optimal Community Selection for Content Consumers and Producers}
In this subsection we characterize the optimal allocation for content consumption an and production by an agent $y \in \setR$ in a given community structure.

First we characterize the optimal allocation for content consumption by a given agent $y \in \setR$. 
More precisely, for a given community structure $(\setC, \{\alpha_{\setC}(y)\}_{y \in \setR},  \{\beta_{\setC}(\cdot|y)\}_{y \in \setR})$, as well as a given agent $y \in \setR$, we characterize the optimal allocation for content consumption
$$\as_{\setC}(y) = \{ \asC(y) \}_{C \in \setC}$$
given by
$$\as_{\setC}(y) = \alpha^*_{\setC} = 
\underset{\alpha_{\setC}(y): || \alpha_{\setC}(y)  || \leq E_p}{\arg\max} \sum_{C \in \setC}  \alpha_{C}(y) \int_{\setR} \Bsbl Q_C(x) p(x|y) - \beta_C(x) c \Bsbr dx,$$
where $E_p$ is the maximal fraction of time that agent can allocate to content consumption. 

We have the following result.
\begin{lemma}\label{lemma:single_community_consumption}
Let $(\setC, \{\alpha_{\setC}(y)\}_{y \in \setR},  \{\beta_{\setC}(\cdot|y)\}_{y \in \setR})$ be a given community structure. Assume that  for a given agent $y \in \setR$ we have that
$$ \max_{C \in \setC} \int_{\setR} \Bsbl p(x|y)Q_C(x) - \beta_C(x) c \Bsbr dx > 0,$$
where for the community $C = (C_d,C_s) \in \setC$ we have 
$$Q_C(x) = \int_{y \in C_s}  \beta_C(x|y) q(x|y) dy$$
and
$$\bC(x) =  \int_{y \in C_s}  \beta_C(x|y) dy.$$
Then an optimal content consumption rate allocation for this agent $y$ is given by
$$ \alpha^*_C(y) = \left \{
\begin{array}{rl}
E_p, & C = C^*(y) \\
0, & \mbox{otherwise,}
\end{array} \right .$$
where
$$C^*(y) = \arg \max_{C \in \setC} \int_{\setR} \Bsbl p(x|y)Q_C(x) -  \beta_C(x) c\Bsbr dx,$$
and $E_p$, $0 < E_p \leq 1$, is the maximal fraction of time that agent can allocate to content consumption, 
i.e. we have that
$$\alpha^*_{\setC}(y) = \\
\underset{\alpha_{\setC}(y): || \alpha_{\setC}(y)  || \leq E_p}{\arg\max} \sum_{C \in \setC}  \alpha_{C}(y) \int_{\setR} \Bsbl Q_C(x) p(x|y) - \beta_C(x) c \Bsbr dx.$$
\end{lemma}
Lemma~\ref{lemma:single_community_consumption} states that is optimal for an agent $ y \in \setR$ to allocate all its content consumption to a single community in a community structure$(\setC, \{\alpha_{\setC}(y)\}_{y \in \setR},  \{\beta_{\setC}(\cdot|y)\}_{y \in \setR})$, namely the community $C^*(y)$ given by 
$$C^*(y) = \arg \max_{C \in \setC} \int_{\setR} \Bsbl p(x|y)Q_C(x) -  \beta_C(x) c\Bsbr dx.$$

\begin{proof}
Recall the definition of the utility rate for content consumption given by 
$$ U^{(d)}_C(y) = \alpha_{C}(y) \int_{\setR} \Bsbl [ Q_C(x) p(x|y) - \beta_C(x) c \Bsbr dx.$$
The result of the lemma follows then immediately from the fact that the utility rate for content consumption is linear in $\alpha_{C}(y)$.
\end{proof}

\newpage
Next we characterize the optimal allocation for content production of a given agent $y \in \setR$. More precisely, for a given community structure $(\setC, \{\alpha_{\setC}(y)\}_{y \in \setR},  \{\beta_{\setC}(\cdot|y)\}_{y \in \setR})$, and a given agent $y \in \setR$, we characterize the optimal allocation for content production 
$$\bs_{\setC}(y) = \{ \bsC(\cdot |y) \}_{C \in \setC}$$
given by
$$\bs_{\setC}(\cdot|y) = 
\underset{\beta_{\setC}(\cdot|y): || \beta_{\setC}(\cdot|y)  || \leq E_q}{\arg\max} \sum_{C \in \setC} \int_{\setR} \beta_{C}(x|y) [ q(x|y) P_C(x) - \alpha_C c] dx,$$
where $E_q$ is the maximal rate that agent can allocate to content production.

The next lemma shows that it is optimal for agent $y$ to allocate all its content production rate in a given community structure  $(\setC, \{\alpha_{\setC}(y)\}_{y \in \setR},  \{\beta_{\setC}(\cdot|y)\}_{y \in \setR})$ to a single community $C \in \setC$. 

\begin{lemma}\label{lemma:single_community_production}
Let $(\setC, \{\alpha_{\setC}(y)\}_{y \in \setR},  \{\beta_{\setC}(\cdot|y)\}_{y \in \setR})$ be a given community structure. Assume that for a given agent  $y \in \setR$ we have that
$$ \max_{C \in \setC} \Big ( q(\xsC(y)|y)P_C(\xsC(y)) - \alpha_C c \Big )  > 0,$$
where
$$\xsC(y) = \arg \max_{x \in \setR} q(x|y) P_C(x),$$
$$P_{C}(x) = \int_{y \in C_d}  \alpha_C(y) p(x|y) dy,$$
and
$$\aC = \int_{y \in I_C} \aC(y) dy.$$
Then the optimal content production rate allocation for this agent $y$ is given by
$$ \beta^*_C(x|y) = \left \{
\begin{array}{rl}
E_q \delta(\xsC(y) -x) , & C = C^*(y) \\
0, & \mbox{otherwise,}
\end{array} \right .$$
where
$$C^*(y) = \arg \max_{C \in \setC} q(x^*_{C,y}|y)P_C(x^*_{C,y})$$
and $E_q$ is the maximal rate that agent can allocate to content production,
i.e. we have that
$$\beta^*_{\setC}(\cdot|y) = \\
\underset{\beta_{\setC}(\cdot|y): || \beta_{\setC}(\cdot|y)  || \leq E_q}{\arg\max} \sum_{C \in \setC} \int_{\setR} \beta_{C}(x|y) \Bsbl q(x|y) P_C(x) - \alpha_C c \Bsbr dx.$$
\end{lemma}

\begin{proof}
Note that if for a community $C \in \setC$ we have that
$$ \max_{x \in \setR} \Big ( q(x|y)P_C(x) - \alpha_C c \Big ) \leq 0,$$
then the optimal utility that agent $y$ can achieve in this community is equal to $0$, and hence it is optimal for agent $y$ to set $\aC(y) = 0$ for this community. As a result, to prove the lemma it suffices to consider communities $ C \in \setC$ such that
$$ \max_{x \in \setR} \Big ( q(x|y)P_C(x) - \alpha_C c \Big ) > 0.$$
Let $\setCp$ be this set, i.e. we have that
$$\setCp = \Bsl C \in \setC |  \max_{x \in \setR} \Bbl q(x|y)P_C(x) - \aC c  \Bbr > 0 \Bsr.$$

By Proposition~\ref{prop:optimal_production}, we have that if agent $y$ allocates the content production rate
$$\beta_C(y) = \int_{\setR} \beta_C(x|y)  dx$$ 
to community $C$, and we have that
$$ \max_{x \in \setR} \Big ( q(x|y)P_C(x) - \alpha_C c \Big ) > 0,$$
then the optimal utility rate for content production for agent $y$ in community $C$ is given by
$$\beta_C(y)\max_{x \in \setR} \Big ( q(x|y) P_C(x) - \alpha_{C} c \Big ).$$
Let $\beta^*_{\setC}(\cdot|y)$ be the rate allocation for content production as given in the lemma, and let $\beta_{\setC}(\cdot|y)$ be another rate allocation function such that
$$\sum_{C \in \setCp} \beta_C(y)  \leq E_q$$
and
$$\bC(y) = 0, \qquad C \notin \setCp,$$
where
$$\beta_C(y) = \int_{\setR} \beta_C(x|y)  dx.$$ 
To prove the lemma, we have to show that the utility rate  for content production of agent $y$ is not higher under the allocation $\beta_{\setC}(\cdot|y)$ compared with allocation  $\beta^*_{\setC}(\cdot|y)$. 
Note that the optimal utility rate for agent $y$ under $\beta^*(\cdot|y)$ is given by
$$ E_q \max_{x \in \setR} \Big ( q(x|y) P_{C^*(y)}(x) - \alpha_{C^*(y)} c \Big)$$
where
$$C^*(y) = \arg \max_{x \in \setR} q(x|y) P_C(x),$$
and the  utility rate for agent $y$ under $\beta_{\setC}(\cdot|y)$ is given by
$$ \sum_{C \in \setC} \beta_C(y) \max_{x \in \setR}  \Big ( q(x|y) P_C(x) - \alpha_C c \Big ).$$
Therefore in order to show that the utility rate  for content production of agent $y$ is not higher under the allocation $\beta_{\setC}(\cdot|y)$ compared with allocation  $\beta^*_{\setC}(\cdot|y)$, we have to show that
$$ E_q \max_{x \in \setR} \Big ( q(x|y) P_{C^*(y)}(x) - \alpha_{C^*(y)} c \Big) -
\sum_{C \in \setC} \beta_C(y)\max_{x \in \setR}  \Big ( q(x|y) P_C(x) - \alpha_C c \Big ) \geq 0.$$
Note that for this expression we have that
\begin{eqnarray*}
  E_q\max_{x \in \setR} \Big ( q(x|y) P_{C^*(y)}(x) - \alpha_{C^*(y)} c \Big )
-   \sum_{C \in \setC} \beta_C(y)\max_{x \in \setR} \Big ( q(x|y) P_C(x) - \alpha_C c \Big ) = \\
\left ( E_q - \sum_{C \in \setC} \beta_C(y) \right ) \max_{x \in \setR} \Big ( q(x|y) P_{C^*(y)}(x) - \alpha_{C^*(y)} c \Big ) + \dots \\
+ \sum_{C \in \setC} \beta_C(y) \Big [ \max_{x \in \setR} \Big ( q(x|y) P_{C^*(y)}(x) - \alpha_{C^*(y)} c \Big ) -  \max_{x \in \setR} \Big ( q(x|y) P_C(x) - \alpha_C c \Big )  \Big ].
\end{eqnarray*}
By assumption we have that
$$\max_{C \in \setC} \Big ( q(\xsC(y)|y)P_C(\xsC(y)) - \alpha_C c \Big ) > 0,$$
and by construction we have that
$$\sum_{C \in \setC} \beta_C(y) \leq E_q.$$
This implies that
$$\left ( E_q - \sum_{C \in \setC} \beta_C(y) \right ) \max_{x \in \setR} \Big ( q(x|y) P_{C^*(y)}(x) - \alpha_{C^*(y)} c \Big ) \geq 0.$$
Furthermore, by assumption we have that
$$C^*(y) = \arg \max_{C \in \setC} q(\xsC(y)|y)P_C(\xsC(y)),$$
and we obtain that
$$\sum_{C \in \setC} \beta_C(y) \Big [ \max_{x \in \setR} \Big ( q(x|y) P_{C^*(y)}(x) - \alpha_{C^*(y)} c \Big ) -  \max_{x \in \setR} \Big ( q(x|y) P_C(x) - \alpha_C c \Big )  \Big ] \geq 0.$$
Combining these two results, we obtain that
$$ E_q\max_{x \in \setR} \Big ( q(x|y) P_{C^*(y)}(x) - \alpha_{C^*(y)} c \Big )
-   \sum_{C \in \setC} \beta_C(y)\max_{x \in \setR} \Big ( q(x|y) P_C(x) - \alpha_C c \Big )\geq 0.$$
and the result of the lemma follows. 
\end{proof}

\newpage
\subsection{Existence of Feasible Communities}
Consider a Nash equilibrium  $(\setC^*, \{\alpha^*_{\setC}(y)\}_{y \in \setR},  \{\beta^*_{\setC}(\cdot|y)\}_{y \in \setR})$ as given by Proposition~\ref{prop:nash}. We then have that each interval community $C=(I_C,I_C) \in \setC^*$ is a feasible community as defined in Definition~\ref{def:feasible}, i.e. we have that
$$\UdC(u) > 0, \qquad y \in I_C$$
and
$$\UsC(u) > 0, \qquad y \in I_C,$$
and the utility rates of all agents in the community are strictly positive. Therefore, in order to prove  Proposition~\ref{prop:nash}, we have to show that feasible interval communities exist.  The next result provides sufficient conditions for an interval community $C=(I_C,I_C)$ to be a feasible community. 
\begin{lemma}\label{lemma:feasible_community}
Let  $C=(C_d,C_s)$ be an interval community, i.e. we have that  $ C_d = C_s = I_C,$ where $I_C \subset \setR$ is an interval in $\setR$, 
and let 
$$l_0 = \sup \{ l \in [0,L] | f(l) g(l) - c > 0 \}.$$
If we have that  
$$0 < 2 L_C < l_0,$$ 
then for all rate allocations $\{\alpha_C(y)\}_{y \in C_s}$ and $\{\beta_C(\cdot|y)\}_{y \in C_s}$ such that
$$\alpha_C(y) > 0, \qquad y \in C_d$$
and
$$ \int_{\setR} \beta_C(x|y) dx > 0, \qquad y \in C_s,$$
the community $C$ is a feasible community under  $\{\alpha_C(y)\}_{y \in C_s}$ and $\{\beta_C(\cdot|y)\}_{y \in C_d}$. 
\end{lemma}

\begin{proof}
Without loss of generality, we can assume that the interval $I_C$ is given by 
$$I_C= [-L_C, L_C].$$ 
Recall that the utility rate for content producers $y \in C_s = I_C$ is given by
$$U^{(s)}_C(y)  = \int_{\setR} \beta_{C}(x|y) \int_{z\in C_d = I_C} \alpha_C(z) \Bsbl q(x|y) p(x|z)- c \Bsbr dz dx.$$
By assumption we have that
$$g(2L_C) > 0,$$
and it follows 
$$I_C \subseteq supp(q(\cdot|y)), \qquad y \in C_s = I_C.$$
Furthermore  by Assumption~\ref{ass:fg} we have that 
$$  q(x|y) p(x|z) >   g(2L_C) f(2L_C), \qquad x,y,z \in I_C.$$
Combining these two results, we obtain for
$$ y \in C_d = I_C$$
that
\begin{eqnarray*}
U^{(s)}_C(y) &=& 
\int_{x \in I_C} \beta_{C}(x|y) \int_{z\in I_C} \alpha_C(z) \Bsbl  q(x|y) p(x|z)- c \Bsbr dz dx \\
&\geq& \int_{x \in I_C} \beta_{C}(x|y) \int_{z\in I_C} \alpha_C(z) \Bsbl  g(2 L_C) f(2 L_C)- c \Bsbr dz dx \\
&=&  \Bsbl  g(2L_C) f(2L_C)- c \Bsbr \int_{x \in I_C} \beta_{C}(x|y) \int_{z\in I_C} \alpha_C(z) dzdx.
\end{eqnarray*}
By assumption we have that
$$\alpha_C(y) > 0, \qquad y \in C_d = I_C$$
and
$$ \int_{\setR} \beta_C(x|y) dx > 0, \qquad y \in C_s = I_C,$$
and we obtain that
$$ \int_{z\in I_C} \alpha_C(z) dz > 0$$
and
$$\int_{x \in I_C} \beta_{C}(x|y) \int_{z\in I_C} \alpha_C(z) dzdx > 0, \qquad y \in C_s = I_C.$$
Furthermore by assumption we have that
$$ 0 < 2L_C < l_0,$$
and it follows that
$$f(2L_C) g(2L_C) - c > 0.$$
Combing the above results, we obtain that
$$ \Bsbl  g(2L_C) f(2L_C)- c \Bsbr \int_{x \in I_C} \beta_{C}(x|y) \int_{z\in I_C} \alpha_C(z) dzdx > 0$$
and
$$U^{(s)}_C(y) > 0, \qquad y \in C_s = I_C.$$

Similarly, for the utility rate for content consumers $y \in C_d = I_C$ we have that
\begin{eqnarray*}
U^{(d)}_C(y) &=& 
\alpha_C(y) \int_{x \in I_C} \int_{z\in I_C} \beta_{C}(x|z)  \Bsbl  q(x|z) p(x|y)- c \Bsbr dz dx \\
&\geq&  \alpha_C(y) \int_{x \in I_C}  \int_{z\in I_C}  \beta_{C}(x|z) \Bsbl  g(2 L_C) f(2 L_C)- c \Bsbr dz dx \\
&=&  \alpha_C(y) \Bsbl  g(2 L_C) f(2 L_C- c \Bsbr \int_{x \in I_C} \int_{z\in I_C} \beta_{C}(x|z)  dzdx.
\end{eqnarray*}
By the same argument as given above we obtain that
$$U^{(d)}_C(y) > 0, \qquad y \in C_d = I_C,$$
and the result of lemma follows. 
\end{proof}

\newpage
\subsection{Proof of Proposition~\ref{prop:nash}}
In this subsection we prove Proposition~\ref{prop:nash}. We do this by construction, i.e. we construct a  Nash equilibrium $(\setC^*, \{\alpha^*_{\setC}(y)\}_{y \in \setR},  \{\beta^*_{\setC}(\cdot|y)\}_{y \in \setR})$ where
$$\setC^* = \{ C^1,...,C^K\},$$
and the communities
$$C^k=(C^k_d,C^k_s), \qquad k=1,...,K,$$
are given by a set of mutually non-overlapping intervals $\{I_k\}_{k=1,...,K}$ of equal length, i.e. we have that
$$C^k=(C^k_d,C^k_s) = (I_k,I_k), \qquad k=1,...,K,$$
with
$$|I_k| = |I_{k'}|, \qquad k,k'=1,...,K.$$

We construct such a Nash equilibrium as follows.
By Assumption~\ref{ass:fg} the functions $f$ and $g$ are continuous, and it follows that under the condition
$$f(0)g(0) - c>0$$
there exists a $l_0 > 0$ such that
$$f(2l_0) g(2l_0) - c > 0.$$
Furthermore, there exists a $l^*$, $0< l^* < l_0$,  and a integer $K$ such that
$$ K l^* = L,$$
where $L$ is the half-length of the set $\setR = [-L,L)$. 
Let 
$$b_k = 2 k l^*, \qquad k=0,...,K$$
and let
$$I_k = [b_{k-1},b_{k}), \qquad k=1,...K.$$
Using these definitions, we construct a community structure $(\setC^*, \{\alpha^*_{\setC}(y)_{y \in \setR},  \{\beta^*_{\setC}(\cdot|y)_{y \in \setR}\}$ where
$$\setC^* = \{C^1,...,C^K\}$$
and
$$C^k= (C^k_d,C^k_d) = (I_k,I_k), \qquad k=1,...K.$$
Furthermore, for $k=1,...,K$, we let
$$\alpha^*_{C^k}(y) = E_p, \qquad y \in I_k,$$
and
$$\beta^*_{C^k}(x|y) = E_q \delta(x^*_{C^k}(y)-x), \qquad y \in I_k,$$
where
$$x^*_{C^k}(y) = \arg\max_{x \in \setR} q(x|y)P_{C^k}(x)$$
and
$$P_{C^k}(x) =  E_p \int_{y \in I_k} p(x|y) dy, \qquad x \in \setR.$$
Note that $Q^*_{C^k}(x)$, $k=1,...,K$,  is then given by
$$Q^*_{C^k}(x) =  E_q \int_{y \in I_k} E_q \delta(x^*_{C^k}(y)-x) q(x|y) dy, \qquad x \in \setR.$$

As we have that
$$ 0 < l^* \leq l_0,$$
it follows from Lemma~\ref{lemma:feasible_community} the the communities $C^k$, $k=1,..,K$, are feasible communities under the rate allocations  $\{\alpha^*_{\setC}(y)\}_{y \in \setR}$ and   $\{\beta^*_{\setC}(\cdot|y)\}_{y \in \setR})$, and we have that
$$ U^{(d)}_{C^k}(y) > 0, \qquad y \in I_k, k=1,...,K,$$
and
$$ U^{(s)}_{C^k}(y) > 0, \qquad y \in I_k, k=1,...K.$$

Furthermore,  by Lemma~\ref{lemma:UdC}~and~\ref{lemma:single_community_consumption}, it follows that the above rate functions $\alpha^*_{C^k}(y)$ and $\beta^*_{C^k}(\cdot|y$ are indeed the optimal rate functions for $y \in I_k$, $k=1,...,K$, i.e. we have that
\begin{align*}
&\alpha^*_{\setC}(y) = 
&\underset{\alpha_{\setC}(y): || \alpha_{\setC}(y) || \leq E_p}{\arg\max} \sum_{C \in \setC}  \alpha_{C}(y) \int_{\setR} \Bsbl Q^*_C(x) p(x|y) - E_q c \Bsbr dx
\end{align*}
and
\begin{align*}
&\beta^*_{\setC}(\cdot|y) = 
&\underset{\beta_{\setC}(\cdot|y): || \beta_{\setC}(\cdot|y)  || \leq E_q}{\arg\max} \sum_{C \in \setC} \int_{\setR} \beta_{C}(x|y) \Bsbl q(x|y) P_C(x) - 2l^*E_p c \Bsbr dx.
\end{align*}
To see this, note the following.
As we have that
$$ U^{(d)}_{C^k}(y) > 0, \qquad y \in I_k, k=1,...,K,$$
it follows from Lemma~\ref{lemma:UdC} that
\begin{eqnarray*}
C^k
&=& \arg\max_{C \in \setC^*}   \int_{\setR} \Bsbl p(x|y)Q^*_C(x) - \beta_C(x) c \Bsbr  dx \\
&=&    \arg\max_{C \in \setC^*} \int_{\setR} \Bsbl p(x|y)Q^*_C(x) - E_q c \Bsbr dx, \qquad y \in I_k, k=1,...,K.
\end{eqnarray*}
From Lemma~\ref{lemma:single_community_consumption}, we then obtain that is it optimal for agents $y \in C^k_d = I^k$ to consume content exclusively in community $C^k$ and set
$$ \alpha_{C^k}(y) = E_p, \qquad y \in C^k_d = I^k.$$
Similarly, as we have that
$$ U^{(s)}_{C^k}(y) > 0, \qquad y \in I_k, k=1,...K,$$
it follows from  Lemma ~\ref{lemma:UsC} that
\begin{eqnarray*}
C^k
&=& \arg\max_{C \in \setC^*}\Big ( q(\xsC(y)|y)P_C(\xsC(y)) - \alpha_C c \Big )\\
&=& \arg\max_{C \in \setC^*}\Big ( q(\xsC(y)|y)P_C(\xsC(y)) - 2l^*E_p c \Big ), \qquad y \in C^k_s = I^k, k =1,...K,
\end{eqnarray*}
where
$$\xsC(y) =  \arg \max_{x \in \setR} q(x|y) P_C(x), \qquad C \in \setC^*, y \in \setR.$$
From Lemma~\ref{lemma:single_community_production} we then obtain  that is it optimal for agents $y \in C^k_s = I^k$ to produce content exclusively in community $C^k$, and set
$$ \beta^*_C(x|y) = \left \{
\begin{array}{rl}
E_q \delta(x^*_{C^k}(y) -x) , & C = C^k \\
0, & \mbox{otherwise,}
\end{array} \right .$$
where by Lemma~\ref{lemma:xs} $x^*_{C^k}(y)$ is the unique solution to the optimization problem
$$ x^*_{C^k}(y) =  \arg \max_{x \in \setR} q(x|y) P_{C^k}(x), \qquad y \in \setR, k=1,...,K.$$

It then follows that the community structure  $(\setC^*, \{\alpha^*_{\setC}(y)_{y \in \setR},  \{\beta^*_{\setC}(\cdot|y)_{y \in \setR}\}$ is a Nash equilibrium.
Furthermore, by construction we have that
$$\cup_{k=1,...,K} I_k = \setR,$$
and  $(\setC^*, \{\alpha^*_{\setC}(y)_{y \in \setR},  \{\beta^*_{\setC}(\cdot|y)_{y \in \setR}\}$ is a covering community structure.

The result of the proposition then follows.

\newpage
\section{Proof for Proposition~\ref{prop:size}}\label{app:size}
In this appendix we prove  Proposition~\ref{prop:size}.

 Let $(\setC^*, \{\alpha^*_{\setC}(y)\}_{y \in \setR},  \{\beta^*_{\setC}(\cdot|y)\}_{y \in \setR})$ be a Nash equilibrium as given by Proposition~\ref{prop:nash}, and
consider a given community $C \in \setC^*$.  We then prove the proposition by characterizing the utility for content consumption $\UsC(y_0)$
for agent
$$y_0 = mid(I_C)$$
in this community. Note that we have that
$$\UsC(y_0) = E_pE_q \int_{z \in I_C} \Bsbl q(\xs(y_0)|y_0)p(\xs(y_0)|z) - c \Bsbr dz.$$

By Proposition~\ref{prop:xs} we have that
$$ \xs(y_0) = y_0 = \arg\max_{x \in \setR} q(x|y_0) P_C(x).$$
It then follows
$$\UsC(y_0) 
= E_pE_q \int_{z \in I_C} \Bsbl g(0)f(||\xs(y_0)-z||) - c \Bsbr dz
=  2 E_pE_q \int_0^{L_C} [g(0)f(x) - c] dx,$$
where $L_C$ is the length of the interval $I_C$.
By Proposition~\ref{prop:nash}  we have that
$$ \UsC(y_0) > 0,$$
and we obtain that
$$\int_0^{L_C} \Bsbl g(0)f(x) - c \Bsbr dx > 0.$$
The result then follows combining the above result with the fact that by Assumption~\ref{ass:fg} the function $f(x)$, $x \in [0,L]$, is strictly decreasing in $x$.

%% file: paper.bbl
\begin{thebibliography}{1}

\bibitem{geo-p}
Anthony Bonato, Jeannette Janssen, and Pawel Pralat.
\newblock Geometric protean graphs.
\newblock {\em CoRR}, abs/1111.0207, 2011.

\bibitem{goel}
Reza Bosagh~Zadeh, Ashish Goel, Kamesh Munagala, and Aneesh Sharma.
\newblock On the precision of social and information networks.
\newblock In {\em Proceedings of the first ACM conference on Online social
  networks}, pages 63--74. ACM, 2013.

\bibitem{game}
Wei Chen, Zhenming Liu, Xiaorui Sun, and Yajun Wang.
\newblock A game-theoretic framework to identify overlapping communities in
  social networks.
\newblock {\em Data Mining and Knowledge Discovery}, 21(2):224--240, 2010.

\bibitem{survey_clustering}
Santo Fortunato.
\newblock Community detection in graphs.
\newblock {\em Physics Reports}, 486(3-5):75 -- 174, 2010.

\bibitem{filtering}
Mangesh Gupte, MohammadTaghi Hajiaghayi, Lu~Han, Liviu Iftode, Pravin Shankar,
  and Raluca~M Ursu.
\newblock News posting by strategic users in a social network.
\newblock In {\em Internet and Network Economics}, pages 632--639. Springer,
  2009.

\bibitem{hegde}
Nidhi Hegde, Laurent Massouli{\'e}, and Laurent Viennot.
\newblock Self-organizing flows in social networks.
\newblock In {\em Structural Information and Communication Complexity}, pages
  116--128. Springer, 2013.

\bibitem{Kronecker}
Jure Leskovec, Deepayan Chakrabarti, Jon Kleinberg, Christos Faloutsos, and
  Zoubin Ghahramani.
\newblock Kronecker graphs: An approach to modeling networks.
\newblock {\em J. Mach. Learn. Res.}, 11:985--1042, March 2010.

\bibitem{rudin}
Walter Rudin.
\newblock {\em Principles of Mathematical Analysis}.
\newblock McGraw-Hill, 1976.

\end{thebibliography}
